\documentclass[journal]{IEEEtran}
\usepackage{}
\usepackage{amsfonts}
\usepackage{bm}
\usepackage{makecell}
\usepackage{savesym}
\usepackage{amsmath}
\usepackage{cases}
\usepackage{xcolor}
\savesymbol{iint}
\restoresymbol{TXF}{iint}
\usepackage{amssymb}
\usepackage{booktabs}
\usepackage{latexsym}
\usepackage{psfrag}
\usepackage{graphicx,subfigure}
\DeclareMathOperator*{\bigtimes}{\text{\large $\times$}}
\newcommand{\tabincell}[2]{\begin{tabular}{@{}#1@{}}#2\end{tabular}}
\usepackage{hyperref}
\usepackage{color}
\usepackage{amsfonts,mathrsfs,bbm,amsmath,stmaryrd,amssymb,pifont}%
\usepackage{amscd,graphicx,array,dsfont,texdraw,tikz}%
\usetikzlibrary{arrows,decorations.markings,shapes,shadows,fadings}
\usepackage{bbm,mathbbold,mathrsfs}%
\usepackage{array}
\usepackage{epstopdf}
\allowdisplaybreaks[4]

\usepackage[linesnumbered,ruled,vlined]{algorithm2e}
\newtheorem{theorem}{Theorem}
\newtheorem{lemma}{Lemma}
\newtheorem{fact}{Fact}
\newtheorem{remark}{Remark}
\newtheorem{assumption}{Assumption}
\newtheorem{definition}{Definition}
\newtheorem{proposition}{Proposition}
\newtheorem{corollary}{Corollary}

\begin{document}

\title{Distributed Neural Policy Gradient Algorithm for Global Convergence of Networked Multi-Agent Reinforcement Learning}
\author{Pengcheng Dai,~Yuanqiu Mo,~Wenwu Yu,~\IEEEmembership{Senior Member,~IEEE},~and~Wei Ren,~\IEEEmembership{Fellow,~IEEE}
\thanks{This work was supported in part by the National
Science and Technology Major Project of China under Grant 2022ZD0120001; in part by the National Natural Science Foundation of China under Grant 62233004 and Grant 62303112;  in part by the Jiangsu Provincial Scientific Research Center of Applied Mathematics under Grant BK20233002; and in part by the Natural Science Foundation of Jiangsu Province under Grant BK20230826. \emph{(Corresponding author: Wenwu Yu.)}}
\thanks{P. Dai and Y. Mo are with School of Mathematics, Southeast University, Nanjing 210096, China
(e-mail: Jldaipc@163.com; yuanqiumo@seu.edu.cn).
%W. Yu is also with the School of Automation, Southeast University, Nanjing 210096, China.
}
\thanks{W. Yu is with the School of Mathematics, Frontiers Science Center for Mobile Information Communication and Security, Southeast University, Nanjing 210096, China and also with the Purple Mountain Laboratories, Nanjing 211111, China (e-mail: wwyu@seu.edu.cn).}
\thanks{W. Ren is with the Department of Electrical and Computer
Engineering, University of California, Riverside, CA 92521, USA. (ren@ece.ucr.edu).}
}

\markboth{Manuscript For Review}
{Manuscript For Review}
%\markboth{Manuscript For Review}
%{Manuscript For Review}
%\markboth{IEEE TRANSACTIONS ON INDUSTRIAL INFORMATICS}
%{IEEE TRANSACTIONS ON INDUSTRIAL INFORMATICS}

\maketitle

\begin{abstract}
% describe the MARL problem
This paper studies the networked multi-agent reinforcement learning (NMARL) problem, where the objective of agents is
%efficiently communicate with neighbors and collaboratively optimize their execution policies of actions
to collaboratively maximize the discounted average cumulative rewards.
% different to the exist work (neural network)
%Different from the most existing methods that suffer from poor expression due to linear function approximation, we propose a distributed neural policy gradient algorithm with novel two-layer neural network structures designed in critic and actor steps instead.
Different from the existing methods that suffer from poor expression due to linear function approximation, we propose a distributed neural policy gradient algorithm that features two innovatively designed neural networks, specifically for the approximate $Q$-functions and policy functions of agents.
% introduce the distributed way
%In the proposed algorithm, each agent achieves collaborative policy estimation only by using the critic network parameters of its neighbors in critic step, and optimizes its local actor network parameter in actor step.
This distributed neural policy gradient algorithm consists of two key components: the distributed critic step and the decentralized actor step.
In the distributed critic step, agents receive the approximate $Q$-function parameters from their neighboring agents via a time-varying communication networks to collaboratively evaluate the joint policy.
In contrast, in the decentralized actor step, each agent updates its local policy parameter solely based on its own approximate $Q$-function.
%In the convergence analysis, we first prove the global convergence of agents on the joint policy evaluation in the distributed critic and then prove the global convergence of the overall distributed neural policy gradient algorithm on objective function.
In the convergence analysis, we first establish the global convergence of agents for the joint policy evaluation in the distributed critic step.
Subsequently, we rigorously demonstrate the global convergence of the overall distributed neural policy gradient algorithm with respect to the objective function.
%Finally, the effectiveness of the proposed algorithm are demonstrated by comparing it with a centralized algorithm through simulation in the robot path planning environment.
Finally, the effectiveness of the proposed algorithm is demonstrated by comparing it with a centralized algorithm through simulation in the robot path planning environment.
%To the best of our knowledge, the work appears to be the first study of the distributed algorithm with neural network in the NMARL problem.
%The work seems to be the pioneering study on the application of neural networks in distributed algorithms for NMARL problem,
\end{abstract}

\IEEEpeerreviewmaketitle

\section{Introduction}
% introduction of RL+MARL
%Recently, reinforcement learning (RL)~\cite{Sutton1998} has attracted increasing attention in academia and industry.
Recently, reinforcement learning (RL)~\cite{Sutton1998} has garnered growing attention in both academic and industrial domains.
%As an extension of RL with single agent, multi-agent reinforcement learning (MARL) has achieved remarkable performances in many complex scenarios,
As an extension of single-agent RL (SARL), multi-agent RL (MARL) has demonstrated remarkable performance in various complex scenarios, such as smart grid~\cite{Dai2020TII}-\cite{Dai2021TCYB}, intelligent transportation~\cite{Chu2020TITS}-\cite{Wang2021TCYB},  cyber-physical systems~\cite{Ding2017Automatica}-\cite{Dai2020TNSE}, and wireless communication~\cite{Nasir2019AreasCommun}-\cite{Zhao2019TWC}, etc.
\par
% centralized approach
To address the MARL problem, a commonly used method is to conceptualize it as SARL problem.
However, the centralized algorithm in \cite{Wei2018ICKDD} necessitates the use of global state-action pair, which limits its scalability in large-scale scenarios.
%However, the centralized algorithm in~\cite{Wei2018ICKDD} requires the global state-action pair, which compromises scalability and hinders generalization in large-scale MARL problems.
% introduce the IQL method
%To tackle the scalability challenge posed by the centralized methods, a viable solution is to treat each agent as an autonomous individual for independent learning, as proposed in~\cite{Tan1993}.
To overcome the scalability challenge associated with centralized method, a feasible solution is to model each agent as an autonomous individual for independent learning, as proposed in~\cite{Tan1993}.
%Although the independent learning approach effectively addresses the issue of scalability, it encounters a new challenge referred to as ``non-stationary environment''.
%In this scenario, each agent engaging in independent learning perceives other agents as part of the dynamic environment.
%As the update of the policies of agents, the environment experienced by each agent will constantly change.
%On the other hand,
Although the independent learning approach~\cite{Tan1993} successfully addresses scalability issue, it introduces a new challenge known as the ``non-stationary environment''.
In this context, each agent conducting independent learning views other agents as components of a dynamic environment.
As the policies of agents are updated, the environment perceived by each agent continuously evolves.
Moreover, it primarily focuses on the local optimization of agents and fails to ensure the cooperative optimization objective in the MARL problems.
\par
% introduce the value-based methods
In order to tackle the challenges of scalability and non-stationarity, several value-based algorithms within the centralized training with decentralized execution (CTDE) framework have been proposed for solving the MARL problem with discrete and finite state-action spaces~\cite{QMIX}-\cite{QPLEX}.
In CTDE framework, agents establish local $Q$-functions for executing their actions while updating their $Q$-function parameters via a centralized controller.
% policy-based methods
Given the stability and convergence challenges faced by the CTDE framework when addressing MARL problem with large or continuous state-action spaces, several actor-critic (AC)-based algorithms utilizing the centralized critic with decentralized actor (CCDA) framework have been further developed in~\cite{COMA}-\cite{DOP}.
%% policy-based methods
%Considering the stability and convergence challenges faced by the CTDE framework in dealing with the MARL problems with large or continuous state-action spaces, several actor-critic (AC)-based algorithms have been further proposed by using the centralized critic with decentralized actor (CCDA) framework in~\cite{COMA}-\cite{DOP}.
\par
% introduce the difficult in theory
%In stark contrast to the tremendous empirical successes of the MARL algorithms~\cite{QMIX}-\cite{DOP} in virtual scenes, their theoretical understanding, particularly regarding neural networks, remains limited.
In sharp contrast to the remarkable empirical successes of MARL algorithms~\cite{QMIX}-\cite{DOP} in virtual environments, the theoretical comprehension of these algorithms, especially concerning neural networks, remains relatively constrained.
% introduce the difficult in AC-based methods
%Specifically, in AC-based methods~\cite{COMA}-\cite{DOP},
%the alternating contribution of the critic step and actor step aims to maximize the objective function.
%However, it is possible for the critic step to converge towards an undesirable stationary point, leading to bias and even divergence of the estimated policy gradient during the actor step.
%This uncertainty impedes global convergence of AC-based methods.
Specifically, in AC-based methods~\cite{COMA}-\cite{DOP}, the alternating contributions of critic step and actor step are designed to optimize the objective function.
However, critic step may converge to an undesirable stationary point, which can introduce bias and potentially cause divergence in estimated policy gradient during actor step.
This uncertainty hinders the global convergence of AC-based methods.
\par
% introduce the convergence of linear approximation in MARL
%Although the theoretical analysis of AC-based algorithms poses challenges, considerable efforts have been made to establish the convergence of AC-based algorithms with linear approximation in the networked MARL (NMARL) field, where agents can exchange information with its neighboring agents through a communication network.
Although the theoretical analysis of AC-based algorithms presents significant challenges, substantial efforts have been devoted to establishing the convergence of AC-based algorithms with linear approximation in networked MARL (NMARL) domain, where agents are capable of exchanging information with their neighboring agents via a communication network.
%For example, a fully decentralized AC algorithm based on linear approximation has been proposed to tackle the NMARL problem over undirected communication networks, offering a rigorous guarantee of convergence in~\cite{Zhang2018ICML}.
%% introduce more general AC method in MARL
%For the fixed and time-varying directed communication networks in more general settings, two {\color{blue}distinct} distributed {\color{blue}linear approximation} AC algorithms {\color{blue}with
%provable convergence guarantees} were respectively proposed in~\cite{Dai2022TNNLS}.
For instance, a fully decentralized AC algorithm based on linear approximation was introduced to address the NMARL problem over undirected communication networks, providing a rigorous convergence guarantee in~\cite{Zhang2018ICML}.
% introduce more general AC method in MARL
For more general settings involving fixed and time-varying directed communication networks, two distinct distributed linear approximation AC algorithms with provable convergence guarantees were separately proposed in~\cite{Dai2022TNNLS}.
\par
% introduce the convergence in SARL
Owing to the overly ideal assumption of linear approximation and its limitations in expressive capability in~\cite{Zhang2018ICML}-\cite{Dai2022TNNLS}, an AC-based algorithm utilizing neural network approximation was proposed for the SARL problem in~\cite{Wang2019}.
%To address the overly ideal assumption of linear approximation and its limitation in expressive power as discussed in~\cite{Zhang2018ICML}-\cite{Dai2022TNNLS}, an AC-based algorithm utilizing neural network approximation was introduced for solving the SARL problem in~\cite{Wang2019}.
Unfortunately, the neural network proposed in~\cite{Wang2019} is specifically designed for the SARL problem and cannot be directly adapted to the NMARL problem.
Moreover, in the NMARL problem, the decisions made by agents are intricately interdependent, and communication among them may be subject to uncertainty.
Consequently, designing a neural network-based approximation algorithm with provable convergence guarantees for the NMARL problem remains a significant challenge.
%Unfortunately, the neural network in~\cite{Wang2019} is specially proposed for the SARL problem and cannot be directly extended to MARL problem.
%Additionally, in the NMARL problem, the decisions made by agents are intricately interconnected and communication among them can be subject to uncertainty.
%Hence, it is challenging to design a neural network approximation algorithm with provable convergence guarantee for the NMARL problem.
\begin{table}[h!]
\centering
\caption{Overview of existing AC-based algorithms}
\begin{tabular}{|@{\hspace{0.3em}}c@{\hspace{0.3em}}|@{\hspace{0.3em}}c@{\hspace{0.3em}}%
                |@{\hspace{0.3em}}c@{\hspace{0.3em}}|c@{\hspace{0.3em}}|c@{\hspace{0.3em}}|}
  \hline
  %after \\: \hline or \cline{col1-col2} \cline{col3-col4} ...
  References & \tabincell{c}{Multi-agent\\ RL} & Distributed & \tabincell{c}{Function\\ approximation} & \tabincell{c}{Convergence\\ analysis}\\
  \hline\hline
  [14]-[15] & $\surd$ & $\bigtimes$  & Neural network & $\bigtimes$\\
  $\text{[}$16]-[17]& $\surd$ & $\surd$ & Linear & $\surd$\\
  $\text{[}$18] & $\bigtimes$ & $\bigtimes$ & Neural network & $\surd$\\
  This work & $\surd$ & $\surd$ &  Neural network & $\surd$\\
  \hline
\end{tabular}
\smallskip\\
\footnotesize{$\surd$ means this feature is involved}\;\;
\footnotesize{\hspace*{1.2em}$\bigtimes$ means this feature is uninvolved}\label{Tab1}
\vspace{-2ex}
\end{table}
\par
%The above challenges highlight the difficulty of designing provable neural network approximation algorithms for MARL problems.
%The motivation of this paper is to address the aforementioned challenges by developing a distributed neural policy gradient algorithm with global convergence for the NMARL problem.
%The main contributions of our work are presented in Table~\ref{Tab1}, which are detailed as follows.
The primary motivation of this paper is to tackle the aforementioned challenges by developing a distributed neural policy gradient algorithm with global convergence for the NMARL problem. The key contributions of our work are summarized in Table~\ref{Tab1} and elaborated as follows.
\par
(1) To address the poor expression of the existing methods with linear function approximation, two innovative neural networks are designed for the approximate $Q$-functions and policy functions of agents, respectively.
%On this basis, we propose a distributed neural policy gradient algorithm that consists of two components: the distributed critic step and the decentralized actor step.
%Specially, in the distributed critic step,
%agents receive their neighboring agents' approximate function parameters through the time-varying communication networks to achieve the collaborative evaluation of the joint policy.
%Subsequently, in the decentralized actor step, agents solely update their local policy parameters based on their respective approximate functions.
On this foundation, we propose a distributed neural policy gradient algorithm that consists two key components: the distributed critic step and the decentralized actor step.
Specifically, in the distributed critic step, agents receive the approximate $Q$-function parameters of their neighboring agents via a time-varying communication networks to collaboratively evaluate the joint policy.
In contrast, in the decentralized actor step, agents independently update their local policy parameters based on their respective approximate $Q$-functions.
%\par
%(1) A distributed neural policy gradient algorithm is designed for the NMARL problem, where each agent relies solely on its local information and information from neighboring agents.
%Specially, a novel two-layer neural network architecture is designed for the critic network and the actor network in distributed neural policy gradient algorithm to ensure the global convergence of the algorithm.
%To the best of our knowledge, the work appears to be the first study of the distributed neural algorithm with provable global optimality and convergence guarantee in NMARL problems.
\par
(2) In the convergence analysis, we first establish the global convergence of agents for the joint policy evaluation in the distributed critic step.
Subsequently, we rigorously demonstrate the global convergence
of the overall distributed neural policy gradient algorithm with
respect to the objective function.
%In order to establish the convergence analysis of the distributed neural policy gradient algorithm, we first prove the global convergence of agents on the joint policy evaluation in the distributed critic step and then prove the global convergence of the overall algorithm on the objective function.
\par
(3)
%The distributed neural policy gradient algorithm is tested in the simulation experiments on robot swarms in different path networks.
The simulation results in the robot path planning environment demonstrate that the performance of our algorithm, as measured by the objective function, closely approximates that of the centralized algorithm that serves as a benchmark for globally optimal performance.
This effectively validates the global convergence of our proposed algorithm.
%The simulation results on the robot path planning environment indicate that the performance of our algorithm, as quantified by the objective function, closely approximates that of the centralized algorithm, which serves as a baseline for globally optimal performance.
%It effectively demonstrates the global convergence of our proposed algorithm.
\par
The rest of this paper is organized as follows.
Section~\ref{SectionIIpreliminaries} presents the notations and foundational knowledge related to network.
The model and relevant knowledge concerning the NMARL problem are described in Section~\ref{SectionProblemformulation}.
A distributed neural policy gradient algorithm is proposed in Section~\ref{DistributedNeuralpolicygradientalgorithm}.
%is structured as follows.
%Section~\ref{SectionIIpreliminaries} introduces the notations and preliminary knowledge of network.
%The model and the related knowledge of the NMARL problem are introduced in Section~\ref{SectionProblemformulation}.
%Section~\ref{DistributedNeuralpolicygradientalgorithm} proposes a distributed neural policy gradient algorithm.
%The global convergence of the distributed critic step is
The global convergence of agents for the joint policy evaluation in the distributed critic step
is established in Section~\ref{SectionIV}.
Section~\ref{theconvergenceofoverallalgorithm} proves the global convergence of the overall algorithm.
%Section~\ref{theconvergenceofoverallalgorithm} presents the global convergence of the overall algorithm.
Some simulation results of our algorithm are provided in Section~\ref{SectionSimulations}.
Finally, Section~\ref{SectionVConclusions} presents the conclusion and discusses further work.
\section{Preliminaries}\label{SectionIIpreliminaries}
\subsection{Notations}
%In this paper, $\mathbb{R}$ and $\mathbb{R}^{n}$ represent the set of real-valued and $n$-dimensional real-valued vectors, respectively.
In this paper, $\mathbb{R}$ and $\mathbb{R}^{n}$ denote the set of real numbers and the set of $n$-dimensional real-valued vectors, respectively.
$\lfloor m\rfloor$ represents the maximum integer not greater than $m$.
For a vector $\phi\in\mathbb{R}^{n}$, $\|\phi\|_{2}$ is denoted as the $\mathcal{L}_{2}$-norm of vector $\phi$.
%$\mathds{1}\{\cdot\}$ represents the indicator function and $\mathbf{1}_{N}$ represents an $N$-dimensional vector with all elements equal to 1.
$\mathds{1}\{\cdot\}$ denotes the indicator function, while $\mathbf{1}_{N}$ signifies an $N$-dimensional vector where all elements are equal to 1.
%For set $S_{B}\subseteq\mathbb{R}^{N}$,
$P_{S_{B}}(\cdot)$ is a projection operator on space $S_{B}$.
%$\mathbb{E}_{\mathrm{init}}[\cdot]$ represents the expectation with respect to random initialization and $\mathbb{E}[\cdot]$ denotes the expectations with respect to all the randomness.
$\mathbb{E}_{\mathrm{init}}[\cdot]$ represents the expectation with respect to random initialization, while $\mathbb{E}[\cdot]$ denotes the expectation with respect to all sources of randomness.
\subsection{Network}
%$\mathcal{G}_{t}\big(\mathcal{N},\mathcal{E}_{t}\big)$ represents the time-varying directed communication networks, where $\mathcal{N}=\{1,\cdots,N\}$ is the agent set and $\mathcal{E}_{t}\subseteq\mathcal{N}\times\mathcal{N}$ is the edge set at time $t$.
$\mathcal{G}_{t}(\mathcal{N},\mathcal{E}_{t})$ denotes a directed network, where $\mathcal{N}=\{1,\cdots,N\}$ represents the set of agents and $\mathcal{E}_{t}\subseteq\mathcal{N}\times\mathcal{N}$ signifies the edge set at time $t$.
%An edge $e_{ij}(t)\in\mathcal{E}_{t}$ means that agent $i$ can receive the information from its neighboring agent $j$ at time $t$.
An edge $e_{ij}(t) \in \mathcal{E}_t$ indicates that agent $i$ is capable of receiving information from its neighboring agent $j$ at time $t$.
In $\mathcal{G}_{t}$, define $\mathcal{N}_{i}(t)=\{j|e_{ij}(t)\in\mathcal{E}_{t},\forall j\in\mathcal{N}\}$ as the neighborhood of agent $i$.
We call that agent $i_{1}$ has a directed path to agent $i_{k}$ in $\mathcal{G}(t)$, if there exists an ordered edges $i_{1}\rightarrow i_{2}\rightarrow \cdots\rightarrow i_{k}$ with $e_{i_{j+1},i_{j}}(t)\in\mathcal{E}(t)$ for all $j=1,\cdots,k-1$.
$\mathcal{G}_{t}$ is strongly connected if there exists at least one directed path between any two distinct agents.
Let $A(t) = [a_{ij}(t)]_{N \times N}$ denote the weight matrix of $\mathcal{G}_t$, where $a_{ij}(t) > 0$ if $j \in \mathcal{N}_i(t)$, and $a_{ij}(t) = 0$ otherwise.
The time-varying networks $\{\mathcal{G}_{t}(\mathcal{N},\mathcal{E}_{t})\}_{t\geq0}$ is said to be uniformly strongly
connected if there exists an integer $D$ such that $\big(\mathcal{N},\bigcup_{t=kD}^{(k+1)D-1}\mathcal{E}_{t}\big)$ is
strongly connected for any $k\geq0$.
%Let $A(t) = [a_{ij}(t)]_{N \times N}$ denote the weight matrix of $\mathcal{G}_t$, where $a_{ij}(t) > 0$ if $j \in \mathcal{N}_i(t)$, and $a_{ij}(t) = 0$ otherwise.
%Define $A(t)=[a_{ij}(t)]_{N\times N}$ as the weight matrix of $\mathcal{G}_{t}$ and satisfy $a_{ij}(t)>0$ if $j\in\mathcal{N}_{i}(t)$ and $a_{ij}(t)=0$ otherwise.
\section{NMARL problem}\label{SectionProblemformulation}
%This section {\color{blue}introduces the model and the related knowledge of the NMARL problem.}
%This section introduces the NMARL problem and discusses the fundamental concepts and challenges associated with employing function approximation for solving the NMARL problem.
\subsection{Description of NMARL problem}\label{describednetworkMARL}
NMARL problem can be described as $\big(\mathcal{G}_{t}(\mathcal{N},\mathcal{E}_{t}),$ $\{\mathcal{S}_{i}\}_{i\in\mathcal{N}},\{\mathcal{A}_{i}\}_{i\in\mathcal{N}},\mathcal{P},\{r_{i}\}_{i\in\mathcal{N}},\zeta,\gamma\big)$, where the detailed description of each element is summarized as follows.
%where $\mathcal{N}=\{1,$ $\cdots,N\}$ represents the set of agents.
\par
\textbf{State and action}:
$\mathcal{S}_{i}$ and $\mathcal{A}_{i}$ represent the local state space and the local action space of agent $i\in\mathcal{N}$, respectively.
$s_{i}\in\mathcal{S}_{i}$ and $a_{i}\in\mathcal{A}_{i}$ denote the local state and local action of agent $i$, respectively.
Define the global state as $\bm{s} = (s_1, \cdots, s_N) \in \mathcal{S} = \mathcal{S}_1 \times \cdots \times \mathcal{S}_N$, and analogously define the global action as $\bm{a} = (a_1, \cdots, a_N) \in \mathcal{A} = \mathcal{A}_1 \times \cdots \times \mathcal{A}_N$.
$\bm{s}_t = (s_{1,t}, \cdots, s_{N,t})$ and $\bm{a}_t = (a_{1,t}, \cdots, a_{N,t})$ denote the global state and global action at time $t$, respectively.
%Define the global state as $\bm{s}=(s_{1},\cdots,s_{N})\in\mathcal{S}=\mathcal{S}_{1}\times\cdots\times\mathcal{S}_{N}$ and similarly the global action as $\bm{a}=(a_{1},\cdots,a_{N})\in\mathcal{A}=\mathcal{A}_{1}\times\cdots\times\mathcal{A}_{N}$.
%Moreover, $\bm{s}_{t}=(s_{1,t},\cdots,s_{N,t})$ and $\bm{a}_{t}=(a_{1,t},\cdots,a_{N,t})$ represent the global state and global action at time $t$, respectively.
\par
\textbf{State transition probability function}:
$\mathcal{P}(\bm{s}'|\bm{s},\bm{a}):\mathcal{S}\times\mathcal{A}\times\mathcal{S}\rightarrow[0,1]$ is the state transition probability function.
\par
\textbf{Reward function}:
$r_{i}(\bm{s},\bm{a}):\mathcal{S}\times\mathcal{A}\rightarrow\mathbb{R}$ denotes the local reward function for agent $i$.
Define $\bar{r}(\bm{s},\bm{a})=(1/N)\sum_{i=1}^{N}r_{i}(\bm{s},\bm{a})$ as the average reward function.
%Specially, define
%$r_{i,t}=r_{i}(\bm{s}_{t},\bm{a}_{t})$ as the instantaneous reward of agent $i$ at time $t$ and $\bar{r}_{t}=(1/N)\sum_{i=1}^{N}r_{i,t}$ as the average reward at time $t$.
Specifically, define $r_{i,t} = r_{i}(\bm{s}_{t}, \bm{a}_{t})$ as the instantaneous reward for agent $i$ at time $t$, and $\bar{r}_{t} = (1/N) \sum_{i=1}^{N} r_{i,t}$ as the average reward across all agents at time $t$.
Similar to the NMARL problems discussed in~\cite{Zhang2018ICML}-\cite{Dai2022TNNLS} and~\cite{Doan2019ICML}-\cite{Suttle2020IFAC}, we assume that each agent can access the global state-action, while the reward function remains private and is solely utilized by the respective agent.
\par
\textbf{Initial state distribution}: $\zeta$ is the distribution of the initial state $\bm{s}_{0}$
\par
\textbf{Discount factor}: $0<\gamma<1$ is the discount factor.
\par
Define the local policy of each agent $i$ as  $\pi_{\theta_{i}}=\pi_{i}(\cdot|\bm{s};\theta_{i})$, where $\theta_{i}$ represents the local policy parameter of agent $i$.
The joint policy of all agents is defined as $\bm{\pi}_{\bm{\theta}}(\cdot|\bm{s})=\prod_{i=1}^{N}\pi_{i}(\cdot|\bm{s};\theta_{i})$, where $\bm{\theta}=(\theta^{\top}_{1},\cdots,\theta^{\top}_{N})^{\top}$ is the joint policy parameter.
%\par
%In the NMARL problem, define $J(\bm{\pi_{\theta}})$ as
%the discounted average cumulative reward under joint policy $\bm{\pi_{\theta}}$, which is given by
In the NMARL problem, define $J(\bm{\pi_{\theta}})$ as the discounted average cumulative rewards under the joint policy $\bm{\pi_{\theta}}$, which is formulated as
\begin{align}\label{thedefinelongtermreturninMARL}
J(\bm{\pi_{\theta}})=(1-\gamma)\mathbb{E}_{\bm{s}\sim\zeta}\Big[\sum_{t=0}^{\infty}\gamma^{t}\bar{r}_{t}\Big|\bm{s}_{0}=\bm{s},\bm{a}_{t}\!\sim\!\bm{\pi}_{\bm{\theta}}(\cdot|\bm{s}_{t})\Big].
\end{align}
%\begin{align}\label{thelongtermreturninMARL}
%\max_{\bm{\pi}}J(\bm{\pi})=(1-\gamma)\mathbb{E}\Big[\sum_{t=0}^{\infty}\gamma^{t}r_{\mathrm{ave},t}|\bm{a}_{t}\sim\bm{\pi}(\cdot|s_{t})\Big],
%\end{align}
The objective of agents is to find the optimal policy to maximize $J(\bm{\pi_{\theta}})$, i.e.,
$\max_{\bm{\pi_{\theta}}}J(\bm{\pi_{\theta}})$.
\subsection{Policy gradient in NMARL problem}
In the NMARL problem, for any joint policy $\bm{\pi_{\theta}}$, define the global $Q$-function as
%$Q^{\bm{\pi_{\theta}}}(\bm{s},\bm{a})$ as
\begin{align}\label{theglobalQfunction}
Q^{\bm{\pi}_{\bm{\theta}}}(\bm{s},\bm{a})=(1-\gamma)\mathbb{E}_{\bm{\pi_{\theta}}}\Big[\sum_{t=0}^{\infty}\gamma^{t}\bar{r}_{t}\Big|\bm{s}_{0}=\bm{s},\bm{a}_{0}=\bm{a}\Big] \end{align}
and the Bellman evaluation operator $\mathcal{T}^{\bm{\pi}_{\bm{\theta}}}(\cdot)$ as
\begin{align}\label{theBellmanevaluationoperatorofpiTheta}
\big[\mathcal{T}^{\bm{\pi}_{\bm{\theta}}}(Q)\big](\bm{s},\bm{a})=&\mathbb{E}_{\bm{\pi_{\theta}}}\big[(1-\gamma)\bar{r}(\bm{s},\bm{a})+\gamma Q(\bm{s}',\bm{a}')\big|\notag\\
&\bm{s}'\sim\mathcal{P}(\cdot|\bm{s},\bm{a}),\bm{a}'\sim\bm{\pi_{\theta}}(\cdot|\bm{s}')\big],
\end{align}
where $Q\in\mathbb{R}^{|\mathcal{S}||\mathcal{A}|}$.
%where $r_{\mathrm{ave}}(s,\bm{a})=(1/N)\sum_{i=1}^{N}r_{i}(s,\bm{a})$ is the average reward of all agents.
 According to the definition of $Q^{\bm{\pi}_{\bm{\theta}}}(\bm{s},\bm{a})$ in (\ref{theglobalQfunction}), we have that $Q^{\bm{\pi}_{\bm{\theta}}}=[Q^{\bm{\pi}_{\bm{\theta}}}(\bm{s},\bm{a})]_{(\bm{s},\bm{a})\in\mathcal{S}\times\mathcal{A}}$ is the fixed point of $\mathcal{T}^{\bm{\pi}_{\bm{\theta}}}(\cdot)$, i.e.,
\begin{align}\label{thefixedpoint}
Q^{\bm{\pi}_{\bm{\theta}}}(\bm{s},\bm{a})=\big[\mathcal{T}^{\bm{\pi}_{\bm{\theta}}}(Q^{\bm{\pi}_{\bm{\theta}}})\big](\bm{s},\bm{a}),\forall(\bm{s},\bm{a})\in\mathcal{S}\times\mathcal{A}.
\end{align}
%By the definitions of (\ref{thestatevaluefunctioninSARL}) and (\ref{theactionvaluefunctioninSARL}),
%$V^{\bm{\pi}}(s)$ and $Q^{\bm{\pi}}(s,\bm{a})$ are related by $V^{\bm{\pi}}(s)=\sum_{\bm{a}\in\mathcal{A}}\bm{\pi}(\bm{a}|s)Q^{\bm{\pi}}(s,\bm{a})$ for all $s\in\mathcal{S}$.
%Consider that in MARL, the joint policy $\bm{\pi}$ together with the state transition probability $\mathcal{P}$ will induce a Markov chain over state space $\mathcal{S}$, we denote
%Denote $\varrho_{\bm{\pi}}(s):\mathcal{S}\rightarrow[0,1]$ and $\varsigma_{\bm{\pi}}(s,\bm{a}):\mathcal{S}\times\mathcal{A}\rightarrow[0,1]$ as the stationary state distribution for all $s\in\mathcal{S}$ and the stationary state-action distribution for all $(s,\bm{a})\in\mathcal{S}\times\mathcal{A}$, respectively.
%By the definitions of $\varrho_{\bm{\pi}}(s)$ and $\varsigma_{\bm{\pi}}(s,\bm{a})$, we can have $\varsigma_{\bm{\pi}}(s,\bm{a})=\varrho_{\bm{\pi}}(s)\bm{\bm{\pi}}(\bm{a}|s)$.
%For joint policy $\bm{\pi_{\theta}}$,
Denote $\nu_{\bm{\pi_{\theta}}}(\bm{s})$ as the state visitation measure of $\bm{s}$ under $\bm{\pi_{\theta}}$ and initial
state distribution $\zeta$, which is described as
\begin{align}
\nu_{\bm{\pi_{\theta}}}(\bm{s})=(1-\gamma)\mathbb{E}_{\bm{s}_{0}\sim\zeta,\bm{\pi_{\theta}}}\Big[\sum_{t=0}^{\infty}\gamma^{t}\mathbb{P}(\bm{s}_{t}=\bm{s})\Big|\bm{s}_{0}\Big],\label{thestatevisitationmeasureinSARL}
%\sigma_{\bm{\pi}}(s,\bm{a})=(1-\gamma)\mathbb{E}\Big[\sum_{t=0}^{\infty}\gamma^{t}\mathbb{P}(s_{t}=s,\bm{a}_{t}=\bm{a})\Big].\label{thestat-actionevisitationmeasureinSARL}
\end{align}
where $\mathbb{P}(\bm{s}_{t}=\bm{s})$ represents the probability of $\bm{s}_{t}=\bm{s}$ occurring at time $t$.
Moreover, denote $\sigma_{\bm{\pi_{\theta}}}(\bm{s},\bm{a})$ as the state-action visitation measure of $(\bm{s},\bm{a})$ and satisfy $\sigma_{\bm{\pi_{\theta}}}(\bm{s},\bm{a})=\nu_{\bm{\pi_{\theta}}}(\bm{s})\bm{\pi_{\theta}}(\bm{a}|\bm{s})$.
\par
Based on the above definitions, the policy gradient theorem~\cite{Sutton2000} of the NMARL problem is described as follows.
\begin{lemma}\label{thelemmaofpolicygradienttheorem}
For any joint policy $\bm{\pi_{\theta}}$, the policy gradient of $J(\bm{\pi}_{\bm{\theta}})$ with respect to $\theta_{i}$ is represented as
\begin{align}\label{thepolicygradienttheorem}
\nabla_{\theta_{i}}J(\bm{\pi}_{\bm{\theta}})=\mathbb{E}_{\sigma_{\bm{\pi}_{\bm{\theta}}}}\big[Q^{\bm{\pi}_{\bm{\theta}}}(\bm{s},\bm{a})\nabla_{\theta_{i}}\log\pi_{i}(a_{i}|\bm{s};\theta_{i})\big].
\end{align}
\end{lemma}
\par
Lemma~\ref{thelemmaofpolicygradienttheorem} provides the policy gradient for agent $i$.
Nevertheless, the calculation of  $Q^{\bm{\pi}_{\bm{\theta}}}(\bm{s},\bm{a})$ poses substantial challenges, especially in large-scale scenarios.
To tackle this issue, the prevalent approach is to approximate $Q^{\bm{\pi}_{\bm{\theta}}}(\bm{s},\bm{a})$.
%Lemma~\ref{thelemmaofpolicygradienttheorem} establishes
%the policy gradient of agent $i$.
%However, the calculation of $Q^{\bm{\pi}_{\bm{\theta}}}(\bm{s},\bm{a})$ presents significant challenges, particularly when dealing with large-scale scenarios.
%In order to address this issue, the most common method is to approximate $Q^{\bm{\pi}_{\bm{\theta}}}(\bm{s},\bm{a})$.
\subsection{Approximate $Q$-function in NMARL problem}
%For convenience of expression in the subsequent, we represent the global state-action $(\bm{s},\bm{a})$ by $\bm{z}$.
%In contrast to the existing linear function approximation methods,
%a novel neural network (i.e., approximate function) is designed in Fig.~\ref{neuralnetwork} to approximate the global $Q$-function $Q^{\bm{\pi_{\theta}}}(\bm{z})$ in (\ref{theglobalQfunction}).
For the sake of concise expression in the subsequent discussion, we denote the global state-action pair $(\bm{s},\bm{a})$ by $\bm{z}$. Unlike existing linear function approximation methods, a novel neural network architecture (serving as the approximate function) is designed, as illustrated in Fig.~\ref{neuralnetwork}, to approximate the global $Q$-function $Q^{\bm{\pi_{\theta}}}(\bm{z})$ in (\ref{theglobalQfunction}).
\begin{figure}[!htb]
\centering
\includegraphics[width=0.7\hsize]{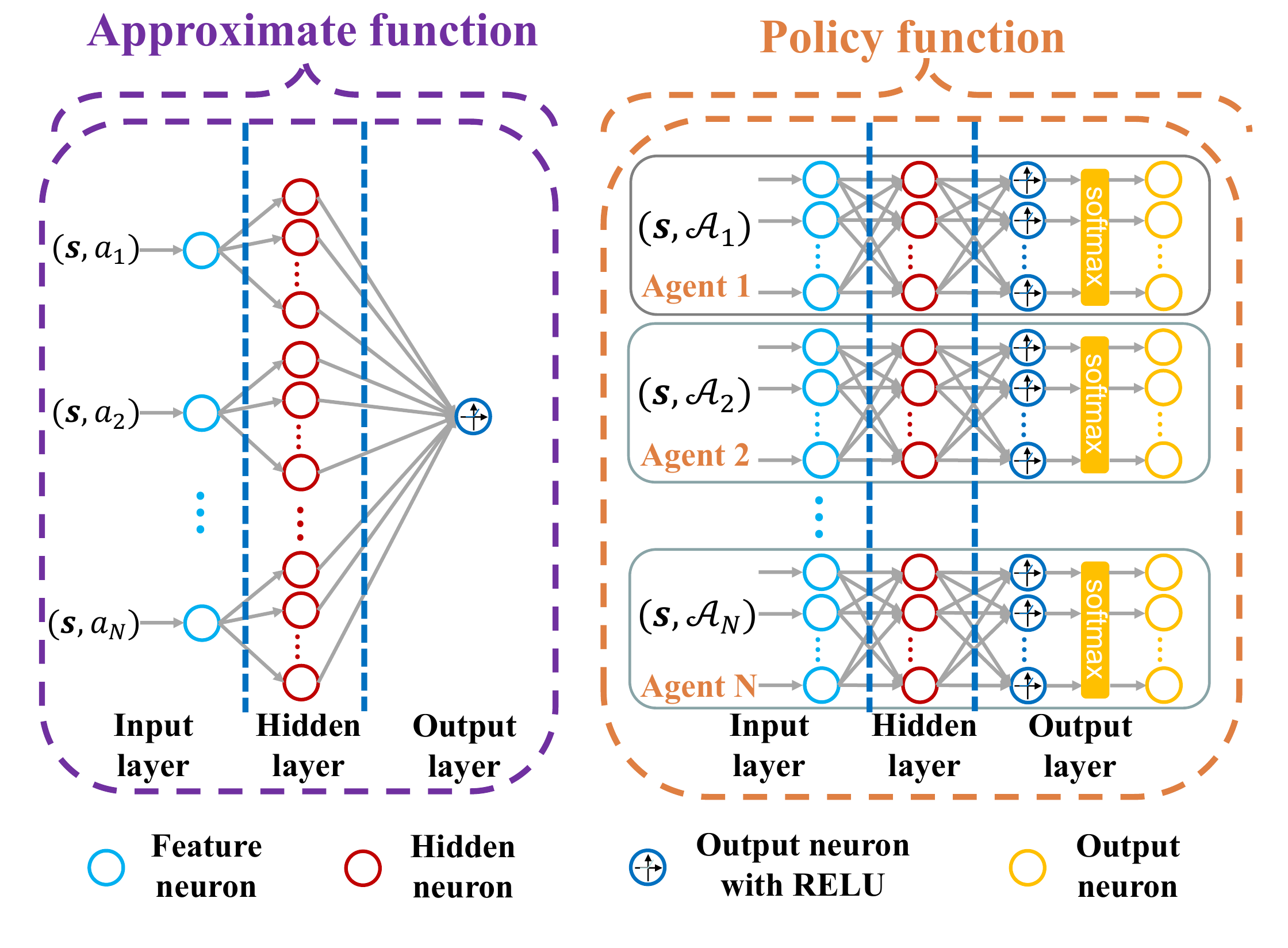}
\caption{Two novel neural
networks for the approximate function and policy function, respectively.}\label{neuralnetwork}
\end{figure}
%As shown in Fig.~\ref{neuralnetwork}, the approximate function is represented as
As depicted in Fig.~\ref{neuralnetwork}, the approximate function is represented as
\begin{align}\label{origitwo-layerneuralnetworkofQfunction}
\widehat{Q}(\bm{z};W')\!=\!\frac{1}{(mN)^{p}}\sum_{r=1}^{mN}b_{r}\mathrm{ReLU}\big((W'_{r})^{\top}x_{\lfloor(r-1)/m\rfloor+1}\big),
\end{align}
where $x_{i}=\phi_{i}(\bm{s},a_{i})\in\mathbb{R}^{d}$ is the feature vector of state-action pair $(\bm{s},a_{i})$ and satisfies $\|x_{i}\|_{2}=1$,
$W'\in\mathbb{R}^{dmN}$ is the approximate parameter, $m$ is the number of neurons, and $\mathrm{ReLU}(\cdot)$ is the rectified linear unit activation function.
For any joint policy $\bm{\pi_{\theta}}$,
denote $\varsigma_{\bm{\pi_{\theta}}}$ as the stationary distribution of $\bm{z}$.
%{\color{blue}By the definition of $\widehat{Q}(\bm{z};W')$ in (\ref{origitwo-layerneuralnetworkofQfunction}),
The approximate parameter $W'$ can be optimized by minimizing the following mean-squared Bellman error (MSBE) problem:
%\begin{align}\label{theMSBEofmathbfW}
%\min_{\bm{W}\in S^{W}_{B}}\mathrm{MSBE}(\bm{W})=\mathbb{E}_{\varsigma_{\bm{\theta}}}\big[\big(\widehat{Q}(s,\bm{a};\bm{W})-\mathcal{T}^{\bm{\pi}_{\bm{\theta}}}\widehat{Q}(s,\bm{a};\bm{W})\big)^{2}\big],
%\end{align}
%\begin{align}\label{theMSBEofmathbfW}
%\min_{\bm{W}\in S^{W}_{B}}\!\!\mathrm{MSBE}(\bm{W})\!\!=\!\!\mathbb{E}_{\varsigma_{\bm{\theta}}}\big[\big(\widehat{Q}(s,\bm{a};\bm{W})\!\!-\!\!\mathcal{T}^{\bm{\pi}_{\bm{\theta}}}\widehat{Q}(s,\bm{a};\bm{W})\big)^{2}\big],
%\end{align}
\begin{align}\label{theMSBEofmathbfWfirst}
\min_{W'\in S^{W}_{B}}\mathbb{E}_{\varsigma_{\bm{\pi_{\theta}}}}\Big[\Big(\widehat{Q}(\bm{z};W')-\big[\mathcal{T}^{\bm{\pi}_{\bm{\theta}}}\big(\widehat{Q}(W')\big)\big](\bm{z})\Big)^{2}\Big],
\end{align}
where $\widehat{Q}(W')=\big[\widehat{Q}(\bm{z};W')\big]_{\bm{z}\in\mathcal{S}\times\mathcal{A}}$, $S^{W}_{B}=\big\{W'\in\mathbb{R}^{dmN}:\|W'_{(j-1)m+1:jm}-W'_{(j-1)m+1:jm}(0)\|_{2}\leq B/\sqrt{N},\forall\;j=1,\cdots,N\big\}$ denotes the parameter space with $W'(0)\in\mathbb{R}^{dmN}$ serving as the initial value.
Note that $B/\sqrt{N}$ is utilized to characterize the convergence error of the neural network approximation.
%is the parameter space with $W'(0)\in\mathbb{R}^{dmN}$ as the initial value.
%Specially, $B/\sqrt{N}$ is used to character the convergence error of the neural network approximation.
\par
For MSBE problem (\ref{theMSBEofmathbfWfirst}), the approximate parameter can be updated by the following iteration:
\begin{align}\label{thecentralizediteration}
W'(t+1)=&P_{S^{W}_{B}}\Big(W'(t)-\eta_{c}\mathbb{E}_{\varsigma^{+}_{\bm{\pi_{\theta}}}}\big[\delta\big(\bm{z},\bm{z}';W'(t)\big)\notag\\
&\nabla_{W'}\widehat{Q}\big(\bm{z};W'(t)\big)\big]\Big),
\end{align}
where $P_{S^{W}_{B}}(\cdot)$ is the projection operator, $\eta_{c}$ is the learning rate, $\varsigma^{+}_{\bm{\pi_{\theta}}}$ is the stationary distribution of $(\bm{z},\bm{z}')$ generated by $\bm{\pi_{\theta}}$, and  $\delta\big(\bm{z},\bm{z}';W'(t)\big)=\widehat{Q}\big(\bm{z};W'(t)\big)-(1-\gamma)\bar{r}(\bm{z})-\gamma\widehat{Q}\big(\bm{z}';W'(t)\big)$ is the temporal difference (TD) error.
\par
By leveraging (\ref{thepolicygradienttheorem}) and (\ref{thecentralizediteration}), the NMARL problem can be addressed in a centralized way.
However, several substantial challenges remain.
For example, the update of the approximate parameter $W'(t+1)$ in (\ref{thecentralizediteration}) depends on the average reward $\bar{r}(\bm{z})$, which is not directly accessible due to the privacy constraints associated with each agent's individual reward.
Consequently, this necessitates the development of a distributed algorithm that avoids utilizing the average reward of agents directly.
%This motivates us to develop a distributed algorithm that avoids using the average reward of agents.
%Based on the above description, the NMARL problem can be addressed in a centralized way by (\ref{thepolicygradienttheorem}) and (\ref{thecentralizediteration}).
%However, it still poses several significant challenges.
%For instance, the update of approximate parameter $W'(t+1)$ in (\ref{thecentralizediteration}) relies on the average reward $\bar{r}(\bm{z})$, which cannot be directly available due to the privacy of each agent's reward.
%This motivates us to develop a distributed algorithm that avoids using the average reward of agents.
\subsection{Approximate stationary}\label{thesectionofApproximatestationary}
To establish a foundation for the subsequent convergence analysis, we introduce the concept of approximate stationary in this section.
\par
Define $W^{\dag}\in\mathbb{R}^{dmN}$ as the stationary point of (\ref{thecentralizediteration}) and satisfy
\begin{align}
\mathbb{E}_{\varsigma^{+}_{\bm{\pi_{\theta}}}}\big[\delta(\bm{z},\bm{z}';W^{\dag})\nabla_{W'}\widehat{Q}(\bm{z};W^{\dag})\big]^{\top}(W'-W^{\dag})\geq0,\label{theinequalityofstationarypointbefore}
\end{align}
for all $W'\in S^{W}_{B}$.
Based on $W^{\dag}$ in (\ref{theinequalityofstationarypointbefore}), we define a function class $\mathcal{F}^{\dag}_{B,m}$ as
\begin{align}\label{thedefinitionofmathcalFdag}
\mathcal{F}^{\dag}_{B,m}=&\Big\{\frac{1}{(mN)^{p}}\sum_{r=1}^{mN}\Big(b_{r}\mathds{1}\{(W^{\dag}_{r})^{\top}x_{\lfloor(r-1)/m\rfloor+1}>0\}\cdot\notag\\
&W'^{\top}_{r}x_{\lfloor(r-1)/m\rfloor+1}\Big),W'\in S^{W}_{B}\Big\}.
\end{align}
Similar to the analysis in~\cite{Cai2019}, we have that $\widehat{Q}(W^{\dag})=\big[\widehat{Q}(\bm{z};W^{\dag})\big]_{\bm{z}\in\mathcal{S}\times\mathcal{A}}$ is the global optimum of the following mean-squared projected Bellman error (MSPBE) problem:
\begin{align}\label{theMSBEofmathbfW}
\min_{Q\in\mathbb{R}^{|\mathcal{S}||\mathcal{A}|}}\mathbb{E}_{\varsigma_{\bm{\pi_{\theta}}}}\Big[\Big(Q(\bm{z})-\big[P_{\mathcal{F}^{\dag}_{B,m}}\mathcal{T}^{\bm{\pi}_{\bm{\theta}}}\big(Q\big)\big](\bm{z})\Big)^{2}\Big].
\end{align}
%In the definition of $\mathcal{F}^{\dag}_{B,m}$ in (\ref{thedefinitionofmathcalFdag}),
Considering the unknown nature of $W^{\dag}$, the MSPBE problem (\ref{theMSBEofmathbfW}) is not directly solvable.
In light of the specific characteristics of the approximate function $\widehat{Q}(\bm{z};W')$ in (\ref{origitwo-layerneuralnetworkofQfunction}), we adopt the concept of an approximate stationary point as a substitute for the stationary point in the following discussion.
Define $\widehat{Q}_{0}(\bm{z};W')$ as the local linearization of $\widehat{Q}(\bm{z};W')$ in (\ref{origitwo-layerneuralnetworkofQfunction}) at $W'(0)$, which is characterized by
\begin{align}\label{theQ0xW}
\widehat{Q}_{0}(\bm{z};W')=\frac{1}{(mN)^{p}}\sum_{r=1}^{mN}b_{r}\mathrm{ReLU}\big(W'^{\top}_{r}(0)x_{\lfloor(r-1)/m\rfloor+1}\big).
\end{align}
For joint policy $\bm{\pi}_{\bm{\theta}}$, the approximate stationary point $W^{*}\in\mathbb{R}^{dmN}$ of (\ref{thecentralizediteration}) satisfies
%For all $W'\in S^{W}_{B}$, the approximate stationary point $W^{*}$
\begin{align}
\mathbb{E}_{\varsigma^{+}_{\bm{\pi_{\theta}}}}\big[\delta_{0}(\bm{z},\bm{z}';W^{*})\nabla_{W}\widehat{Q}_{0}(\bm{z};W^{*})\big]^{\top}(W'-W^{*})\geq0\label{theinequalityofstationarypoint}
\end{align}
for all $W'\in S^{W}_{B}$, where $\delta_{0}(\bm{z},\bm{z}';W^{*})=\widehat{Q}_{0}(\bm{z};W^{*})-(1-\gamma)\bar{r}(\bm{z})-\gamma\widehat{Q}_{0}(\bm{z}';W^{*})$.
By utilizing $W'(0)$, we define another function class $\mathcal{F}_{B,m}$ as
\begin{align}\label{thedefinitionofmathcalF}
\mathcal{F}_{B,m}=&\Big\{\frac{1}{(mN)^{p}}\sum_{r=1}^{mN}\Big(b_{r}\mathds{1}\{W'^{\top}_{r}(0)x_{\lfloor(r-1)/m\rfloor+1}>0\}\cdot\notag\\
&W'^{\top}_{r}x_{\lfloor(r-1)/m\rfloor+1}\Big),W'\in S^{W}_{B}\Big\}.
\end{align}
According to the definitions of (\ref{theBellmanevaluationoperatorofpiTheta}) and (\ref{theinequalityofstationarypoint}),
it can be clearly observed that  $\widehat{Q}_{0}(W^{*})=\big[\widehat{Q}_{0}(W^{*})\big]_{\bm{z}\in\mathcal{S}\times\mathcal{A}}$ is the global optimum of the following MSPBE problem:
\begin{align}\label{theMSBEofmathbfWtwo}
\min_{Q\in\mathbb{R}^{|\mathcal{S}||\mathcal{A}|}}\mathbb{E}_{\varsigma_{\bm{\pi_{\theta}}}}\Big[\Big(Q(\bm{z})-\big[P_{\mathcal{F}_{B,m}}\mathcal{T}^{\bm{\pi}_{\bm{\theta}}}\big(Q\big)\big](\bm{z})\Big)^{2}\Big].
\end{align}
\begin{remark}
The specification of $B$ within the parameter space $S^{W}_{B}=\big\{W'\in\mathbb{R}^{dmN}:\|W'_{(j-1)m+1:jm}-W'_{(j-1)m+1:jm}(0)\|_{2}\leq B/\sqrt{N},\forall\;j=1,\cdots,N\big\}$ ensures that the designed neural network operates in a well-defined and rational manner.
As a result, when characterizing the convergence of the approximate function with a sufficiently large value for $m$, it is sufficient to concentrate on $\mathcal{F}_{B,m}$ as defined in (\ref{thedefinitionofmathcalF}), rather than considering $\mathcal{F}^{\dag}_{B,m}$ in (\ref{thedefinitionofmathcalFdag}).
%The setting of $B$ in the parameter space $S^{W}_{B}=\big\{W'\in\mathbb{R}^{dmN}:\|W'_{(j-1)m+1:jm}-W'_{(j-1)m+1:jm}(0)\|_{2}\leq B/\sqrt{N},\forall\;j=1,\cdots,N\big\}$ guarantees that the designed neural network operates rationally.
%Consequently, it becomes adequate to focus on $\mathcal{F}_{B,m}$ as defined in (\ref{thedefinitionofmathcalF}), rather than considering $\mathcal{F}^{\dag}_{B,m}$ in (\ref{thedefinitionofmathcalFdag}), when characterizing convergence of the approximate function with a suitably large value for $m$.
%ensures the rationality of designed neural network that it is sufficient to consider $\mathcal{F}_{B,m}$ in (\ref{thedefinitionofmathcalF}) instead of $\mathcal{F}^{\dag}_{B,m}$ in (\ref{thedefinitionofmathcalFdag}) for characterizing the convergence of the approximation function with a sufficiently large $m$.
%Since the initial value $W'(0)$ in the definition of $\mathcal{F}_{B,m}$ is given and fixed, this will simplify the subsequent convergence analysis.
Given that the initial value $W'(0)$ in the definition of $\mathcal{F}_{B,m}$ is specified and remains constant, this property will simplify the subsequent convergence analysis.
\end{remark}
%\par
%{\color{blue}Based on the above description, the NMARL problem can be addressed in a centralized way by (\ref{thepolicygradienttheorem}) and (\ref{thecentralizediteration}).
%However, it still poses several significant challenges.
%For instance, the update of approximate parameter $W'(t+1)$ in (\ref{thecentralizediteration}) relies on the average reward $\bar{r}(\bm{z})$,} which cannot be directly available due to the privacy of each agent's reward.
%{\color{blue}This motivates us to develop a distributed algorithm that avoids using the average reward of agents.}
\section{Distributed neural policy gradient algorithm}\label{DistributedNeuralpolicygradientalgorithm}
%{\color{blue}This section proposes} a distributed neural policy gradient algorithm for the NMARL problem.
%\subsection{Approximate function and policy function for each agent}\label{DCDAframework}
%In order to propose a distributed algorithm for the NMARL problem, we first design the neural network for approximate function and policy function of each agent.
%{\color{blue}To tackle the challenge of the global $Q$-function $Q^{\bm{\pi}_{\bm{\theta}}}(\bm{s},\bm{a})$ approximation resulting from the absence of average reward $\bar{r}(\bm{s},\bm{a})$, our approach is that each agent approximates $Q^{\bm{\pi}_{\bm{\theta}}}(\bm{s},\bm{a})$ based solely on its local reward and the parameter estimates from its neighbors through a time-varying network....}
\subsection{Expressions for approximate $Q$-function and policy function}\label{Overparameterizedneuralpolicy}
%In DCDA framework, the action-value function and policy function can be approximated in different forms, this paper focus on the approximation form of neural network.
%The DCDA framework enables the approximated $Q$-function and policy  in various forms.
%To enhance the expression and learning capability of the algorithm, this paper emphasizes the utilization of neural networks for approximation.
By using the neural network in Fig.~\ref{neuralnetwork},
the approximate $Q$-function $\widehat{Q}_{i}(\bm{z};W_{i})$ of agent $i$ is represented as
\begin{align}\label{two-layerneuralnetworkofQfunction}
\widehat{Q}_{i}(\bm{z};W_{i})=\frac{1}{(mN)^{p}}\sum_{r=1}^{mN}b_{r}\mathrm{ReLU}(W_{i,r}^{\top}x_{\lfloor(r-1)/m\rfloor+1}),
\end{align}
where $\frac{1}{2}<p<1$ and $W_{i}=(W_{i,1}^{\top},\cdots,W_{i,mN}^{\top})^{\top}\in\mathbb{R}^{dmN}$ denotes the approximate parameter and satisfies the condition $W_{i}(0)=W'(0)$.
%where $\frac{1}{2}<p<1$ and  $W_{i}=(W_{i,1}^{\top},\cdots,W_{i,mN}^{\top})^{\top}\in\mathbb{R}^{dmN}$ is {\color{blue}the approximate parameter} and satisfies $W_{i}(0)=W'(0)$.
%{\color{blue}In order to ensure the global convergence, we define the parameter space of $W_{i}$ as}
%$S^{W}_{B}=\big\{W'\in\mathbb{R}^{dmN}:\|W'_{(j-1)m+1:jm}-W'_{(j-1)m+1:jm}(0)\|_{2}\leq B/\sqrt{N},\forall\;j=1,\cdots,N\big\}$.
%Since the scalar homogeneity property of $\mathrm{ReLU}(\cdot)$, i.e., $\mathrm{ReLU}(cu)=c\mathrm{ReLU}(u)$ for all $c>0$ and $u\in\mathbb{R}$, we maintain the fixed values of $\{b_{r}\}_{r=1,\cdots,Nm}$ and solely update $\bm{W}_{i}$ during the training process.
Due to the scalar homogeneity property of $\mathrm{ReLU}(\cdot)$ that $\mathrm{ReLU}(cu)=c\mathrm{ReLU}(u)$ for all $c>0$ and $u\in\mathbb{R}$, it is feasible to maintain the fixed value of $b_{r}$ without modification while only updating $W_{i}$ during the training process.
\begin{remark}
%In the approximate function (\ref{two-layerneuralnetworkofQfunction}), $\frac{1}{2}<p<1$ is designed to guarantee that each agent can accurately evaluate the joint policy
%and mitigate the policy gradient error caused by the different rewards.
%Specially, in the proof of (ii) of Lemma~\ref{Lemma4-6}, the range of $\frac{1}{2}<p<1$ is used in (\ref{firstterminexpectationofgradientdifference}) to regulate the $3R^{2}_{0}$ to avoid the existence of constant term in the error upper bound.
In the approximate function (\ref{two-layerneuralnetworkofQfunction}), the condition $\frac{1}{2} < p < 1$ is specifically designed to ensure that each agent can accurately evaluate the joint policy while mitigating the policy gradient error induced by differing reward structures.
More precisely, in the proof of part (ii) of Lemma~\ref{Lemma4-6}, the range $\frac{1}{2} < p < 1$ is utilized in (\ref{expectationofgradientdifference}) and (\ref{firstterminexpectationofgradientdifference}) to regulate the term $12R^{2}_{0}$, thereby preventing the presence of a constant term in the upper bound of the error.
\end{remark}
%Considering that $\mathrm{ReLU}(\cdot)$ satisfies scalar homogeneity, i.e., $\mathrm{ReLU}(cu)=c\mathrm{ReLU}(u)$ for all $c>0$ and $u\in\mathbb{R}$, we keep $\{b_{r}\}_{r\in[Nm]}$ fixed and only update $\bm{W}$ during the training process.
%{\color{red}Define $\widehat{Q}_{0}(s,\bm{a};\bm{W})$ as the local linearization of $\widehat{Q}(s,\bm{a};\bm{W})$ at $\bm{W}(0)$, which is represented as
%\begin{align}\label{theQ0xW}
%\widehat{Q}_{0}(s,\bm{a};\bm{W})=\frac{1}{(Nm)^{p}}\sum_{r=1}^{Nm}b_{r}\mathrm{ReLU}\big(W_{r}^{\top}(0)x_{\lfloor(r-1)/m\rfloor+1}\big).
%\end{align}}
%{\color{blue}In particular, denote $\widehat{Q}_{0}(s,\bm{a})=\widehat{Q}_{0}\big(s,\bm{a};\bm{W}(0)\big)$ for notational simplicity.}
%{\color{red}According to the definition of $W_{r}(0)$, define a function class $\mathcal{F}_{B,m}$ as
%\begin{align}\label{thedefinitionofmathcalF}
%\mathcal{F}_{B,m}=&\Big\{\frac{1}{(Nm)^{p}}\sum_{r=1}^{Nm}\big(b_{r}\mathds{1}\{W^{\top}_{r}(0)x_{\lfloor(r-1)/m\rfloor+1}>0\}\cdot\notag\\
%&W^{\top}_{r}x_{\lfloor(r-1)/m\rfloor+1}\big),\bm{W}\in S^{W}_{B}\Big\}.
%\end{align}}
\par
%As depicted in the right side neural network in Fig.~\ref{neuralnetwork}, we define $f_{i}(\bm{s},a_{i};\theta_{i})$ as the feature value of $(\bm{s},a_{i})$ and satisfy
As illustrated in the neural network on the right side of Fig.~\ref{neuralnetwork}, we define $f_{i}(\bm{s},a_{i};\theta_{i})$ as the feature value corresponding to the state-action pair $(\bm{s},a_{i})$, which satisfies
\begin{align}\label{thetwolayerneuralnetworkinAC}
f_{i}(\bm{s},a_{i};\theta_{i})=\frac{1}{(mN)^{p}}\sum_{r=1}^{m}b_{(i-1)m+r}\mathrm{ReLU}(\theta^{\top}_{i,r}x_{i}),
\end{align}
where $\theta_{i}=(\theta_{i,1}^{\top},\cdots,\theta_{i,m}^{\top})^{\top}\in\mathbb{R}^{md}$.
%is initialized to $\theta_{i}(0)=\big(\theta^{\top}_{i,1}(0),\cdots,\theta^{\top}_{i,m}(0)\big)^{\top}\in\mathbb{R}^{md}$ with $\theta_{i,r}(0)\sim N(0,I_{d}/d)$ for all $r\in[m]$
%As shown in Fig~\ref{neuralnetworkactor}, based on the two-layer neural network designed in (\ref{thetwolayerneuralnetworkinAC}),
By using (\ref{thetwolayerneuralnetworkinAC}), the local policy $\pi_{i}(a_{i}|\bm{s};\theta_{i})$ of agent $i$ is denoted as
\begin{align}\label{theparameterizedpolicy}
\pi_{i}(a_{i}|\bm{s};\theta_{i})=\frac{\exp[f_{i}(\bm{s},a_{i};\theta_{i})]}{\sum_{a'_{i}\in\mathcal{A}_{i}}\exp[f_{i}(\bm{s},a'_{i};\theta_{i})]}.
\end{align}
%Recall that
%$\bm{\pi}(\bm{a}|s;\bm{\theta})=\prod_{i=1}^{N}\pi_{i}(a_{i}|s;\theta_{i})$, without loss of generality, we keep $\{b_{r}\}_{r\in[Nm]}$ fixed and only update the $\{\theta_{i}\}_{i\in[N]}$ throughout training.
Denote  $\psi_{i}(\bm{s},a_{i};\theta_{i})=\big(\psi^{\top}_{i,1}(\bm{s},a_{i};\theta_{i}),\cdots,\psi^{\top}_{i,m}(\bm{s},a_{i};\theta_{i})\big)^{\top}$ as the feature mapping of $f_{i}(\bm{s},a_{i};\theta_{i})$ and satisfy
\begin{align}\label{thefeaturemapofpsi}
\psi_{i,r}(\bm{s},a_{i};\theta_{i})=\frac{b_{(i-1)m+r}}{(mN)^{p}}\mathds{1}\{\theta^{\top}_{i,r}x_{i}>0\}x_{i},
\end{align}
for all $r=1,\cdots,m$.
Moreover, define $\bm{\psi}(\bm{z};\bm{\theta})$ as the joint feature mapping and expressed as
\begin{align}\label{thejointfeaturemapping}
\bm{\psi}(\bm{z};\bm{\theta})=\big(\psi^{\top}_{1}(s,a_{1};\theta_{1}),\cdots,\psi^{\top}_{N}(s,a_{N};\theta_{N})\big)^{\top}.
\end{align}
Since $b_{r}\sim\mathrm{Unif}(\{-1,1\})$ and $\|x_{i}\|_{2}=1$, we have
\begin{align}
\|\psi_{i,r}(\bm{s},a_{i};\theta_{i})\|_{2}\leq\frac{1}{(mN)^{p}}\label{thenormofpsi}
\end{align}
and
\begin{align}\label{thevarofneuralpolicy}
\left\{
\begin{array}{ll}
f_{i}(\bm{s},a_{i};\theta_{i})=\psi^{\top}_{i}(\bm{s},a_{i};\theta_{i})\theta_{i}\\
\nabla_{\theta_{i}}f_{i}(\bm{s},a_{i};\theta_{i})=\psi_{i}(\bm{s},a_{i};\theta_{i})
\end{array}
\right.
\end{align}
almost everywhere.
\par
For $\psi_{i,r}(\bm{s},a_{i};\theta_{i})$ in (\ref{thefeaturemapofpsi}),
we define
\begin{align}
\left\{
\begin{array}{ll}\notag
\overline{\psi}_{i,r}(\bm{s},a_{i};\theta_{i})\!=\!\psi_{i,r}(\bm{s},a_{i};\theta_{i})-\mathbb{E}_{a'_{i}\sim\pi_{\theta_{i}}(\cdot|\bm{s})}\big[\psi_{i}(\bm{s},a'_{i};\theta_{i})\big]\\
\overline{\psi}_{i}(\bm{s},a_{i};\theta_{i})=\big(\overline{\psi}^{\top}_{i,1}(\bm{s},a_{i};\theta_{i}),\cdots,\overline{\psi}^{\top}_{i,m}(\bm{s},a_{i};\theta_{i})\big)^{\top}
\end{array}
\right.
\end{align}
%$\overline{\psi}_{i,r}(\bm{s},a_{i};\theta_{i})=\psi_{i,r}(\bm{s},a_{i};\theta_{i})-\mathbb{E}_{a'_{i}\sim\pi_{\theta_{i}}(\cdot|\bm{s})}\big[\psi_{i}(\bm{s},a'_{i};\theta_{i})\big]$ and $\overline{\psi}_{i}(\bm{s},a_{i};\theta_{i})=\big(\overline{\psi}^{\top}_{i,1}(\bm{s},a_{i};\theta_{i}),\cdots,\overline{\psi}^{\top}_{i,m}(\bm{s},a_{i};\theta_{i})\big)^{\top}$,
and have the
following proposition.
\begin{proposition}\label{ThepolicygradienttheoremforMARL}
For the local policy $\pi_{i}(a_{i}|\bm{s};\theta_{i})$ of each agent $i\in\mathcal{N}$ as defined in (\ref{theparameterizedpolicy}),
the policy gradient of $J(\bm{\pi}_{\bm{\theta}})$ with respect to $\theta_{i}$ is represented as
\begin{align}\label{theresultofpolicygradientforagenti}
\nabla_{\theta_{i}}J(\bm{\pi}_{\bm{\theta}})=&\mathbb{E}_{\sigma_{\bm{\pi}_{\bm{\theta}}}}\Big[Q^{\bm{\pi}_{\bm{\theta}}}(\bm{s},\bm{a})\overline{\psi}_{i}(\bm{s},a_{i};\theta_{i})\Big].
\end{align}
\end{proposition}
\begin{proof}
By the definition of $\pi_{i}(a_{i}|\bm{s};\theta_{i})$ in (\ref{theparameterizedpolicy}), we have
\begin{align}\label{theresultofnablalogpii}
&\nabla_{\theta_{i}}\log{\pi_{i}(a_{i}|\bm{s};\theta_{i})}\notag\\
=&\nabla_{\theta_{i}}f_{i}(\bm{s},a_{i};\theta_{i})-\frac{\sum_{a'_{i}\in\mathcal{A}_{i}}\nabla_{\theta_{i}}f_{i}(\bm{s},a'_{i};\theta_{i})\exp[f_{i}(\bm{s},a'_{i};\theta_{i})]}{\sum_{a'_{i}\in\mathcal{A}_{i}}\exp[f_{i}(\bm{s},a'_{i};\theta_{i})]}\notag\\
=&\nabla_{\theta_{i}}f_{i}(\bm{s},a_{i};\theta_{i})-\mathbb{E}_{a'_{i}\sim\pi_{i}(\cdot|\bm{s};\theta_{i})}[\nabla_{\theta_{i}}f_{i}(\bm{s},a'_{i};\theta_{i})]\notag\\
=&\psi_{i}(\bm{s},a_{i};\theta_{i})-\mathbb{E}_{\pi_{\theta_{i}}}[\psi_{i}(\bm{s},a'_{i};\theta_{i})]\notag\\
=&\overline{\psi}_{i}(\bm{s},a_{i};\theta_{i}),
\end{align}
where the third equality comes from (\ref{thevarofneuralpolicy}).
Substituting (\ref{theresultofnablalogpii}) into (\ref{thepolicygradienttheorem}), we can complete the proof.
%\begin{align}
%\nabla_{\theta_{i}}J(\bm{\pi}_{\bm{\theta}})=&\mathbb{E}_{\sigma_{\bm{\pi}_{\bm{\theta}}}}\Big[Q^{\bm{\pi}_{\bm{\theta}}}(\bm{s},\bm{a})\overline{\psi}_{i}(\bm{s},a_{i};\theta_{i})\Big],\notag
%\end{align}
%{\color{blue}which completes the proof.}
\end{proof}
%Proposition~\ref{ThepolicygradienttheoremforMARL} provide a
%\subsection{Shared initialization and compatible function approximation}
\par
Recall the initial value $W'(0)$ of parameter $W_{i}$ in (\ref{two-layerneuralnetworkofQfunction}), define the joint policy parameter space of agents as $\prod_{i=1}^{N}S^{\theta}_{i,B}$, where $S^{\theta}_{i,B}=\{\theta_{i}:\|\theta_{i}-\theta_{i}(0)\|_{2}\leq B/\sqrt{N}\}$ and $\theta_{i,r}(0)=W'_{(i-1)m+r}(0)$ for all $i=1,\cdots,N$ and $r=1,\cdots,m$.
%{\color{blue}Our designed two-layer neural network structure in Fig.~\ref{neuralnetwork} can ensure the compatibility between approximated $Q$-function (\ref{two-layerneuralnetworkofQfunction}) and policy function (\ref{thetwolayerneuralnetworkinAC}), which is essential for the optimality and convergence of policy gradient methods.}
%The definition of compatibility is described as follows.
%\begin{definition}\label{compatiblefunctionapproximations}
%For a joint policy $\bm{\pi_{\theta}}$ and approximate $Q$-function $\widehat{Q}(\bm{z};W')$, define the advantage function $\widehat{A}(\bm{z};\bm{W})$ as
%\begin{align}\label{theadvantagefunction}
%\widehat{A}(\bm{z};\bm{W})=\widehat{Q}(\bm{z};\bm{W})-\mathbb{E}_{\bm{a}'\sim\bm{\pi_{\theta}}(\cdot|\bm{s})}[\widehat{Q}(\bm{s},\bm{a}';\bm{W})].
%\end{align}
%If $\nabla_{\bm{W}}\widehat{A}(\bm{z};\bm{W})=\nabla_{\bm{\theta}}\log\bm{\pi}(\bm{a}|\bm{s};\bm{\theta})$ for all $\bm{z}\in\mathcal{S}\times\mathcal{A}$,
%we call that $\widehat{Q}(\bm{z};\bm{W})$ is compatible with $\bm{\pi}(\bm{a}|\bm{s};\bm{\theta})$.
%\end{definition}
For joint policy $\bm{\pi}_{\bm{\theta}}$,
%For  in (\ref{two-layerneuralnetworkofQfunction}),
we define the advantage function of agent $i$ for $\widehat{Q}_{i}(\bm{z};W_{i})$ as
\begin{align}\label{advantageinproposition}
\widehat{A}_{i}(\bm{z};W_{i})=\widehat{Q}_{i}(\bm{z};W_{i})-\mathbb{E}_{\bm{a}'\sim\bm{\pi_{\theta}}(\cdot|\bm{s})}\big[\widehat{Q}_{i}(\bm{s},\bm{a}';W_{i})\big],
\end{align}
which satisfies the following property.
%{\color{blue}and have the following lemma.}
\begin{lemma}\label{thelemmaofcompatibility}
For the local policy $\pi_{i}(a_{i}|\bm{s};\theta_{i})$ of each agent $i$ as defined in (\ref{theparameterizedpolicy}), we have
%$\widehat{Q}_{i}(\bm{z};W_{i})$ is approximately compatible with $\bm{\pi}_{\bm{\theta}}(\bm{a}|\bm{s})$ as the
\begin{align}\label{theadvantagelemma}
\nabla_{W_{i}}\widehat{A}_{i}(\bm{z};W_{i})=\nabla_{\bm{\theta}}\log{\bm{\pi}_{\bm{\theta}}(\bm{a}|\bm{s})}
\end{align}
with $m$ is sufficiently large.
\end{lemma}
\begin{proof}
%By the definition of advantage function in (\ref{theadvantagefunction}), for each agent $i\in\mathcal{N}$ and $r\in\{1,\cdots,m\}$,
By the definition of $\widehat{A}_{i}(\bm{z};W_{i})$ in  (\ref{advantageinproposition}), it holds almost everywhere that $\nabla_{W_{i}}\widehat{A}_{i}(\bm{z};W_{i})=\big(\overline{\psi}^{\top}_{1}(\bm{s},a_{1};W_{i,1:m}),\cdots,\overline{\psi}^{\top}_{N}(\bm{s},a_{N};$ $W_{i,m(N-1)+1:mN})\big)^{\top}$ and $\nabla_{\bm{\theta}}\log{\bm{\pi}_{\bm{\theta}}(\bm{a}|\bm{s})}=\big(\overline{\psi}^{\top}_{1}(\bm{s},a_{1};\theta_{1}),$ $\cdots,\overline{\psi}^{\top}_{N}(\bm{s},a_{N};\theta_{N})\big)^{\top}$.
%Based on $S^{W}_{B}$ and $S^{\theta}_{i,B}$,
Since the initial value $\theta_{i,r}(0)=W'_{(i-1)m+r}(0)=W_{i,(i-1)m+r}(0)$ for all $i=1,\cdots,N$ and $r=1,\cdots,m$, we can obtain (\ref{theadvantagelemma})
%\begin{align}\notag
%\nabla_{W_{i}}\widehat{A}_{i}(\bm{z};W_{i})=\nabla_{\bm{\theta}}\log{\bm{\pi}_{\bm{\theta}}(\bm{a}|\bm{s})}.
%\end{align}
%$\widehat{Q}(\bm{s},\bm{a};W)$ is approximately compatible with $\bm{\pi}(\bm{a}|\bm{s};\bm{\theta})$
with $m$ is sufficiently large.
\end{proof}
\begin{remark}
Lemma~\ref{thelemmaofcompatibility} indicates that the designed approximate function $\widehat{Q}_{i}(\bm{z};W_{i})$ in (\ref{two-layerneuralnetworkofQfunction}) and the policy function $\pi_{i}(a_{i}|\bm{s};\theta_{i})$ in (\ref{theparameterizedpolicy}) satisfy the condition of compatible function approximation, a key requirement commonly utilized in natural policy gradient methods~\cite{Peters2008Neurocomputing}-\cite{Xu2020NIPS}.
As such, the proposed neural network architectures in Fig.~\ref{neuralnetwork} exhibit the potential for extension to natural policy gradient frameworks.
%Lemma~\ref{thelemmaofcompatibility} indicates that the designed approximate function $\widehat{Q}_{i}(\bm{z};W_{i})$ in (\ref{two-layerneuralnetworkofQfunction}) and policy function $\pi_{i}(a_{i}|\bm{s};\theta_{i})$ in (\ref{theparameterizedpolicy})
%satisfy the condition of the compatible function approximation, which is commonly employed in natural policy gradient methods~\cite{Peters2008Neurocomputing}-\cite{Xu2020NIPS}. Consequently, our designed neural networks possess the potential to be extended to natural policy gradient methods.
\end{remark}
\subsection{Distributed neural policy gradient algorithm}\label{Criticupdate}
To design a distributed algorithm, several key challenges need to be systematically addressed.
For instance,
%for the joint policy $\bm{\pi_{\theta}}$,
(i) how can agents estimate the global $Q$-function $Q^{\bm{\pi_{\theta}}}(\bm{s},\bm{a})$ without access to the average reward information?
(ii) How can agents update their policy parameters $\bm{\theta}$ to cooperatively maximize the objective function $J(\bm{\pi_{\theta}})$ as defined in (\ref{thedefinelongtermreturninMARL})?
%Based on the designed neural networks, some challenges also should be addressed for developing distributed algorithm.
%For example, for joint policy $\bm{\pi_{\theta}}$,
%(i) how can agents estimate the global $Q$-function $Q^{\bm{\pi_{\theta}}}(\bm{s},\bm{a})$ without access to the average reward information;
%(ii) how can agents update their policy parameters $\bm{\theta}$ to cooperate maximize the objective function $J(\bm{\pi_{\theta}})$ in (\ref{thedefinelongtermreturninMARL})?
\begin{figure}[!htb]
\centering
\includegraphics[width=0.8\hsize]{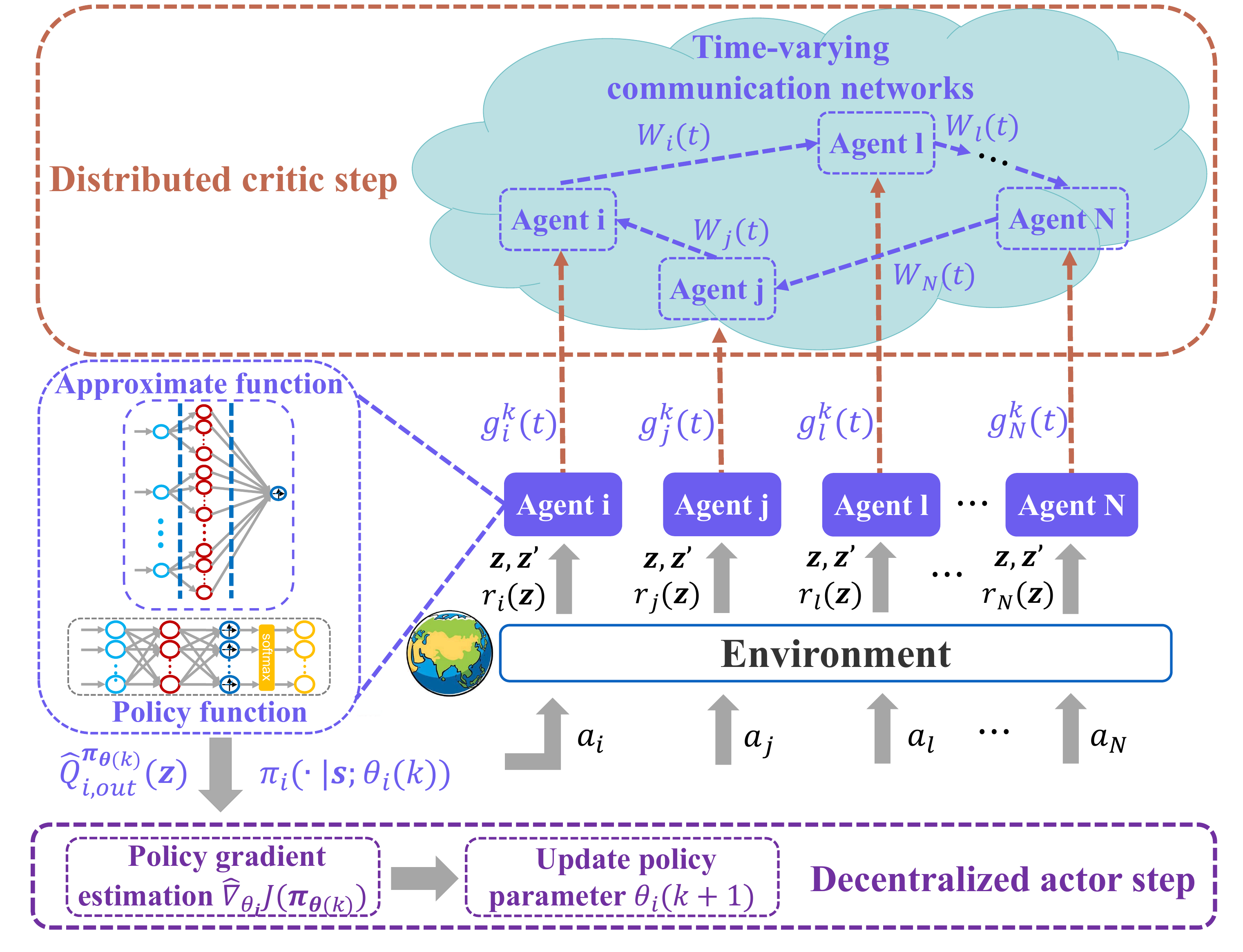}
\caption{The flow diagram of distributed policy gradient algorithm.}\label{DCDAframeworkfigure}
\end{figure}
\par
The key approaches for addressing the aforementioned challenges are illustrated in the flow diagram in Fig.~\ref{DCDAframeworkfigure}. (i) In the distributed critic step, agents estimate the global $Q$-function $Q^{\bm{\pi_{\theta}}}(\bm{z})$ by utilizing only their local rewards and the approximate parameters of their neighboring agents;
(ii) Since the approximate function serves as an estimation of the global $Q$-function $Q^{\bm{\pi_{\theta}}}(\bm{z})$, each agent independently relies on its own approximate function to update its local policy parameter and collaboratively optimize the objective in (\ref{thedefinelongtermreturninMARL}), i.e., decentralized actor step.
%The primary ideas to address the aforementioned challenges are presented in the flow diagram in Fig.~\ref{DCDAframeworkfigure}.
%Specially, (i) in the distributed critic step, agents can estimate the global $Q$-function $Q^{\bm{\pi_{\theta}}}(\bm{z})$ only use their local rewards and the approximate parameters of their neighbors;
%(ii) given that the approximate function serves as an estimation of the global $Q$-function $Q^{\bm{\pi_{\theta}}}(\bm{z})$, each agent can solely rely on its own approximate function to update its local policy parameter and collaboratively optimize the objective (\ref{thedefinelongtermreturninMARL}), thereby performing a decentralized actor step.
\par
%In order to implement the process in Fig.~\ref{DCDAframeworkfigure},
Before developing distributed algorithm, we present the assumption of communication network among agents.
%Before
%the assumption of the communication network among agents is described as follows.
\begin{assumption}\label{theassumptionofcommunication}
The time-varying communication networks $\{\mathcal{G}_{t}(\mathcal{N},\mathcal{E}_{t})\}_{t\geq0}$ among the agents is uniformly strongly connected
%Agents can receive their neighboring agents' approximate parameters through $\{\mathcal{G}_{t}(\mathcal{N},\mathcal{E}_{t})\}_{t\geq0}$.
%Moreover, the corresponding weight matrix $A(t)$ at time $t$ is doubly stochastic and there exists a constant $\alpha$ such that $a_{ij}(t)\geq\alpha,\forall e_{ij}(t)\in\mathcal{E}(t)$.
and its associated weight matrix $A(t)$ at time $t$ is doubly stochastic.
Additionally, there exists a positive constant $\alpha > 0$ such that $a_{ij}(t) \geq \alpha$ for all $e_{ij}(t) \in \mathcal{E}(t)$.
\end{assumption}
%{\color{blue}
%\begin{remark}
%%In general, the knowledge learned by agents
%%is typically confidential and plays a pivotal role in their
%%decision-making process.
%As stated in Assumption~\ref{theassumptionofcommunication}, agents exchange the approximate parameters through time-varying communication networks, thereby significantly enhancing their privacy protection.
%\end{remark}
%}
\par
Assumption~\ref{theassumptionofcommunication} is a common assumption for communication network and also appears in~\cite{Doan2019ICML}.
%Under Assumption~\ref{theassumptionofcommunication}, the detailed operations in Fig.~\ref{DCDAframeworkfigure} are as follows
Next, we introduce the key steps of distributed algorithm.
\par
\textbf{Distributed critic step}:
%Assume that the joint policy is ,
Let $\bm{\pi}_{\bm{\theta}(k)}$ be joint policy in the $k$-th iteration of policy, where $\bm{\theta}(k)=\big(\theta^{\top}_{1}(k),\dots,\theta^{\top}_{N}(k)\big)^{\top}$ denotes the joint policy parameter.
Define
$\varsigma^{+}_{k}$ as the stationary distribution of $(\bm{z},\bm{z}')$ generated by $\bm{\pi}_{\bm{\theta}(k)}$.
For each agent $i$, the update of approximate parameter $W_{i}(t+1)$ with initial value $W_{i}(0)$ is designed as
\begin{subequations}\label{thekeystepindistributedneuralTDalgorithm}
\begin{numcases}{}
g^{k}_{i}(t)=\delta^{k}_{i}(t)\nabla_{W_{i}}\widehat{Q}_{i}\big(\bm{z};W_{i}(t)\big)\label{theTDlearningindistributedneuralTDalgorithm1}\\
\widetilde{W}_{i}(t+1)=\sum_{j\in\mathcal{N}_{i}(t)}a_{ij}(t)W_{j}(t)-\eta_{c,t}g^{k}_{i}(t)\label{theTDlearningindistributedneuralTDalgorithm}\\
W_{i}(t+1)=P_{S^{W}_{B}}\big(\widetilde{W}_{i}(t+1)\big),\label{theprojectindistributedneuralTDalgorithm}
\end{numcases}
\end{subequations}
where $\delta^{k}_{i}(t)=\widehat{Q}_{i}\big(\bm{z};W_{i}(t)\big)-(1-\gamma)r_{i}(\bm{z})-\gamma\widehat{Q}_{i}\big(\bm{z}';W_{i}(t)\big)$ is the TD error with $(\bm{z},\bm{z}')$ sampling drawn from the stationary distribution $\varsigma^{+}_{k}$,
$g^{k}_{i}(t)$ is the stochastic semigradient of agent $i$, and $\eta_{c,t}$ is the learning rate in the distributed critic step.
Note that in (\ref{theTDlearningindistributedneuralTDalgorithm}), the update of the approximate parameters for agents is performed in a distributed manner, where each agent $i$ relies exclusively on the approximate parameters of its neighboring agents.
%(\ref{theTDlearningindistributedneuralTDalgorithm}) shows that the update of the approximate parameters of agents are performed in a distributed way, with each agent $i$ solely relying on its neighbors' approximate parameters.
%After $T_{c}-1$ iterations of (\ref{thekeystepindistributedneuralTDalgorithm}), each agent $i$ can obtain an estimate of the global $Q$-function $Q^{\bm{\pi}_{\bm{\theta}(k)}}(\bm{z})$, denoted as
After $T_{c}-1$ iterations of (\ref{thekeystepindistributedneuralTDalgorithm}), each agent $i$ is able to obtain an estimate of the global $Q$-function $Q^{\bm{\pi}_{\bm{\theta}(k)}}(\bm{z})$, denoted as
\begin{align}\label{theQiout}
\widehat{Q}^{\bm{\pi}_{\bm{\theta}(k)}}_{i,\mathrm{out}}(\bm{z})=\widehat{Q}_{i}\big(\bm{z};W_{i,\mathrm{out}}\big),
\end{align}
where $W_{i,\mathrm{out}}=\frac{1}{\sum_{t=0}^{T_{c}-1}\eta_{c,t}}\sum_{t=0}^{T_{c}-1}\eta_{c,t}W_{i}(t)$.
\par
\textbf{Decentralized actor step}:
%Define $\sigma_{k}$ as the state-action visitation measure of $\bm{\pi}_{\bm{\theta}(k)}$ and $\mathcal{B}=\{\bm{z}_{b}\}_{1\leq b\leq |\mathcal{B}|}$ as a batch sampled from $\sigma_{k}$.
%By using $\widehat{Q}^{k}_{i,\mathrm{out}}(\bm{z})$ in (\ref{theQiout}), each agent $i$ can
%estimate the true policy gradient $\nabla_{\theta_{i}}J(\bm{\pi}_{\bm{\theta}(k)})$ in (\ref{thepolicygradienttheorem}) by
Let $\sigma_{k}$ be the state-action visitation measure corresponding to $\bm{\pi}_{\bm{\theta}(k)}$ and $\mathcal{B}=\{\bm{z}_{b}\}_{1\leq b\leq |\mathcal{B}|}$ denote a batch sampled from $\sigma_{k}$.
By utilizing $\widehat{Q}^{\bm{\pi}_{\bm{\theta}(k)}}_{i,\mathrm{out}}(\bm{z})$ as defined in (\ref{theQiout}), each agent $i$ can estimate the true policy gradient $\nabla_{\theta_{i}}J(\bm{\pi}_{\bm{\theta}(k)})$ in (\ref{thepolicygradienttheorem}) by
\begin{align}\label{thesetimationofJpiTheta}
\widehat{\nabla}_{\theta_{i}}J(\bm{\pi}_{\bm{\theta}(k)})=\frac{1}{|\mathcal{B}|}\sum_{b=1}^{|\mathcal{B}|}\widehat{Q}^{\bm{\pi}_{\bm{\theta}(k)}}_{i,\mathrm{out}}(\bm{z}_{b})\nabla_{\theta_{i}}\log{\pi_{i}\big(a_{i,b}|\bm{s}_{b};\theta_{i}(k)\big)}.
\end{align}
%According to  $\widehat{\nabla}_{\theta_{i}}J(\bm{\pi}_{\bm{\theta}(k)})$ in (\ref{thesetimationofJpiTheta}),
Then, the update of policy parameter $\theta_{i}(k+1)$ is designed as
\begin{align}\label{theupdateofpolicyparameter}
\theta_{i}(k+1)=P_{S^{\theta}_{i,B}}\big(\theta_{i}(k)+\eta_{a}\widehat{\nabla}_{\theta_{i}}J(\bm{\pi}_{\bm{\theta}(k)})\big),
\end{align}
where $\eta_{a}$ is the learning rate in the decentralized actor step.
\par
By integrating the distributed critic step with the decentralized actor step, a distributed neural policy gradient algorithm is proposed in Algorithm~\ref{distributedneuralpolicygradientAlgorithm}.
%Combining the distributed critic step and decentralized actor step, a distributed neural policy gradient algorithm is proposed in Algorithm~\ref{distributedneuralpolicygradientAlgorithm}.
\begin{algorithm}
%\setstretch{0.8}
\SetAlgoLined
\textbf{Input:} The non-negative integers $T_{c}$, $K$, and $|\mathcal{B}|$, the learning rates $\eta_{c,t}$ and $\eta_{a}$\;
\textbf{Initialization}: Initialize the approximate parameter and policy parameter $W_{i,(j-1)m+r}(0)=W'_{(j-1)m+r}(0)=$ $\theta_{j,r}(0)$
%$\sim N(0,I_{d}/N)$
for all $i,j=1,\cdots,N$ and $r=1\cdots,m$\;
\For{$k=0$ $\mathrm{to}$ $K-1$}{
\tcc*[h]{Distributed critic step}\\
\For{$t=0$ $\mathrm{to}$ $T_{c}-1$}{
Sample a tuple $(\bm{z},\{r_{i}(\bm{z})\}_{i\in\mathcal{N}},\bm{z}')$ from the
stationary distribution $\varsigma^{+}_{k}$\;
Each agent $i\in\mathcal{N}$ updates $W_{i}(t+1)$ by (\ref{thekeystepindistributedneuralTDalgorithm})\;}
Each agent $i\in\mathcal{N}$ obtains the approximate $Q$-function
$\widehat{Q}^{\bm{\pi}_{\bm{\theta}(k)}}_{i,\mathrm{out}}(\bm{z})$ in (\ref{theQiout})\;
%$\{\bm{W}_{i,\mathrm{ave}}\}_{i\in[N]}$, $\{\bm{W}_{i}(0)\}_{i\in[N]}$ and $\{b_{r}\}_{r\in[Nm]}$ as the initialization, $K_{c}$ as the number of critic iterations, and $\eta_{c}$ as the fixed step-size in distributed critic step\;
%Take the batch $\mathcal{D}$ with sample $\{(s_{d},\bm{a}_{d})\}_{d\in[D]}$ from the visitation measure $\sigma_{t}$\;
\tcc*[h]{Decentralized actor step}\\
Sample a batch $\{\bm{z}_{b}\}_{1\leq b\leq |\mathcal{B}|}$ from the visitation measure $\sigma_{k}$\;
Each agent $i\in\mathcal{N}$ updates $\theta_{i}(k+1)$ by (\ref{theupdateofpolicyparameter})\;
}
\textbf{Output:} {The joint policy $\{\bm{\pi}_{\bm{\theta}(k)}\}_{k=0}^{K}$}
%{\color{blue}$\bm{\pi}_{\bm{\theta}(\hat{T})}$ with $\hat{T}$ chosen uniformly from $\{0,\cdots,T-1\}$
%}
\caption{Distributed neural policy gradient algorithm}\label{distributedneuralpolicygradientAlgorithm}
\end{algorithm}
\section{Global convergence of distributed critic step}\label{SectionIV}
%This section {\color{blue}establishes the global convergence of the distributed critic step.}
In this section, we establish the global
convergence of agents for the joint policy evaluation in the distributed critic step.
\subsection{Preliminary results}
Before establishing the main results, some useful assumptions and preliminary results are introduced below.
\begin{assumption}\label{theassumptionofreward}
In the NMARL problem, the reward for each agent $i\in\mathcal{N}$ is uniformly bounded, i.e., there exist a constant $R_{0}$ such that $|r_{i}(\bm{z})|\leq R_{0}$ for all $\bm{z}\in\mathcal{S}\times\mathcal{A}$.
\end{assumption}
%{\color{blue}For any joint policy $\bm{\pi}$, define
%%$\varrho_{\bm{\pi}}$ as the stationary state distribution and
%$\nu_{\bm{\pi}}$ as the state visitation measure.}
%The regularity condition for $\nu_{\bm{\pi}}$ is presented below.
\begin{assumption}\label{Regularity Condition}
For any joint policies $\bm{\pi}$ and $\bm{\pi}'$, there exists a constant $c_{0}>0$ such that for all $i\in\mathcal{N}$, $u>0$, and $y\in\mathbb{R}^{d}$,
%the following inequality holds:
\begin{align}
\mathbb{E}_{\bm{\pi}'\cdot\nu_{\bm{\pi}}}\big[\mathds{1}\{|y^{\top}x_{i}|\leq u\}\big]\leq c_{0}\cdot u/\|y\|_{2}
\end{align}
where $\nu_{\bm{\pi}}$ is the state visitation measure of joint policy $\bm{\pi}$.
\end{assumption}
%Assumption~\ref{theassumptionofreward} is commonly used in the NMARL problems with a discount factor. This ensures that the  global $Q$-function is bounded.
%Assumption~\ref{Regularity Condition} essentially imposes a regularity condition on the state transition function $\mathcal{P}$ of the NMARL problem.
\par
%Assumption~\ref{theassumptionofreward} is commonly used in NMARL problems with discount factor, ensuring that the global $Q$-function remains bounded.
%Assumption~\ref{Regularity Condition} primarily enforces a regularity condition on the state transition function $\mathcal{P}$.
Assumption~\ref{theassumptionofreward} suggests that the rewards for all agents are bounded, which is widely adopted in NMARL problems involving a discount factor, guarantees that the global $Q$-function remains bounded.
Assumption~\ref{Regularity Condition} primarily imposes a regularity condition on the state transition function.
%{\color{blue}In what follows, we lay out a regularity condition on the glaobal $Q$-function $Q^{\bm{\pi}}$.}
%\par
%{\color{red}Define a function class
%\begin{align}\label{functionclassinfty}
%\mathcal{F}_{B,\infty}=\Big\{f(\bm{z})=&f_{0}(\bm{z})+\int\mathds{1}\{W^{\top}x>0\}x^{\top}\iota(W)d\mu(W)\notag\\
%&:\|\iota(W)\|_{\infty}\leq B/\sqrt{d}\Big\},
%\end{align}
%where $f_{0}(\bm{z})=\lim_{m\rightarrow\infty}\widehat{Q}\big(\bm{z};W'(0)\big)$, $\mu(\cdot)$ is the density function of distribution $N(0,I_{d}/d)$, and $\iota(\cdot)$ together with $f_{0}(\cdot)$ parameterizes the element of $\mathcal{F}_{B,\infty}$.}
%\begin{assumption}\label{thefunctioninfty}
%%Define a function class
%%\begin{align}\label{functionclassinfty}
%%\mathcal{F}_{B,\infty}=\Big\{h(s,\bm{a})=&h_{0}(s,\bm{a})+\int\mathds{1}\{w^{\top}(s,\bm{a})>0\}(s,\bm{a})^{\top}\cdot\notag\\
%%&\iota(w)d\mu(w):\|\iota(w)\|_{\infty}\leq B/\sqrt{d}\Big\},
%%\end{align}
%%where $h_{0}(s,\bm{a})=h\big(s,\bm{a};\bm{W}(0)\big)$ is the two-layer neural network with the initial parameter $\bm{W}(0)$, $\mu:\mathbb{R}^{d}\rightarrow\mathbb{R}$ is the density function obeying $N(0,I_{d}/d)$, and $\iota:\mathbb{R}^{d}\rightarrow\mathbb{R}^{d}$ together with $h_{0}$ parameterizes the element of $\mathcal{F}_{B,\infty}$.
%In the NMARL problem, for any joint policy $\bm{\pi}$, $Q^{\bm{\pi}}\in\mathcal{F}_{B,\infty}$.
%\end{assumption}
%{\color{blue}explanation of the above assumption....}
\par
%Next, we introduce some facts in this paper.
%Let $w\sim N(0,I_{d}/d)$, some positive constants are defined as
%\begin{align}\label{somepositiveconstants}
%\left\{
%\begin{array}{ll}
%c_{1}=c_{0}\mathbb{E}_{w}[1/\|w\|^{2}_{2}]^{\frac{1}{2}}\\
%c_{2}=c_{0}\big(\mathbb{E}_{w}[\|w\|^{4}_{2}]+\mathbb{E}_{w}[\|w\|^{2}_{2}]^{2}\big)^{\frac{1}{2}}\mathbb{E}_{w}[1/\|w\|^{2}_{2}]^{\frac{1}{2}}\\
%d_{1}=\mathbb{E}_{w}[\|w\|^{2}_{2}].
%\end{array}
%\right.
%\end{align}
%\begin{definition}\label{somepositiveconstants}
%Assume that $w\sim N(0,I_{d}/d)$. Some positive constants are defined as \\
%(i) $c_{1}=c_{0}\mathbb{E}_{w\sim N(0,I_{d}/d)}[1/\|w\|^{2}_{2}]^{1/2}$;\\
%(ii) $c_{2}=c_{0}\big(\mathbb{E}_{w\sim N(0,I_{d}/d)}[\|w\|^{4}_{2}]+\mathbb{E}_{w\sim N(0,I_{d}/d)}[\|w\|^{2}_{2}]^{2}\big)^{1/2}\cdot$ $\mathbb{E}_{w\sim N(0,I_{d}/d)}[1/\|w\|^{2}_{2}]^{1/2}$;\\
%(iii) $d_{1}=\mathbb{E}_{w\sim N(0,I_{d}/d)}[\|w\|^{2}_{2}]$;
%\end{definition}
For any joint policy $\bm{\pi}$, define $\varsigma_{\bm{\pi}}$ as its stationary state-action distribution.
Let $w\sim N(0,I_{d}/d)$,
$c_{1}=c_{0}\mathbb{E}_{w}[1/\|w\|^{2}_{2}]^{\frac{1}{2}}$,
$c_{2}=c_{0}\big(\mathbb{E}_{w}[\|w\|^{4}_{2}]+\mathbb{E}_{w}[\|w\|^{2}_{2}]^{2}\big)^{\frac{1}{2}}\mathbb{E}_{w}[1/\|w\|^{2}_{2}]^{\frac{1}{2}}$, and
$d_{1}=\mathbb{E}_{w}[\|w\|^{2}_{2}]$.
For the sake of convenient representation, we formally define
\begin{align}\label{convenienceofdescription}
\widehat{Q}_{0}(\bm{z})=\widehat{Q}_{0}\big(\bm{z};W'(0)\big).
\end{align}
Based on the above definitions, we can derive the following fact and lemma.
\begin{fact}\label{thefactinpaper}
Suppose Assumption~\ref{Regularity Condition} holds.
For any joint policy $\bm{\pi}$, the following statements (almost everywhere) hold:
\par
(i) $\widehat{Q}_{0}(\bm{z})^{2}\leq(mN)^{-2p}\sum_{r=1}^{mN}\|W'_{r}(0)\|^{2}_{2}$;
\par
(ii) $\mathbb{E}_{\mathrm{init}}\big[\widehat{Q}_{0}(\bm{z})^{2}\big]=(mN)^{1-2p}d_{1}$;
\par
(iii) $\|\nabla_{W'}\widehat{Q}_{0}(\bm{z};W')\|_{2}\leq (mN)^{\frac{1}{2}-p}$;
\par
(iv) $\|\nabla_{W'}\widehat{Q}(\bm{z};W')\|_{2}\leq (mN)^{\frac{1}{2}-p}$;
\par
(v) $\widehat{Q}(\bm{z};W')^{2}\leq2\widehat{Q}_{0}(\bm{z})^{2}+2B^{2}(mN)^{1-2p}$;
\par
(vi) $\widehat{Q}_{0}(\bm{z};W')^{2}\leq2\widehat{Q}_{0}(\bm{z})^{2}+2B^{2}(mN)^{-2p}$.
\end{fact}
\par
%Since the proof of Fact~\ref{thefactinpaper} is easy to obtain, it is omitted here.
The proof of Fact~\ref{thefactinpaper} is straightforward and thus omitted for brevity.
%\begin{proof}
%The proof is easy to obtain and is omitted here.
%%The detailed proof can be found in Appendix~\ref{theproofoffactinpaper}.
%\end{proof}
%{\color{blue}Recall the definition of $\bm{\psi}(\cdot;\cdot)$ in~(\ref{thejointfeaturemapping}), the following lemma can be derived based on Fact~\ref{thefactinpaper}.}
%Based on Fact~\ref{thefactinpaper}, we can derive the following lemma.
\begin{lemma}\label{Lemma1to3}
Suppose Assumption~\ref{Regularity Condition} holds.
For any joint policy $\bm{\pi}$, the following statements hold:
\par
(i) $\mathbb{E}_{\mathrm{init},\varsigma_{\bm{\pi}}}\big[\sum_{r=1}^{mN}\mathds{1}\{|W'^{\top}_{r}(0)x_{\lfloor(r-1)/m\rfloor+1}|\leq\|W'_{r}-W'_{r}(0)\|_{2}\}\big]\leq c_{1}B(mN)^{\frac{1}{2}}$;
\par
(ii) $\mathbb{E}_{\mathrm{init}}\Big[\mathbb{E}_{\varsigma_{\bm{\pi}}}[\widehat{Q}_{0}(\bm{z})^{2}]\mathbb{E}_{\varsigma_{\bm{\pi}}}\big[\sum_{r=1}^{mN}\mathds{1}\{|W'^{\top}_{r}(0)x_{\lfloor(r-1)/m\rfloor+1}|$ $\leq\|W'_{r}-W'_{r}(0)\|_{2}\}\big]\Big]\leq c_{2}B(mN)^{\frac{3}{2}-2p}$;
\par
(iii)
$\mathbb{E}_{\mathrm{init},\varsigma_{\bm{\pi}}}\big[\big|\bm{\psi}^{\top}(\bm{z};W')W''-\bm{\psi}^{\top}\big(\bm{z};W'(0)\big)W''\big|^{2}\big]
\leq4c_{1}B^{3}(mN)^{\frac{1}{2}-2p}$ for all $W',W''\in S^{W}_{B}$.
%In particular, the positive
%constants $c_{1}$ and $c_{2}$ are defined as
%$c_{1}=c_{0}\mathbb{E}_{w\sim N(0,I_{d}/d)}[1/\|w\|^{2}_{2}]^{1/2}$ and $c_{2}=c_{0}\big(\mathbb{E}_{w\sim N(0,I_{d}/d)}[\|w\|^{4}_{2}]+\mathbb{E}_{w\sim N(0,I_{d}/d)}[\|w\|^{2}_{2}]^{2}\big)^{1/2}$ $\mathbb{E}_{w\sim N(0,I_{d}/d)}[1/\|w\|^{2}_{2}]^{1/2}$, respectively.
\end{lemma}
%\begin{proof}
%%Similar proof can be found in~\cite{Cai2019}, therefore omitted here.
%Similar proofs can be found in~\cite{Cai2019}, and are thus omitted here for brevity.
%\end{proof}
\par
Similar proofs can be found in~\cite{Cai2019}, and are thus omitted here for brevity.
\subsection{Global convergence in distributed critic step}
For joint policy $\bm{\pi}_{\bm{\theta}(k)}$ in the $k$-th iteration of policy, define
$\varsigma_{k}$ as the stationary distribution $\bm{z}$ and
$W^{*}_{k}$ as an approximate stationary point of (\ref{thecentralizediteration}) under joint policy $\bm{\pi}_{\bm{\theta}(k)}$.
%the stationary point.
For agent $i\in\mathcal{N}$, define
$\widehat{Q}_{i,t}(\bm{z})=\widehat{Q}_{i}\big(\bm{z};W_{i}(t)\big)$, $\overline{W}(t)=(1/N)\sum_{i=1}^{N}W_{i}(t)$, $\delta^{k}_{0,i}(t)=\widehat{Q}_{0}\big(\bm{z};W_{i}(t)\big)-(1-\gamma)r_{i}(\bm{z})-\gamma\widehat{Q}_{0}\big(\bm{z}';W_{i}(t)\big)$, $\delta^{k}_{0,\mathrm{ave}}(t)=\widehat{Q}_{0}\big(\bm{z};\overline{W}(t)\big)-(1-\gamma)\bar{r}(\bm{z})-\gamma\widehat{Q}_{0}\big(\bm{z}';\overline{W}(t)\big)$, and $\delta^{k*}_{0,\mathrm{ave}}=\widehat{Q}_{0}\big(\bm{z};W_{k}^{*}\big)-(1-\gamma)\bar{r}(\bm{z})-\gamma\widehat{Q}_{0}\big(\bm{z}';W_{k}^{*}\big)$.
Recall the definition of $g^{k}_{i}(t)$ in (\ref{theTDlearningindistributedneuralTDalgorithm1}),
%for joint policy $\bm{\pi}_{\bm{\theta}(k)}$ in the $k$-th iteration of Algorithm~\ref{distributedneuralpolicygradientAlgorithm}},
define
%$\bar{g}^{k}_{i}(t)=\mathbb{E}_{\varsigma^{+}_{k}}[g^{k}_{i}(t)]$, $g^{k}_{\mathrm{ave}}(t)=\frac{1}{N}\sum_{i=1}^{N}g^{k}_{i}(t)$,
\begin{subequations}\label{thedefinitionofstochastic}
\begin{numcases}{}
\bar{g}^{k}_{i}(t)=\mathbb{E}_{\varsigma^{+}_{k}}[g^{k}_{i}(t)]\label{thedefinitionofstochastic-2}\\
g^{k}_{\mathrm{ave}}(t)=\frac{1}{N}\sum_{i=1}^{N}g^{k}_{i}(t)\label{thedefinitionofstochastic-3}\\
\bar{g}^{k}_{\mathrm{ave}}(t)=\frac{1}{N}\sum_{i=1}^{N}\bar{g}^{k}_{i}(t)\label{thedefinitionofstochastic-4}\\
e^{k}_{i}(t)=\sum_{j\in\mathcal{N}_{i}(t)}a_{ij}(t)W_{j}(t)-\eta_{c,t}g^{k}_{i}(t)\notag\\
-P_{S^{W}_{B}}\Big(\sum_{j\in\mathcal{N}_{i}(t)}a_{ij}(t)W_{j}(t)-\eta_{c,t}g^{k}_{i}(t)\Big)\label{thedefinitionofstochastic-5}\\
e^{k}_{\mathrm{ave}}(t)=\frac{1}{N}\sum_{i=1}^{N}e^{k}_{i}(t)\label{thedefinitionofstochastic-6}\\
\bar{g}^{k}_{0,i}(t)=\mathbb{E}_{\varsigma_{k}^{+}}\Big[\delta^{k}_{0,i}(t)\nabla_{W'}\widehat{Q}_{0}\big(\bm{z};W_{i}(t)\big)\Big]\label{thedefinitionofstochastic2-2}\\
\bar{g}^{k}_{0,\mathrm{ave}}(t)=\mathbb{E}_{\varsigma_{k}^{+}}\Big[\delta^{k}_{0,\mathrm{ave}}(t)\nabla_{W'}\widehat{Q}_{0}\big(\bm{z};\overline{W}(t)\big)\Big]\label{thedefinitionofstochastic2-3}\\
\bar{g}^{k*}_{0,\mathrm{ave}}=\mathbb{E}_{\varsigma_{k}^{+}}\Big[\delta^{k*}_{0,\mathrm{ave}}\nabla_{W'}\widehat{Q}_{0}\big(\bm{z};W_{k}^{*}\big)\Big],\label{thedefinitionofstochastic2-1}
\end{numcases}
\end{subequations}
%Specially, the annotations for (\ref{thedefinitionofstochastic-2})-(\ref{thedefinitionofstochastic2-1}) are presented in Table~\ref{tableofmathematicalsymbols}.
where the annotations corresponding to (\ref{thedefinitionofstochastic-2})-(\ref{thedefinitionofstochastic2-1}) are presented in Table~\ref{tableofmathematicalsymbols}.
\begin{table}[!htb]
\begin{center}
\setlength{\abovecaptionskip}{10pt}
\setlength{\belowcaptionskip}{10pt}
\caption{Annotation of symbols.}
\small
\begin{tabular*}{\hsize}{m{0.07\textwidth}|m{0.38\textwidth}}
%{m{0.2\textwidth}<{\centering}m{0.8\textwidth}<{\raggedleft}}
%{@{}@{\extracolsep{\fill}}lllllllllllll}
\toprule[1pt]
\hline
\makecell[c]{Symbols} & \makecell[c]{Annotation of symbols}\\
\hline
\makecell[c]{$g^{k}_{i}(t)$} & Stochastic semigradient of agent $i$\\
\hline
\makecell[c]{$\bar{g}^{k}_{i}(t)$} & Population semigradient of agent $i$\\
\hline
\makecell[c]{$g^{k}_{\mathrm{ave}}(t)$} & Average stochastic semigradient of agents\\
\hline
\makecell[c]{$\bar{g}^{k}_{\mathrm{ave}}(t)$ } & Average population semigradient of agents \\
\hline
\makecell[c]{$e^{k}_{i}(t)$} & Projection error of agent $i$\\
\hline
\makecell[c]{$e^{k}_{\mathrm{ave}}(t)$} & Average projection error of agents\\
\hline
\makecell[c]{$\bar{g}^{k}_{0,i}(t)$} & Locally linearized population semigradient of agent $i$\\
\hline
\makecell[c]{$\bar{g}^{k}_{0,\mathrm{ave}}(t)$} & Locally linearized population semigradient by the average reward $\bar{r}(\bm{z})$ and average approximate parameter $\overline{W}(t)$\\
\hline
\makecell[c]{$\bar{g}_{0,\mathrm{ave}}^{k*}$} & Locally linearized population semigradient by the average reward $\bar{r}(\bm{z})$ and stationary point $W^{*}_{k}$\\
\hline
\bottomrule[1pt]
\end{tabular*}\label{tableofmathematicalsymbols}
\end{center}
\end{table}

\par
To clearly clarify the proof process in this paper, an analytical flowchart is presented in Fig.~\ref{flowchartofproof}
\begin{figure}[!htb]
\centering
\includegraphics[width=1.0\hsize]{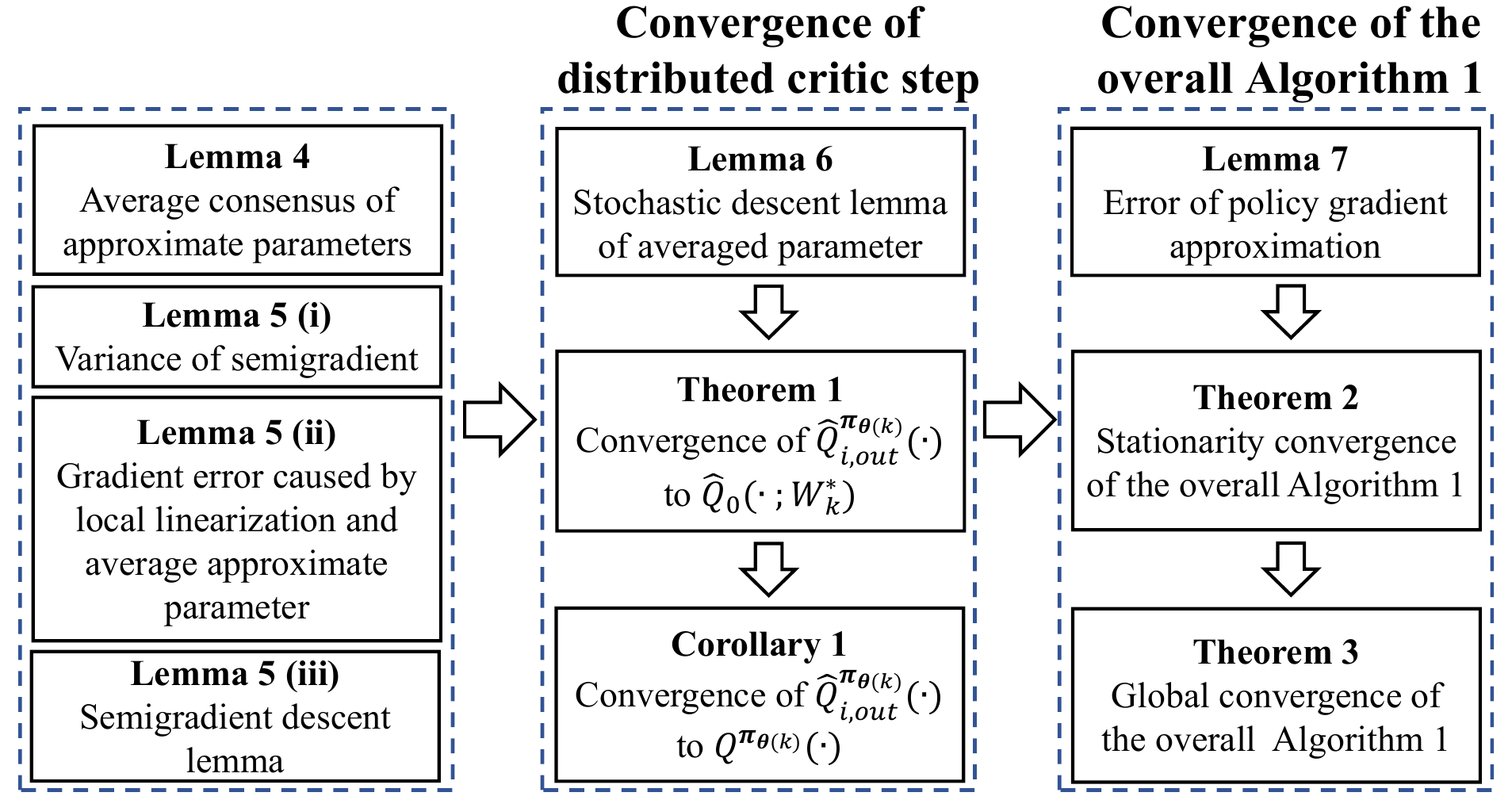}
\caption{The main analytical flowchart of proof process.}\label{flowchartofproof}
\end{figure}
%\par
%{\color{blue}
%In the distributed critic step, the average consensus convergence of the approximate parameter $W_{i}(t)$ to $\overline{W}(t)$ is established in the following.}
\begin{lemma}\label{lemmaofaverageconsensus}
Suppose Assumptions~\ref{theassumptionofcommunication}-\ref{Regularity Condition} hold.
%In the $k$-th iteration of  Algorithm~\ref{distributedneuralpolicygradientAlgorithm},
Let the learning rate $\eta_{c,t}=\frac{1-\gamma}{24\sqrt{t+1}}$ and
$C=I_{N}-(1/N)\mathbf{1}_{N}\mathbf{1}_{N}^{\top}$, then we have that
\begin{align}
&\mathbb{E}_{\mathrm{init}}[\|(C\otimes I_{dmN})W_{tot}(t)\|_{2}]\notag\\
\leq&\frac{\sqrt{md_{1}}N}{\lambda}\lambda^{t}_{D}\!+\!\sum_{t'=0}^{t-1}\frac{(1-\gamma)\lambda^{t-1-t'}_{D}}{24\lambda\sqrt{t'+1}}\mathcal{O}(Bm^{\frac{1}{2}-p}N^{1-p}),
\end{align}
where $W_{tot}(t)=\big(W^{\top}_{1}(t),\cdots,W^{\top}_{N}(t)\big)^{\top}$,  $\lambda=\min\big\{1-1/(2N^{3}),\sup_{t\geq0}\sigma_{2}\big(A(t)\big)\big\}$,
$\sigma_{2}\big(A(t)\big)$ is the second largest singular value of $A(t)$, and $\lambda_{D}=\lambda^{\frac{1}{D}}$.
\end{lemma}
%\begin{proof}
%{\color{blue}The proof can be found in Appendix~\ref{theproofoflemmaofaverageconsensus}.}
%\end{proof}
\par
Lemma~\ref{lemmaofaverageconsensus} establishes the average consensus convergence of $W_{i}(t)$ to $\overline{W}(t)$.
The detailed proof is provided in Appendix~\ref{theproofoflemmaofaverageconsensus}.
By leveraging the result from Lemma~7 in~\cite{nedic2010}, which states that $\lim_{t\rightarrow\infty}\sum_{t'=0}^{t-1}(\lambda^{t-1-t'}_{D}/\sqrt{t'+1})=0$, we further have that $\mathbb{E}_{\mathrm{init}}[\|(C\otimes I_{dmN})W_{\mathrm{tot}}(t)\|_{2}] \to 0$ as $m \to \infty$ and $t \to \infty$.
%and its proof can be found in Appendix~\ref{theproofoflemmaofaverageconsensus}.
%By using the result in Lemma~7 of~\cite{nedic2010} that $\lim_{t\rightarrow\infty}\sum_{t'=0}^{t-1}(\lambda^{t-1-t'}_{D}/\sqrt{t'+1})=0$, we further have $\mathbb{E}_{\mathrm{init}}[\|(C\otimes I_{dmN})W_{\mathrm{tot}}(t)\|_{2}]\rightarrow0$
%as the width of $m\rightarrow \infty$ and $t\rightarrow\infty$.
%\par
%{\color{blue}By the definitions in (\ref{thedefinitionofstochastic}), the gradient analysis is presented as follows.}
\begin{lemma}\label{Lemma4-6}
Suppose Assumptions~\ref{theassumptionofcommunication}-\ref{Regularity Condition} hold.
In the $k$-th ieration of policy in Algorithm~\ref{distributedneuralpolicygradientAlgorithm}, for $t=0,\cdots,T_{c}-1$, we have
\par
%(i) $\mathbb{E}_{\mathrm{init}}\big[\|g^{k}_{i}(t)\|^{2}_{2}\big]\leq3R^{2}_{0}(mN)^{1-2p}+(12d_{1}+12B^{2})(mN)^{2-4p}$;\\
(i) $\mathbb{E}_{\mathrm{init}}\big[\|g^{k}_{i}(t)-\bar{g}^{k}_{i}(t)\|^{2}_{2}\big]\leq12R^{2}_{0}(mN)^{1-2p}+(48d_{1}+48B^{2})(mN)^{2-4p}$;
\par
(ii) $\mathbb{E}_{\mathrm{init}}\big[\|\bar{g}^{k}_{i}(t)-\bar{g}^{k}_{0,\mathrm{ave}}(t)\|^{2}_{2}\big]=\mathcal{O}\big(B^{3}(mN)^{1-2p}\big)$;
\par
(iii) $\mathbb{E}_{\mathrm{init}}\big[\|\bar{g}^{k}_{\mathrm{ave}}(t)\|^{2}_{2}\big]
\leq12\mathbb{E}_{\mathrm{init},\varsigma_{k}}\Big[\Big(\widehat{Q}_{0}\big(\bm{z};\overline{W}(t)\big)-\widehat{Q}_{0}(\bm{z};W^{*}_{k})\Big)^{2}\Big]
+\mathcal{O}\big(B^{3}(mN)^{1-2p}\big)$.
\end{lemma}
%\begin{proof}
%The proof can be found in Appendix~\ref{theproofofLemma4-6}.
%\end{proof}
\par
Lemma~\ref{Lemma4-6} establishes the gradient analysis in the distribiuted critic step, and its proof can be found in Appendix~\ref{theproofofLemma4-6}.
%{\color{blue}By using Lemma~\ref{Lemma4-6}, we establish the stochastic descent lemma of the average parameter $\overline{W}(t)$}.
\begin{lemma}\label{StochasticDescentLemma}
Suppose Assumptions~\ref{theassumptionofcommunication}-\ref{Regularity Condition} hold.
In the $k$-th iteration of policy in Algorithm~\ref{distributedneuralpolicygradientAlgorithm},
the average parameter $\overline{W}(t)$ satisfies
\begin{align}
&\mathbb{E}_{\mathrm{init}}[\|\overline{W}(t+1)-W^{*}_{k}\|^{2}_{2}]\notag\\
\leq&\mathbb{E}_{\mathrm{init}}[\|\overline{W}(t)-W^{*}_{k}\|^{2}_{2}]-\big(2\eta_{c,t}(1-\gamma)-24\eta_{c,t}^{2}\big)\cdot\notag\\
&\mathbb{E}_{\mathrm{init},\varsigma_{k}}\Big[\Big(\widehat{Q}_{0}\big(\bm{z};\overline{W}(t)\big)-\widehat{Q}_{0}(\bm{z};W^{*}_{k})\Big)^{2}\Big]\notag\\
&+\eta_{c,t}\mathcal{O}\big(B^{\frac{5}{2}}(mN)^{\frac{1}{2}-p}\big)+\eta^{2}_{c,t}\mathcal{O}\big(B^{3}(mN)^{1-2p}\big).\notag
\end{align}
\end{lemma}
%\begin{proof}
%The proof can be found in Appendix~\ref{theproofofStochasticDescentLemma}.
%\end{proof}
\par
Lemma~\ref{StochasticDescentLemma} establishes the stochastic descent of the average parameter $\overline{W}(t)$.
The detailed proof can be found in Appendix~\ref{theproofofStochasticDescentLemma}.
%\begin{remark}
%In order to select the appropriate step-size $\eta_{c}$ to ensure the convergence of $\mathbb{E}_{\mathrm{init}}[\|\overline{\bm{W}}(k+1)-\bm{W}^{*}\|]$, we apply (i) and (iii) of Lemma \ref{Lemma4-6} to the scaling of $\mathrm{(i)}$-term $\mathbb{E}_{\mathrm{init}}[\|g_{\mathrm{ave}}(k)\|^{2}_{2}]$ in the proof of Lemma \ref{StochasticDescentLemma}, and get a different conclusion from Lemma \ref{StochasticDescentLemma}, i.e., $\mathbb{E}_{\mathrm{init}}[\|\overline{\bm{W}}(k+1)-\bm{W}^{*}\|^{2}_{2}]\leq\mathbb{E}_{\mathrm{init}}[\|\overline{\bm{W}}(k)-\bm{W}^{*}\|^{2}_{2}]-\big(2\eta_{c}(1-\gamma)-12\eta_{c}^{2}\big)\mathbb{E}_{\mathrm{init},\varsigma_{t}}\Big[\Big(\widehat{Q}_{0}\big(s,\bm{a},\overline{\bm{W}}(k)\big)-\widehat{Q}_{0}(s,\bm{a};\bm{W}^{*})\Big)^{2}\Big]
%+\eta_{c}\mathcal{O}\big(B^{5/2}(Nm)^{1/2-p}\big)+\eta^{2}_{c}\mathcal{O}\big(B^{3}(Nm)^{1-2p}\big)$,
%which shows that with a sufficiently small step-size $\eta_{c}\in\big(0,(1-\gamma)/6\big)$, $\mathbb{E}_{\mathrm{init}}[\|\overline{\bm{W}}(k+1)-\bm{W}^{*}\|]$ decays at each iteration up to the error of local linearization.
%\end{remark}
%\par
Based on results in Lemmas~\ref{lemmaofaverageconsensus}-\ref{StochasticDescentLemma},
%According to the result in Lemma~\ref{StochasticDescentLemma},
the following theorem can be established.
\begin{theorem}\label{ConvergenceofStochasticUpdate}
Suppose Assumptions~\ref{theassumptionofcommunication}-\ref{Regularity Condition} hold.
%In the $k$-th iteration of  Algorithm~\ref{distributedneuralpolicygradientAlgorithm},
Let the learning rate $\eta_{c,t}=\frac{1-\gamma}{24\sqrt{t+1}}$,
the output $\widehat{Q}^{\bm{\pi}_{\bm{\theta}(k)}}_{i,\mathrm{out}}(\bm{z})$ in distributed critic step satisfies
\begin{align}
&\mathbb{E}_{\mathrm{init},\varsigma_{k}}\big[\big(\widehat{Q}^{\bm{\pi}_{\bm{\theta}(k)}}_{i,\mathrm{out}}(\bm{z})-\widehat{Q}_{0}(\bm{z};W^{*}_{k})\big)^{2}\big]\notag\\
\leq&\frac{48B^{2}}{(1-\gamma)\sqrt{T_{c}}}+\frac{1+\ln{T_{c}}}{12\sqrt{T_{c}}}\mathcal{O}\big(B^{3}(mN)^{1-2p}\big)\notag\\
&+\mathcal{O}\big(B^{3}(mN)^{\frac{1}{2}-p}\big).\label{theresultoftheorem1}
\end{align}
%where $\bm{W}^{*}$ is an approximate stationary point and satisfies \begin{align}
%\mathbb{E}_{\varsigma^{+}_{t'}}[\delta_{0}(s,\bm{a},r_{\mathrm{ave}},s',\bm{a}';\bm{W}^{*})\nabla_{\bm{W}}\widehat{Q}_{0}(s,\bm{a};\bm{W}^{*})]^{\top}(\bm{W}-\bm{W}^{*})\geq0,\;\;\mathrm{for\;all}\;\bm{W}\in S_{B}.
%\end{align}
%Here, $\delta_{0}(s,\bm{a},r_{\mathrm{ave}},s',\bm{a}';\bm{W}^{*})=\widehat{Q}_{0}(s,\bm{a};\bm{W}^{*})-(1-\gamma)r_{\mathrm{ave}}-\gamma\widehat{Q}_{0}(s',\bm{a}';\bm{W}^{*})$.
\end{theorem}
\begin{proof}
%Consider $\mathbb{E}_{\mathrm{init},\varsigma_{t}}\Big[\Big(\widehat{Q}_{0}\big(s,\bm{a},\bm{W}_{i}(k)\big)-\widehat{Q}_{0}\big(s,\bm{a};\bm{W}^{*}\big)\Big)^{2}\Big]$, we have
Recall the definition of $\widehat{Q}_{0}(\bm{z};W')$ in (\ref{theQ0xW}), we have
\begin{align}\label{keystepinConvergenceofPopulationUpdate}
&\mathbb{E}_{\mathrm{init},\varsigma_{k}}\Big[\Big(\widehat{Q}_{0}\big(\bm{z};W_{i}(t)\big)-\widehat{Q}_{0}\big(\bm{z};W^{*}_{k}\big)\Big)^{2}\Big]\notag\\
\leq&2\mathbb{E}_{\mathrm{init},\varsigma_{k}}\Big[\Big(\widehat{Q}_{0}\big(\bm{z};W_{i}(t)\big)-\widehat{Q}_{0}\big(\bm{z};\overline{W}(t)\big)\Big)^{2}\Big]\notag\\
&+2\mathbb{E}_{\mathrm{init},\varsigma_{k}}\Big[\Big(\widehat{Q}_{0}\big(\bm{z};\overline{W}(t)\big)-\widehat{Q}_{0}\big(\bm{z};W^{*}_{k}\big)\Big)^{2}\Big]\notag\\
\leq& 2\mathbb{E}_{\mathrm{init},\varsigma_{k}}\Big[\Big(\widehat{Q}_{0}\big(\bm{z};\overline{W}(t)\big)-\widehat{Q}_{0}\big(\bm{z};W^{*}_{k}\big)\Big)^{2}\Big]\notag\\
&+8B^{2}(Nm)^{1-2p},
\end{align}
where the last inequality uses (iii) of Fact~\ref{thefactinpaper} and the fact that $W_{i}(t),\overline{W}(t)\in S^{W}_{B}$.
According to Lemma \ref{StochasticDescentLemma}, we have
\begin{align}\label{changeresultofLemmaPopulationDescentLemma}
&\eta_{c,t}(1-\gamma)\mathbb{E}_{\mathrm{init},\varsigma_{k}}\Big[\Big(\widehat{Q}_{0}\big(\bm{z},\overline{W}(t)\big)-\widehat{Q}_{0}\big(\bm{z};W^{*}_{k}\big)\Big)^{2}\Big]\notag\\
\leq&\mathbb{E}_{\mathrm{init}}[\|\overline{W}(t)-W^{*}_{k}\|^{2}_{2}]-\mathbb{E}_{\mathrm{init}}[\|\overline{W}(t+1)-W^{*}_{k}\|^{2}_{2}]\notag\\
&+\eta_{c,t}\mathcal{O}\big(B^{\frac{5}{2}}(mN)^{\frac{1}{2}-p}\big)+\eta^{2}_{c,t}\mathcal{O}\big(B^{3}(mN)^{1-2p}\big).
\end{align}
By combining (\ref{keystepinConvergenceofPopulationUpdate}) and (\ref{changeresultofLemmaPopulationDescentLemma}), we have
\begin{align}\label{simpleConvergenceofPopulationUpdate}
&\frac{\eta_{c,t}(1-\gamma)}{2}\mathbb{E}_{\mathrm{init},\varsigma_{k}}\Big[\Big(\widehat{Q}_{0}\big(\bm{z};W_{i}(t)\big)-\widehat{Q}_{0}\big(\bm{z};W^{*}_{k}\big)\Big)^{2}\Big]\notag\\
\leq&\mathbb{E}_{\mathrm{init}}[\|\overline{W}(t)-W^{*}_{k}\|^{2}_{2}]-\mathbb{E}_{\mathrm{init}}[\|\overline{W}(t+1)-W^{*}_{k}\|^{2}_{2}]\notag\\
&+\eta_{c,t}\mathcal{O}\big(B^{\frac{5}{2}}(mN)^{\frac{1}{2}-p}+B^{2}(mN)^{1-2p}\big)\notag\\
&+\eta^{2}_{c,t}\mathcal{O}\big(B^{3}(mN)^{1-2p}\big).
\end{align}
Telescoping (\ref{simpleConvergenceofPopulationUpdate}) for $t=0,1,\cdots,T_{c}-1$, we have
\begin{align}\label{TelescopingsimpleConvergenceofPopulationUpdate}
&\frac{1-\gamma}{2}\sum_{t=0}^{T_{c}-1}\eta_{c,t}\mathbb{E}_{\mathrm{init},\varsigma_{k}}\Big[\Big(\widehat{Q}_{0}\big(\bm{z};W_{i}(t)\big)-\widehat{Q}_{0}\big(\bm{z};W^{*}_{k}\big)\Big)^{2}\Big]\notag\\
\leq&\mathbb{E}_{\mathrm{init}}[\|\overline{W}(0)-W^{*}_{k}\|^{2}_{2}]-\mathbb{E}_{\mathrm{init}}[\|\overline{W}(T_{c})-W^{*}_{k}\|^{2}_{2}]\notag\\
&+\sum_{t=0}^{T_{c}-1}\eta_{c,t}\mathcal{O}\big(B^{\frac{5}{2}}(mN)^{\frac{1}{2}-p}\big)\notag\\
&+\sum_{t=0}^{T_{c}-1}\eta^{2}_{c,t}\mathcal{O}\big(B^{3}(mN)^{1-2p}\big).
\end{align}
Dividing (\ref{TelescopingsimpleConvergenceofPopulationUpdate}) by $\frac{1-\gamma}{2}\sum_{t=0}^{T_{c}-1}\eta_{c,t}$ and using the Jensen's inequality, we can have
\begin{align}
&\mathbb{E}_{\mathrm{init},\varsigma_{k}}\big[\big(\widehat{Q}_{0}(\bm{z};W_{i,\mathrm{out}})-\widehat{Q}_{0}(\bm{z};W^{*}_{k})\big)^{2}\big]\notag\\
\leq&\frac{24\big(\mathbb{E}_{\mathrm{init}}[\|\overline{W}(0)-W^{*}_{k}\|^{2}_{2}]-\mathbb{E}_{\mathrm{init}}[\|\overline{W}(T_{c})-W^{*}_{k}\|^{2}_{2}]\big)}{(1-\gamma)^{2}\sqrt{T_{c}}}\notag\\
&+\frac{1+\ln{T_{c}}}{24\sqrt{T_{c}}}\mathcal{O}\big(B^{3}(mN)^{1-2p}\big)+\mathcal{O}\big(B^{\frac{5}{2}}(mN)^{\frac{1}{2}-p}\big),\label{themidresulrintheorem1}
\end{align}
where the inequality comes from the integral test yields in~\cite{Doan2019ICML} that
\begin{gather}
\frac{1-\gamma}{2}\sum_{t=0}^{T_{c}-1}\eta_{t}\geq\frac{(1-\gamma)^{2}\sqrt{T_{c}}}{24},\notag\\ \sum_{t=0}^{T_{c}-1}\eta^{2}_{t}\leq \frac{(1-\gamma)^{2}(1+\ln{T_{c}})}{576}.\notag
\end{gather}
Recall $\widehat{Q}^{\bm{\pi}_{\bm{\theta}(k)}}_{i,\mathrm{out}}(\bm{z})=\widehat{Q}_{i}(\bm{z};W_{i,\mathrm{out}})$ in (\ref{theQiout}), we further have
\begin{align}
&\mathbb{E}_{\mathrm{init},\varsigma_{k}}\big[\big(\widehat{Q}^{\bm{\pi}_{\bm{\theta}(k)}}_{i,\mathrm{out}}(\bm{z})-\widehat{Q}_{0}(\bm{z};W^{*}_{k})\big)^{2}\big]\notag\\
\leq&2\mathbb{E}_{\mathrm{init},\varsigma_{k}}\big[\big(\widehat{Q}_{i}(\bm{z};W_{i,\mathrm{out}})-\widehat{Q}_{0}(\bm{z};W_{i,\mathrm{out}})\big)^{2}\big]\notag\\
&+2\mathbb{E}_{\mathrm{init},\varsigma_{k}}\big[\big(\widehat{Q}_{0}(\bm{z};W_{i,\mathrm{out}})-\widehat{Q}_{0}(\bm{z};W^{*}_{k})\big)^{2}\big]\notag\\
\leq&\frac{48B^{2}}{(1-\gamma)\sqrt{T_{c}}}+\frac{1+\ln{T_{c}}}{12\sqrt{T_{c}}}\mathcal{O}\big(B^{3}(mN)^{1-2p}\big)\notag\\
&+\mathcal{O}\big(B^{3}(mN)^{\frac{1}{2}-p}\big),\notag
\end{align}
where the last inequality can be achieved by (\ref{themidresulrintheorem1}) and (iii) of Lemma~\ref{Lemma1to3}.
%Hence, the proof of Theorem~\ref{ConvergenceofStochasticUpdate} is completed.
\end{proof}
\par
%Theorem~\ref{ConvergenceofStochasticUpdate} establishes the global convergence of in distributed critic step
Theorem~\ref{ConvergenceofStochasticUpdate} indicates that $\widehat{Q}^{\bm{\pi}_{\bm{\theta}(k)}}_{i,\mathrm{out}}(\bm{z})$ converges to $\widehat{Q}_{0}(\bm{z};W^{*}_{k})$ as $T_{c}\rightarrow\infty$ and $m\rightarrow\infty$.
Based on Theorem~\ref{ConvergenceofStochasticUpdate},
we can further quantifies the distance between $\widehat{Q}^{\bm{\pi}_{\bm{\theta}(k)}}_{i,\mathrm{out}}(\bm{z})$ and $Q^{\bm{\pi}_{\bm{\theta}(k)}}(\bm{z})$ in the following corollary.
\begin{corollary}\label{ConvergenceofStochasticUpdatetoQpi}
Suppose Assumptions~\ref{theassumptionofcommunication}-\ref{Regularity Condition} hold.
%In the $k$-th iteration of  Algorithm~\ref{distributedneuralpolicygradientAlgorithm},
Let the learning rate $\eta_{c,t}=\frac{1-\gamma}{24\sqrt{t+1}}$,
the output $\widehat{Q}^{\bm{\pi}_{\bm{\theta}(k)}}_{i,\mathrm{out}}(\bm{z})$ of each agent $i$ satisfies
\begin{align}
&\mathbb{E}_{\mathrm{init},\varsigma_{k}}\big[\big(\widehat{Q}^{\bm{\pi}_{\bm{\theta}(k)}}_{i,\mathrm{out}}(\bm{z})-Q^{\bm{\pi}_{\bm{\theta}(k)}}(\bm{z})\big)^{2}\big]\notag\\
\leq&\frac{96B^{2}}{(1-\gamma)\sqrt{T_{c}}}+\frac{1+\ln{T_{c}}}{6\sqrt{T_{c}}}\mathcal{O}\big(B^{3}(mN)^{1-2p}\big)\notag\\
&+\frac{2\mathbb{E}_{\varsigma_{k}}\big[\big(P_{\mathcal{F}_{B,m}}Q^{\bm{\pi}_{\bm{\theta}(k)}}(\bm{z})-Q^{\bm{\pi}_{\bm{\theta}(k)}}(\bm{z})\big)^{2}\big]}{(1-\gamma)^{2}}\notag\\
&+\mathcal{O}\big(B^{3}(mN)^{\frac{1}{2}-p}\big).\label{theinequalityofcorollary2}
\end{align}
\end{corollary}
\begin{proof}
Similar to the definition of $\mathcal{T}^{\bm{\pi}_{\bm{\theta}}}(\cdot)$ in (\ref{theBellmanevaluationoperatorofpiTheta}),
define $\mathcal{T}^{\bm{\pi}_{\bm{\theta}(k)}}(\cdot)$ as the Bellman evaluation operator of joint policy $\bm{\pi}_{\bm{\theta}(k)}$.
%For $Q^{\bm{\pi}_{\bm{\theta}(t)}}=[Q^{\bm{\pi}_{\bm{\theta}(t)}}(\bm{z})]_{\bm{z}\in\mathcal{S}\times\mathcal{A}}$, we have
%\begin{align}\label{thefixedpointtt}
%{\color{blue}Q^{\bm{\pi}_{\bm{\theta}(t)}}=\mathcal{T}^{\bm{\pi}_{\bm{\theta}}}(Q^{\bm{\pi}_{\bm{\theta}}(t)}).}
%\end{align}
By the property of approximate stationary $W^{*}_{k}$
in Section~\ref{thesectionofApproximatestationary}, we have
\begin{align}
&\|\widehat{Q}_{0}(W^{*}_{k})-Q^{\bm{\pi}_{\bm{\theta}(k)}}\|_{\varsigma_{k}}\notag\\
\leq&\|\widehat{Q}_{0}(W^{*}_{k})-P_{\mathcal{F}_{B,m}}Q^{\bm{\pi}_{\bm{\theta}(k)}}\|_{\varsigma_{k}}\!+\!\|P_{\mathcal{F}_{B,m}}Q^{\bm{\pi}_{\bm{\theta}(k)}}-Q^{\bm{\pi}_{\bm{\theta}(k)}}\|_{\varsigma_{k}}\notag\\
\leq&\|P_{\mathcal{F}_{B,m}}\mathcal{T}^{\bm{\pi}_{\bm{\theta}(k)}}\widehat{Q}_{0}(W^{*}_{k})-P_{\mathcal{F}_{B,m}}\mathcal{T}^{\bm{\pi}_{\bm{\theta}(k)}}Q^{\bm{\pi}_{\bm{\theta}(k)}}\|_{\varsigma_{k}}\notag\\
&+\|P_{\mathcal{F}_{B,m}}Q^{\bm{\pi}_{\bm{\theta}(k)}}-Q^{\bm{\pi}_{\bm{\theta}(k)}}\|_{\varsigma_{k}}\notag\\
\leq&\gamma\|\widehat{Q}_{0}(W^{*}_{k})\!-\!Q^{\bm{\pi}_{\bm{\theta}(k)}}\|_{\varsigma_{k}}\!+\!\|P_{\mathcal{F}_{B,m}}Q^{\bm{\pi}_{\bm{\theta}(k)}}\!-\!Q^{\bm{\pi}_{\bm{\theta}(k)}}\|_{\varsigma_{k}},\label{theinequalityinCorollary2}
\end{align}
where the second inequality follows from the fact that $\mathcal{T}^{\bm{\pi}_{\bm{\theta}(k)}}Q^{\bm{\pi}_{\bm{\theta}(k)}}=Q^{\bm{\pi}_{\bm{\theta}(k)}}$
%and  $\widehat{Q}(\cdot,\cdot;\bm{W}^{*})=\Pi_{\mathcal{F}_{B,m}}\mathcal{T}^{\bm{\pi}_{\bm{\theta}(t)}}\widehat{Q}(\cdot,\cdot;\bm{W}^{*})$ and $Q^{\bm{\pi}_{\bm{\theta}(t)}}=\mathcal{T}^{\bm{\pi}_{\bm{\theta}(t)}}Q^{\bm{\pi}_{\bm{\theta}(t)}}$,
and the last inequality is derived from Lemma 4.2 in~\cite{Cai2019}, which establishes that $P_{\mathcal{F}_{B,m}}\mathcal{T}^{\bm{\pi}_{\bm{\theta}(k)}}$ is a $\gamma$-contraction.
Based on (\ref{theinequalityinCorollary2}), we further have
\begin{align}\label{inequalityofdistance}
&\mathbb{E}_{\varsigma_{k}}[|\widehat{Q}_{0}(\bm{z};W^{*}_{k})-Q^{\bm{\pi}_{\bm{\theta}(k)}}(\bm{z})|^{2}]\notag\\
\leq&\frac{1}{(1-\gamma)^{2}}\mathbb{E}_{\varsigma_{k}}[|P_{\mathcal{F}_{B,m}}Q^{\bm{\pi}_{\bm{\theta}(k)}}(\bm{z})-Q^{\bm{\pi}_{\bm{\theta}(k)}}(\bm{z})|^{2}].
\end{align}
Since
\begin{align}
&\mathbb{E}_{\mathrm{init},\varsigma_{k}}\big[\big(\widehat{Q}^{\bm{\pi}_{\bm{\theta}(k)}}_{i,\mathrm{out}}(\bm{z})-Q^{\bm{\pi}_{\bm{\theta}(k)}}(\bm{z})\big)^{2}\big]\notag\\
\leq&2\mathbb{E}_{\mathrm{init},\varsigma_{k}}\big[\big(\widehat{Q}^{\bm{\pi}_{\bm{\theta}(k)}}_{i,\mathrm{out}}(\bm{z})-\widehat{Q}_{0}(\bm{z};W^{*}_{k})\big)^{2}\big]\notag\\
&+2\mathbb{E}_{\mathrm{init},\varsigma_{k}}\big[\big(\widehat{Q}_{0}(\bm{z};W^{*}_{k})-Q^{\bm{\pi}_{\bm{\theta}(k)}}(\bm{z})\big)^{2}\big],\label{beforethecorollary1}
%\leq&\frac{192B^{2}}{(1-\gamma)^{2}K_{c}}+\frac{2\mathbb{E}_{\varsigma_{t}}\big[\big(\Pi_{\mathcal{F}_{B,m}}Q^{\bm{\pi}_{\bm{\theta}(t)}}(\bm{z})-Q^{\bm{\pi}_{\bm{\theta}(t)}}(\bm{z})\big)^{2}\big]}{(1-\gamma)^{2}}\notag\\
%&+\mathcal{O}\big(B^{3}(mN)^{\frac{1}{2}-p}\big).\notag
\end{align}
substituting (\ref{inequalityofdistance}) and (\ref{theresultoftheorem1}) into (\ref{beforethecorollary1}), we can complete the proof.
\end{proof}
%\begin{remark}
\par
Corollary~\ref{ConvergenceofStochasticUpdatetoQpi} establishes the convergence result of $\widehat{Q}^{\bm{\pi}_{\bm{\theta}(k)}}_{i,\mathrm{out}}$ to the global $Q$-function $Q^{\bm{\pi}_{\bm{\theta}(k)}}(\bm{z})$.
%where $\frac{2\mathbb{E}_{\varsigma_{k}}\big[\big(P_{\mathcal{F}_{B,m}}Q^{\bm{\pi}_{\bm{\theta}(k)}}(\bm{z})-Q^{\bm{\pi}_{\bm{\theta}(k)}}(\bm{z})\big)^{2}\big]}{(1-\gamma)^{2}}$ and $\mathcal{O}\big(B^{3}(mN)^{\frac{1}{2}-p}\big)$ are approximate errors of the neural network.
%\par
Define a function class
\begin{align}\label{functionclassinfty}
\mathcal{F}_{B,\infty}=\Big\{f(\bm{z})=&f_{0}(\bm{z})+\int\mathds{1}\{W^{\top}x>0\}x^{\top}\iota(W)d\mu(W)\notag\\
&:\|\iota(W)\|_{\infty}\leq B/\sqrt{d}\Big\},
\end{align}
where $f_{0}(\bm{z})=\lim_{m\rightarrow\infty}\widehat{Q}\big(\bm{z};W'(0)\big)$, $\mu(\cdot)$ is the density function of distribution $N(0,I_{d}/d)$, and $\iota(\cdot)$ together with $f_{0}(\cdot)$ parameterizes the element of $\mathcal{F}_{B,\infty}$.
If the global $Q$-function $Q^{\bm{\pi}_{\bm{\theta}(k)}}$ satisfies a more restrictive assumption that $Q^{\bm{\pi}_{\bm{\theta}(k)}}\in\mathcal{F}_{B,\infty}$,
then we can further have that $\widehat{Q}^{\bm{\pi}_{\bm{\theta}(k)}}_{i,\mathrm{out}}$ converges to $Q^{\bm{\pi}_{\bm{\theta}(k)}}$ as $m\rightarrow\infty$ and $T_{c}\rightarrow\infty$.
%we have a more refined convergence characterization that
%\begin{align}
%\mathbb{E}_{\mathrm{init},\varsigma_{t}}\big[\big(\widehat{Q}^{\bm{\pi}_{\bm{\theta}(t)}}_{i,\mathrm{out}}(\bm{z})-Q^{\bm{\pi}_{\bm{\theta}(t)}}(\bm{z})\big)^{2}\big]=\mathcal{O}\big(B^{3}(Nm)^{\frac{1}{2}-p}\big).\label{therestrictiveconclusion}
%\end{align}
%\end{remark}
\section{Global Convergence of distributed neural policy gradient algorithm}\label{theconvergenceofoverallalgorithm}
%{\color{blue}Based on the convergence result of the distributed critic step, this section establishes the global convergence of  Algorithm~\ref{distributedneuralpolicygradientAlgorithm}.}
%\par
Before analyzing the global convergence, we introduce the stationarity point of the objective function $J(\bm{\pi}_{\bm{\theta}})$ as below.
\begin{definition}\label{thedefinitionofstationarypoint}
In the NMARL problem, $\widehat{\bm{\theta}}$ is called as the stationary point of $J(\bm{\pi}_{\bm{\theta}})$ in (\ref{thedefinelongtermreturninMARL}), if
\begin{align}\label{thestationarypointinNMARLproblem}
\nabla_{\bm{\theta}}J(\bm{\pi}_{\widehat{\bm{\theta}}})^{\top}(\bm{\theta}-\widehat{\bm{\theta}})\leq0,\forall\bm{\theta}\in S^{\theta}_{B}.
\end{align}
\end{definition}
Based on the stationary point in Definition~\ref{thedefinitionofstationarypoint}, we establish the stationarity convergence of Algorithm~\ref{distributedneuralpolicygradientAlgorithm}
\subsection{Stationarity convergence of Algorithm~\ref{distributedneuralpolicygradientAlgorithm}}
For joint policy $\bm{\pi}_{\bm{\theta}(k)}$ in the $k$-th iteration of policy, define $\nu_{k}$ and $\sigma_{k}$ are the state visitation and the state-action visitation measures, respectively.
Moreover, define $\xi_{i,k}=\widehat{\nabla}_{\theta_{i}}J(\bm{\pi}_{\bm{\theta}(k)})-\mathbb{E}_{\sigma_{k}}[\widehat{\nabla}_{\theta_{i}}J(\bm{\pi}_{\bm{\theta}(k)})]$ and $\xi_{k}=(\xi^{\top}_{1,k},\cdots,\xi^{\top}_{N,k})^{\top}$.
Some useful assumptions of regularity condition of the NMARL problem are introduced below.
\begin{assumption}\label{LipschitzContinuousPolicyGradient}
%For the joint policy sequenece $\{\bm{\pi}_{\bm{\theta}(k)}\}_{k=0}^{K-1}$,
There exists constants $L>0$ and $\sigma_{\xi}>0$ such that $\nabla_{\bm{\theta}}J(\bm{\pi}_{\bm{\theta}})$ is $L$-Lipschitz continuous and $\mathbb{E}_{\sigma_{k}}[\|\xi_{i,k}\|^{2}_{2}]\leq\frac{\sigma^{2}_{\xi}}{|\mathcal{B}|N}$ for all $i\in\mathcal{N}$ and $k=0,\cdots,K-1$.
\end{assumption}
%\begin{assumption}\label{VarianceUpperBoundinpolicygradient}
%{\color{blue}For the joint policy sequenece $\{\bm{\pi}_{\bm{\theta}(k)}\}_{k=0}^{K-1}$,}
%there exists a constant $\sigma_{\xi}>0$ such that $\mathbb{E}_{\sigma_{k}}[\|\xi_{i,k}\|^{2}_{2}]\leq\frac{\sigma^{2}_{\xi}}{|\mathcal{B}|N}$ for all $i\in\mathcal{N}$ and $k=0,\cdots,K-1$.
%\end{assumption}
\par
Assumption~\ref{LipschitzContinuousPolicyGradient} imposes the regularity conditions on the smoothness of the objectice function $J(\bm{\pi}_{\bm{\theta}})$ with respect to $\bm{\theta}$ and the upper bound on variance of $\widehat{\nabla}_{\theta_{i}}J(\bm{\pi}_{\bm{\theta}(k)})$ for all $i\in\mathcal{N}$ during the iteration process.
\begin{assumption}\label{RegularityConditiononsigmaandvarsigma}
For joint policy sequenece $\{\bm{\pi}_{\bm{\theta}(k)}\}_{k=0}^{K-1}$, there exists a constant $\kappa_{a}>0$, such that
\begin{align}\label{theinequalityogRegularityConditiononsigmaandvarsigma}
\Big\{\mathbb{E}_{\varsigma_{k}}\big[\big(\mathrm{d}\sigma_{k}(\bm{z})/\mathrm{d}\varsigma_{k}(\bm{z})\big)^{2}\big]\Big\}^{1/2}\leq\kappa_{a},\forall k=0,\cdots,K-1,
\end{align}
%where $\sigma_{t'}$ and $\varsigma_{t'}$ are the the stateaction
%visitation measure and the stationary distribution of $\bm{\pi}_{\bm{\theta}(t')}$, respectively. Meanwhile,
where $\mathrm{d}\sigma_{k}/\mathrm{d}\varsigma_{k}$ is the Radon-Nikodym derivative of $\sigma_{k}$ with respect to $\varsigma_{k}$.
\end{assumption}
%{\color{red}explain the assumption}
%\par
%{\color{blue}In the NMARL problem, the definition of stationarity point of objective function $J(\bm{\pi}_{\bm{\theta}})$ is introduced below.}
%\begin{definition}\label{thedefinitionofstationarypoint}
%In the NMARL problem, $\widehat{\bm{\theta}}$ is called as the stationary point of $J(\bm{\pi}_{\bm{\theta}})$ in (\ref{thedefinelongtermreturninMARL}), if
%\begin{align}\label{thestationarypointinNMARLproblem}
%\nabla_{\bm{\theta}}J(\bm{\pi}_{\widehat{\bm{\theta}}})^{\top}(\bm{\theta}-\widehat{\bm{\theta}})\leq0,\forall\bm{\theta}\in S^{\theta}_{B}.
%\end{align}
%\end{definition}
%{\color{blue}By the definition of stationary point $\widehat{\bm{\theta}}$ in (\ref{thestationarypointinNMARLproblem}), .}
%\subsection{Stationarity convergence of Algorithm~\ref{distributedneuralpolicygradientAlgorithm}}
\par
Assumption~\ref{RegularityConditiononsigmaandvarsigma} gives a regularity condition on the discrepancy between the stateaction visitation measure and the stationary state-action distribution corresponding to the policy during the iteration process.
Under Assumptions~\ref{LipschitzContinuousPolicyGradient}-\ref{RegularityConditiononsigmaandvarsigma}, we have the following results.
%the policy gradient approximation error can be characterized as follows.
\begin{lemma}\label{thefirstlemmaoftheequalityofLipschitzbefore}
Suppose Assumption~\ref{RegularityConditiononsigmaandvarsigma} holds.
For the joint policy $\bm{\pi}_{\bm{\theta}(k)}$, we have that
\begin{align}\label{thepolicygradientapproximationerror}
&\mathbb{E}[\|\nabla_{\theta_{i}}J(\bm{\pi}_{\bm{\theta}(k)})-\widehat{\nabla}_{\theta_{i}}J(\bm{\pi}_{\bm{\theta}(k)})\|^{2}_{2}]\notag\\
\leq&2\mathbb{E}[\|\xi_{i,k}\|^{2}_{2}]\!+8\kappa^{2}_{a}(mN)^{1-2p}\mathbb{E}[\|Q^{\bm{\pi}_{\bm{\theta}(k)}}\!-\widehat{Q}_{i,\mathrm{out}}^{\bm{\pi}_{\bm{\theta}(k)}}\|^{2}_{\varsigma_{k}}].
\end{align}
\end{lemma}
\par
Lemma~\ref{thefirstlemmaoftheequalityofLipschitzbefore} establishes the policy gradient approximation error in the decentralized actor step, and its proof is presented in Appendix~\ref{theproofofthefirstlemmaoftheequalityofLipschitzbefore}. %and its proof can be found in Appendix~\ref{theproofofthefirstlemmaoftheequalityofLipschitzbefore}.
%\begin{proof}
%The proof can be found in Appendix~\ref{theproofofthefirstlemmaoftheequalityofLipschitzbefore}.
%\end{proof}
%{\color{blue}In the $k$-th iteration of  Algorithm~\ref{distributedneuralpolicygradientAlgorithm},} define $\delta^{\bm{\theta}}_{k}=\big((\delta^{\bm{\theta}}_{1,k})^{\top},\cdots,(\delta^{\bm{\theta}}_{N,k})^{\top}\big)^{\top}$ with the $i$-th element obeying
%\begin{align}
%\delta^{\bm{\theta}}_{i,k}=\big(\theta_{i}(k+1)-\theta_{i}(k)\big)/\eta_{a}.\label{thedeltainactorstep}
%\end{align}
%Then, we have the following lemma.
%\begin{lemma}\label{thefirstlemmaoftheequalityofLipschitz}
%Suppose Assumption~\ref{RegularityConditiononsigmaandvarsigma} holds.
%For the joint policy $\bm{\pi}_{\bm{\theta}(k)}$ in Algorithm~\ref{distributedneuralpolicygradientAlgorithm}, we have
%\begin{align}\label{thelemmainequalityinactorstepadded}
%&\big|\big(\nabla_{\theta_{i}}J(\bm{\pi}_{\bm{\theta}(k)})-\mathbb{E}_{\sigma_{k}}[\widehat{\nabla}_{\theta_{i}}J(\bm{\pi}_{\bm{\theta}(k)})]\big)^{\top}\delta^{\bm{\theta}}_{i,k}\big|\notag\\
%\leq&4\kappa_{a}\eta^{-1}_{a}BN^{-\frac{1}{2}}\|Q^{\bm{\pi}_{\bm{\theta}(k)}}-\widehat{Q}^{\bm{\pi}_{\bm{\theta}(k)}}_{i,\mathrm{out}}\|_{\varsigma_{k}}.
%\end{align}
%\end{lemma}
%\begin{proof}
%The detailed proof can be found in Appendix~\ref{theproofofthefirstlemmaoftheequalityofLipschitz}.
%\end{proof}
\par
Define $\delta^{\bm{\theta}}_{k}=\big((\delta^{\bm{\theta}}_{1,k})^{\top},\cdots,(\delta^{\bm{\theta}}_{N,k})^{\top}\big)^{\top}$ with the $i$-th element obeying
\begin{align}
\delta^{\bm{\theta}}_{i,k}=\big(\theta_{i}(k+1)-\theta_{i}(k)\big)\big/\eta_{a}\label{thedeltainactorstep}
\end{align}
and the gradient mapping of policy parameter as
\begin{align}\label{thedefinnitionofrho}
\bm{\rho}(k)=\big[P_{S^{\theta}_{B}}\big(\bm{\theta}(k)+\eta_{a}\nabla_{\bm{\theta}}J(\bm{\pi}_{\bm{\theta}(k)})\big)-\bm{\theta}(k)\big]\big/\eta_{a}.
\end{align}
Based on the above definitions, the stationarity convergence of Algorithm~\ref{distributedneuralpolicygradientAlgorithm} is established as follows.
\begin{theorem}\label{ConvergencetoStationaryPoint}
Suppose Assumptions~\ref{theassumptionofcommunication}-\ref{RegularityConditiononsigmaandvarsigma} hold.
In the Algorithm~\ref{distributedneuralpolicygradientAlgorithm}, let the learning rate $\eta_{a}\in(0,\frac{1}{2L})$, $K\geq4L^{2}$, and $T_{c}=\mathcal{O}\big((mN)^{2p-1}\big)$.
If $Q^{\bm{\pi}_{\bm{\theta}(k)}}\in\mathcal{F}_{B,\infty}$ for all $k=0,\cdots,K-1$ during the iteration process, then it holds that
\begin{align}\label{theresultinlemmakeyfinal}
\frac{1}{K}\sum_{k=0}^{K-1}\mathbb{E}[\|\bm{\rho}(k)\|^{2}_{2}]\leq&\frac{16R_{0}}{\eta_{a}K}+\frac{4\sigma^{2}_{\xi}}{|\mathcal{B}|}+\tilde{\mathcal{O}}\big(B^{\frac{5}{2}}N^{\frac{1}{2}}T^{-\frac{1}{2}}_{c}\big)\notag\\
&+\mathcal{O}\big(B^{\frac{5}{2}}m^{\frac{1-2p}{4}}N^{\frac{3-2p}{4}}\big).
\end{align}
%{\color{blue}where $\hat{T}$ is uniformly sampled among $\{0,\cdots,T-1\}$.}
\end{theorem}
\begin{proof}
By using Assumption~\ref{LipschitzContinuousPolicyGradient}, we have
\begin{align}
&J(\bm{\pi}_{\bm{\theta}(k+1)})-J(\bm{\pi}_{\bm{\theta}(k)})\notag\\
\geq&\eta_{a}\nabla_{\bm{\theta}}J(\bm{\pi}_{\bm{\theta}(k)})^{\top}\delta^{\bm{\theta}}_{k}-\frac{L}{2}\|\bm{\theta}(k+1)-\bm{\theta}(k)\|^{2}_{2}.\label{theinequalityofLipschitz}
\end{align}
%where $\delta^{\bm{\theta}}_{t}=\big((\delta^{\bm{\theta}}_{1,t})^{\top},\cdots,(\delta^{\bm{\theta}}_{N,t})^{\top}\big)^{\top}$ with $\delta^{\bm{\theta}}_{i,t}=\big(\bm{\theta}_{i}(t+1)-\bm{\theta}_{i}(t)\big)/\eta$.
Recall that $\xi_{i,k}=\widehat{\nabla}_{\theta_{i}}J(\bm{\pi}_{\bm{\theta}(k)})-\mathbb{E}_{\sigma_{k}}[\widehat{\nabla}_{\theta_{i}}J(\bm{\pi}_{\bm{\theta}(k)})]$, we have
\begin{align}\label{theequalityofLipschitz}
\nabla_{\theta_{i}}J(\bm{\pi}_{\bm{\theta}(k)})^{\top}\delta^{\bm{\theta}}_{i,k}=&\!\underbrace{\big(\nabla_{\theta_{i}}J(\bm{\pi}_{\bm{\theta}(k)})\!-\!\mathbb{E}_{\sigma_{k}}[\widehat{\nabla}_{\theta_{i}}J(\bm{\pi}_{\bm{\theta}(k)})]\big)^{\top}\delta^{\bm{\theta}}_{i,k}}_{(\mathrm{i})}\notag\\
&\underbrace{-\xi^{\top}_{i,k}\delta^{\bm{\theta}}_{i,k}}_{(\mathrm{ii})}+\underbrace{\widehat{\nabla}_{\theta_{i}}J(\bm{\pi}_{\bm{\theta}(k)})^{\top}\delta^{\bm{\theta}}_{i,k}}_{(\mathrm{iii})}.
\end{align}
\par
For the $(\mathrm{i})$-th term on the right-hand side of (\ref{theequalityofLipschitz}), we use Proposition~\ref{ThepolicygradienttheoremforMARL} and have
\begin{align}
&\Big|\big(\nabla_{\theta_{i}}J(\bm{\pi}_{\bm{\theta}(k)})-\mathbb{E}_{\sigma_{k}}[\widehat{\nabla}_{\theta_{i}}J(\bm{\pi}_{\bm{\theta}(k)})]\big)^{\top}\delta^{\bm{\theta}}_{i,k}\Big|\notag\\
=&\Big|\mathbb{E}_{\sigma_{k}}\big[\overline{\psi}_{i}(\bm{s},a_{i};\theta_{i})\big(Q^{\bm{\pi}_{\bm{\theta}(k)}}(\bm{z})-\widehat{Q}^{\bm{\pi}_{\bm{\theta}(k)}}_{i,\mathrm{out}}(\bm{z})\big)\big]^{\top}\delta^{\bm{\theta}}_{i,k}\Big|\notag\\
\leq&\|\delta^{\bm{\theta}}_{i,k}\|_{2}\big\|\mathbb{E}_{\sigma_{k}}\big[\overline{\psi}_{i}(\bm{s},a_{i};\theta_{i})\big(Q^{\bm{\pi}_{\bm{\theta}(k)}}(\bm{z})-\widehat{Q}^{\bm{\pi}_{\bm{\theta}(k)}}_{i,\mathrm{out}}(\bm{z})\big)\big]\big\|_{2}\notag\\
\leq&\|\delta^{\bm{\theta}}_{i,k}\|_{2}\mathbb{E}_{\sigma_{k}}[\|\overline{\psi}_{i}(\bm{s},a_{i};\theta_{i})\|_{2}|Q^{\bm{\pi}_{\bm{\theta}(k)}}(\bm{z})-\widehat{Q}^{\bm{\pi}_{\bm{\theta}(k)}}_{i,\mathrm{out}}(\bm{z})|]\notag\\
\leq&4\eta^{-1}_{a}BN^{-\frac{1}{2}}\mathbb{E}_{\sigma_{k}}[|Q^{\bm{\pi}_{\bm{\theta}(k)}}(\bm{z})-\widehat{Q}^{\bm{\pi}_{\bm{\theta}(k)}}_{i,\mathrm{out}}(\bm{z})|]\notag\\
\leq&4\eta^{-1}_{a}BN^{-\frac{1}{2}}\Big\{\mathbb{E}_{\varsigma_{k}}\big[\big(\mathrm{d}\sigma_{k}(\bm{z})/\mathrm{d}\varsigma_{k}(\bm{z})\big)^{2}\big]\Big\}^{1/2}\cdot\notag\\
&\|Q^{\bm{\pi}_{\bm{\theta}(k)}}(\bm{z})-\widehat{Q}^{\bm{\pi}_{\bm{\theta}(k)}}_{i,\mathrm{out}}(\bm{z})\|_{\varsigma_{k}}\notag\\
\leq&4\kappa_{a}\eta^{-1}_{a}BN^{-\frac{1}{2}}\|Q^{\bm{\pi}_{\bm{\theta}(k)}}(\bm{z})-\widehat{Q}^{\bm{\pi}_{\bm{\theta}(k)}}_{i,\mathrm{out}}(\bm{z})\|_{\varsigma_{k}},\label{thei-thtermintheequalityofLipschitz}
\end{align}
where the second inequality follows from the Jensen's inequality, the third inequality uses the fact that $\|\delta^{\bm{\theta}}_{i,t}\|_{2}\leq2\eta^{-1}_{a}BN^{-\frac{1}{2}}$ and $\|\overline{\psi}_{i}(\bm{s},a_{i};\theta_{i})\|_{2}\leq2(mN)^{\frac{1}{2}-p}\leq2$, the forth inequality comes from the Cauchy-Schwartz inequality, and the last inequality uses Assumption~\ref{RegularityConditiononsigmaandvarsigma}.
%\begin{align}\label{thei-thtermintheequalityofLipschitz}
%&\big(\nabla_{\theta_{i}}J(\bm{\pi}_{\bm{\theta}(k)})-\mathbb{E}_{\sigma_{k}}[\widehat{\nabla}_{\theta_{i}}J(\bm{\pi}_{\bm{\theta}(k)})]\big)^{\top}\delta^{\bm{\theta}}_{i,k}\notag\\
%\geq&-4\kappa_{a}\eta^{-1}_{a}BN^{-\frac{1}{2}}\|Q^{\bm{\pi}_{\bm{\theta}(k)}}-\widehat{Q}^{\bm{\pi}_{\bm{\theta}(k)}}_{i,\mathrm{out}}\|_{\varsigma_{k}}.
%\end{align}
\par
For the $(\mathrm{ii})$-th term on the right-hand side of (\ref{theequalityofLipschitz}), we have
\begin{align}\label{theii-thtermintheequalityofLipschitz}
-\xi^{\top}_{i,k}\delta^{\bm{\theta}}_{i,k}\geq-\frac{1}{2}(\|\xi_{i,k}\|^{2}_{2}+\|\delta^{\bm{\theta}}_{i,k}\|^{2}_{2}).
\end{align}
\par
For the $(\mathrm{iii})$-th term on the right-hand side of (\ref{theequalityofLipschitz}), we use (\ref{thedeltainactorstep}) and have
\begin{align}\label{theiii-thtermintheequalityofLipschitz}
&\widehat{\nabla}_{\theta_{i}}J(\bm{\pi}_{\bm{\theta}(k)})^{\top}\delta^{\bm{\theta}}_{i,k}\notag\\
=&\Big[\delta^{\bm{\theta}}_{i,k}-\Big(\theta_{i}(k+1)-\big(\theta_{i}(k)+\eta_{a}\widehat{\nabla}_{\theta_{i}}J(\bm{\pi}_{\bm{\theta}(k)})\big)\Big)\Big/\eta_{a}\Big]\delta^{\bm{\theta}}_{i,k}\notag\\
\geq&\|\delta^{\bm{\theta}}_{i,k}\|^{2}_{2},
\end{align}
where the inequality comes from the property of projection operation that for all $\bm{\theta}',\bm{\theta}''\in S^{\theta}_{B}$,
\begin{align}\label{thepropertyofprojectionoperation}
&\Big(P_{S^{\theta}_{i,B}}\big(\theta_{i}'+\eta_{a}\widehat{\nabla}_{\theta_{i}}J(\bm{\pi}_{\bm{\theta}'})\big)-\big(\theta_{i}'+\eta_{a}\widehat{\nabla}_{\theta_{i}}J(\bm{\pi}_{\bm{\theta}'})\big)\Big)^{\top}\notag\\
&\Big(P_{S^{\theta}_{i,B}}\big(\theta_{i}'+\eta_{a}\widehat{\nabla}_{\theta_{i}}J(\bm{\pi}_{\bm{\theta}'})\big)-\theta_{i}''\Big)\leq0.
\end{align}
Substituting (\ref{thei-thtermintheequalityofLipschitz}), (\ref{theii-thtermintheequalityofLipschitz}), and (\ref{theiii-thtermintheequalityofLipschitz}) into (\ref{theequalityofLipschitz}), we have
\begin{align}\label{thei-ii-iii-equalityofLipschitz}
\nabla_{\theta_{i}}J(\bm{\pi}_{\bm{\theta}(k)})^{\top}\delta^{\bm{\theta}}_{i,k}
\geq&-4\kappa_{a}\eta^{-1}_{a}BN^{-\frac{1}{2}}\|Q^{\bm{\pi}_{\bm{\theta}(k)}}-\widehat{Q}^{\bm{\pi}_{\bm{\theta}(k)}}_{i,\mathrm{out}}\|_{\varsigma_{k}}\notag\\
&-\frac{1}{2}\|\xi_{i,k}\|^{2}_{2}+\frac{1}{2}\|\delta^{\bm{\theta}}_{i,k}\|^{2}_{2}.
\end{align}
By combining (\ref{theinequalityofLipschitz}) and (\ref{thei-ii-iii-equalityofLipschitz}), we have
\begin{align}
\frac{1-L\eta_{a}}{2}\mathbb{E}&[\|\delta^{\bm{\theta}}_{k}\|^{2}_{2}]\leq\frac{1}{\eta_{a}}\mathbb{E}[J(\bm{\pi}_{\bm{\theta}(k+1)})-J(\bm{\pi}_{\bm{\theta}(k)})]\notag\\
&+4\kappa_{a}\eta^{-1}_{a}BN^{-\frac{1}{2}}\sum_{i=1}^{N}\mathbb{E}[\|Q^{\bm{\pi}_{\bm{\theta}(k)}}-\widehat{Q}^{\bm{\pi}_{\bm{\theta}(k)}}_{i,\mathrm{out}}\|_{\varsigma_{k}}]\notag\\
&+\frac{1}{2}\mathbb{E}[\|\xi_{k}\|^{2}_{2}].\label{theresultequalityofLipschitz}
\end{align}
Recall that $0<\eta_{a}<\frac{1}{2L}$ and $K\geq4L^{2}$, we have
\begin{align}\label{theinequalityofrho}
%&\mathbb{E}[\|\bm{\rho}(\hat{K})\|^{2}_{2}]\leq
&\frac{1}{K}\sum_{k=0}^{K-1}\mathbb{E}[\|\bm{\rho}(k)\|^{2}_{2}]\notag\\
\leq&\frac{1}{K}\sum_{k=0}^{K-1}\big(2\mathbb{E}[\|\delta^{\bm{\theta}}_{k}\|^{2}_{2}\|]+2\mathbb{E}[\|\bm{\rho}(k)-\delta^{\bm{\theta}}_{k}\|^{2}_{2}]\big)\notag\\
\leq&\frac{1}{K}\sum_{k=0}^{K-1}\big(4(1-L\eta_{a})\mathbb{E}[\|\delta^{\bm{\theta}}_{k}\|^{2}_{2}\|]+2\mathbb{E}[\|\bm{\rho}(k)-\delta^{\bm{\theta}}_{k}\|^{2}_{2}]\big).
\end{align}
Based on the definitions of $\bm{\rho}(k)$ in (\ref{thedefinnitionofrho}) and $\delta^{\bm{\theta}}_{i,k}$ in (\ref{thedeltainactorstep}), we have
\begin{align}\label{theinequalityoferrorbetweenrhoanddelta}
\|\bm{\rho}(k)-\delta^{\bm{\theta}}_{k}\|_{2}=&\frac{1}{\eta_{a}}\mathbb{E}\big[\big\|P_{S^{\theta}_{B}}\big(\bm{\theta}(k)+\eta_{a}\nabla_{\bm{\theta}}J(\bm{\pi}_{\bm{\theta}(k)})\big)\notag\\
&-P_{S^{\theta}_{B}}\big(\bm{\theta}(k)+\eta_{a}\widehat{\nabla}_{\bm{\theta}}J(\bm{\pi}_{\bm{\theta}(k)})\big)\big]\notag\\
\leq&\|\nabla_{\bm{\theta}}J(\bm{\pi}_{\bm{\theta}(k)})-\widehat{\nabla}_{\bm{\theta}}J(\bm{\pi}_{\bm{\theta}(k)})\|_{2}.
\end{align}
Substituting (\ref{theinequalityoferrorbetweenrhoanddelta}) into (\ref{theinequalityofrho}), we further have
\begin{align}
\frac{1}{K}\sum_{k=0}^{K-1}\mathbb{E}[\|\bm{\rho}(k)\|^{2}_{2}]\leq&\frac{1}{K}\sum_{k=0}^{K-1}\Big(4(1-L\eta_{a})\mathbb{E}[\|\delta^{\bm{\theta}}_{k}\|^{2}_{2}\|]\notag\\
&+2\mathbb{E}[\|\nabla_{\bm{\theta}}J(\bm{\pi}_{\bm{\theta}(k)})-\widehat{\nabla}_{\bm{\theta}}J(\bm{\pi}_{\bm{\theta}(k)})\|^{2}_{2}]\Big)\notag\\
\leq&\frac{8}{K\eta_{a}}\mathbb{E}[J(\bm{\pi}_{\bm{\theta}(K)})-J(\bm{\pi}_{\bm{\theta}(0)})]\notag\\
&+\frac{4}{K}\sum_{k=0}^{K-1}\mathbb{E}\|\xi_{k}\|^{2}_{2}+\tilde{\mathcal{O}}\big(B^{\frac{5}{2}}N^{\frac{1}{2}}T^{-\frac{3}{4}}_{c}\big)\notag\\
&+\mathcal{O}\big(B^{\frac{5}{2}}m^{\frac{1-2p}{4}}N^{\frac{3-2p}{4}}\big)\notag\\
\leq&\frac{16R_{0}}{K\eta_{a}}+\frac{4\sigma^{2}_{\xi}}{|\mathcal{B}|}+\tilde{\mathcal{O}}\big(B^{\frac{5}{2}}N^{\frac{1}{2}}T^{-\frac{3}{4}}_{c}\big)\notag\\
&+\mathcal{O}\big(B^{\frac{5}{2}}m^{\frac{1-2p}{4}}N^{\frac{3-2p}{4}}\big),\notag
\end{align}
where the second inequality follows from (\ref{theresultequalityofLipschitz}) and  Theorem~\ref{ConvergenceofStochasticUpdate}, and the last inequality uses Assumption~\ref{LipschitzContinuousPolicyGradient}.
%In particular, the $\varepsilon_{Q}(T)$ is defined as
%\begin{align}
%\varepsilon_{Q}(T)
%=&(4\kappa\eta^{-1}_{a}BN^{-\frac{1}{2}})\sum_{t=1}^{T}\sum_{i=1}^{N}\mathbb{E}[\|Q^{\bm{\pi}_{\bm{\theta}(t)}}-\widehat{Q}^{\bm{\pi}_{\bm{\theta}(t)}}_{i,\mathrm{out}}\|_{\varsigma_{t}}]\notag\\
%&+16\kappa^{2}(mN)^{1-2p}\sum_{t=1}^{T}\sum_{i=1}^{N}\mathbb{E}[\|Q^{\bm{\pi}_{\bm{\theta}(t)}}-\widehat{Q}^{\bm{\pi}_{\bm{\theta}(t)}}_{i,\mathrm{out}}\|^{2}_{\varsigma_{t}}]\notag\\
%=&\mathcal{O}\big(B^{5/2}(Nm)^{1/4-p/2}T^{-1}\big), \label{thekeyconclusionintheorem2}
%\end{align}
%where (\ref{thekeyconclusionintheorem2}) uses the conclusion (\ref{therestrictiveconclusion}).
%Hence, the proof of Theorem~\ref{ConvergencetoStationaryPoint} is completed.
\end{proof}
\par
Theorem~\ref{ConvergencetoStationaryPoint} indicates that  Algorithm~\ref{distributedneuralpolicygradientAlgorithm} can convergence to a stationary point at a rate of $\mathcal{O}(\frac{1}{K})$ with some error terms.
It is important to note that the terms $\tilde{\mathcal{O}}\big(B^{\frac{5}{2}}N^{\frac{1}{2}}T^{-\frac{3}{4}}_{c}\big)+\mathcal{O}\big(B^{\frac{5}{2}}m^{\frac{1-2p}{4}}N^{\frac{3-2p}{4}}\big)$  on the right side of (\ref{theresultinlemmakeyfinal}) represents the error cased by function approximation in the distributed critic step, while $\frac{4\sigma^{2}_{\xi}}{|\mathcal{B}|}$ denotes the approximation error of policy gradient in the decentralized actor step.
\subsection{Global convergence of Algorithhm~\ref{distributedneuralpolicygradientAlgorithm}}
%{\color{blue}In this subsection, we characterize the global optimality of stationary point and based on this, prove the global convergence of Algorithm~\ref{distributedneuralpolicygradientAlgorithm}.}
\par
Define $\bm{\pi}^{*}$ as the optimal policy and satisfy $J(\bm{\pi}^{*})=\max_{\bm{\pi_{\theta}}}J(\bm{\pi_{\theta}})$.
%\begin{align}\label{theoptimalpolicydefinition}
%J(\bm{\pi}^{*})=\max_{\bm{\pi_{\theta}}}J(\bm{\pi_{\theta}}).
%\end{align}
By recalling the definition of $\bm{\psi}(\bm{z};\bm{\theta})$ in (\ref{thejointfeaturemapping}), for stationary policy $\bm{\pi}_{\bm{\widehat{\theta}}}$, define $\bm{\psi}(\bm{\widehat{\theta}})=[\bm{\psi}(\bm{z};\bm{\widehat{\theta}})]_{\bm{z}\in\mathcal{S}\times\mathcal{A}}$ and $u_{\widehat{\bm{\theta}}}(\bm{z})$ as
\begin{align}\label{thedefinitionofuwidehattheta}
u_{\widehat{\bm{\theta}}}(\bm{z})=\frac{\mathrm{d}\sigma_{\bm{\pi}^{*}}}{\mathrm{d}\sigma_{\bm{\pi}_{\widehat{\bm{\theta}}}}}(\bm{z})-\frac{\mathrm{d}\nu_{\bm{\pi}^{*}}}{\mathrm{d}\nu_{\bm{\pi}_{\widehat{\bm{\theta}}}}}(\bm{s})+\bm{\psi}^{\top}_{\widehat{\bm{\theta}}}\widehat{\bm{\theta}},
\end{align}
where $\sigma_{\bm{\pi}^{*}}$ and $\sigma_{\bm{\pi}_{\widehat{\bm{\theta}}}}$ are the state-action visitation measure of $\bm{\pi}^{*}$ and $\bm{\pi}_{\widehat{\bm{\theta}}}$, respectively.
$\nu_{\bm{\pi}^{*}}$ and $\nu_{\bm{\pi}_{\widehat{\bm{\theta}}}}$ are the stationary state-action distribution of $\bm{\pi}^{*}$ and $\bm{\pi}_{\widehat{\bm{\theta}}}$, respectively.
Moreover, $\mathrm{d}\sigma_{\bm{\pi}^{*}}/\mathrm{d}\sigma_{\bm{\pi}_{\widehat{\bm{\theta}}}}$ and $\mathrm{d}\nu_{\bm{\pi}^{*}}/\mathrm{d}\nu_{\bm{\pi}_{\widehat{\bm{\theta}}}}$ are the Radon-Nikodym derivatives.
%{\color{blue}Based on the definition $u_{\widehat{\bm{\theta}}}(\bm{z})$ in (\ref{thedefinitionofuwidehattheta}),
The global optimality of stationary policy $\bm{\pi}_{\bm{\widehat{\theta}}}$ is characterized in the following lemma.
%\begin{lemma}\label{GlobalOptimalityofStationaryPoint}
%{\color{blue}In the NMARL problem, for stationary policy $\bm{\pi}_{\widehat{\bm{\theta}}}$,} it holds
%\begin{align}\label{theinequalityofGlobalOptimalityofStationaryPoint}
%&(1-\gamma)\big(J(\bm{\pi}^{*})-J(\bm{\pi}_{\widehat{\bm{\theta}}})\big)\notag\\
%\leq&2R_{0}\inf_{\bm{\theta}\in S^{\theta}_{B}}\big\{\|u_{\widehat{\bm{\theta}}}(\cdot,\cdot)-\bm{\psi}^{\top}(\cdot,\cdot;\widehat{\bm{\theta}})\bm{\theta}\|_{\sigma_{\bm{\pi}_{\widehat{\bm{\theta}}}}}\big\}.
%\end{align}
%\end{lemma}
%\begin{proof}
%{\color{blue}Similar proof can be found in Theorem 4.8 in~\cite{Wang2019}, so it is omitted here.}
%\end{proof}
%{\color{red}\begin{lemma}\label{ProjectionErrorofoverlinemathcalFBm}
%Let $\widehat{f}\in\mathcal{F}_{B,\infty}$ with $\widehat{f}\in\mathcal{F}_{B,\infty}$ defined in (\ref{functionclassinfty}). For any $\delta>0$, it holds with probability at least $1-\delta$ that
%\begin{align}\label{theinequalityofProjectionErrorofoverlinemathcalFBm}
%\|\Pi_{\mathcal{F}_{B,m}}\widehat{f}-\widehat{f}\|_{\sigma}\leq B(Nm)^{1/2-p}\big[1+\sqrt{2\log{(1/\delta)}}\big],
%\end{align}
%where $\sigma$ is a state-action visitation
%measure over $\mathcal{S}\times\mathcal{A}$.
%\end{lemma}
%\begin{proof}
%See Lemma B.2 in \cite{Wang2019} for a detailed proof.
%\end{proof}}
%\par
%In the following theorem, we characterize the global optimality of the stationary point $\widehat{\bm{\theta}}$.
\begin{lemma}\label{thetheoremofGlobalOptimalityofStationaryPoint}
Suppose Assumptions~\ref{theassumptionofcommunication}-\ref{RegularityConditiononsigmaandvarsigma} hold.
If the stationary policy $\bm{\pi}_{\widehat{\bm{\theta}}}$ satisfies $u_{\widehat{\bm{\theta}}}\in\mathcal{F}_{B,\infty}$, then it holds that  \begin{align}\label{theresuleoftheoremofGlobalOptimalityofStationaryPoint}
(1-\gamma)\mathbb{E}_{\mathrm{init}}[J(\bm{\pi}^{*})-J(\bm{\pi}_{\widehat{\bm{\theta}}})]=\mathcal{O}\big(B^{\frac{3}{2}}(mN)^{\frac{1}{4}-p}\big).
\end{align}
%where the expectations are taken over the random initialization.
\end{lemma}
\begin{proof}
%Let $u_{\widehat{\bm{\theta}}}=[u_{\widehat{\bm{\theta}}}(\bm{z})]_{\bm{z}\in\mathcal{S}\times\mathcal{A}}$, we use the Theorem~4.8 in~\cite{Wang2019} and have
Let $u_{\widehat{\bm{\theta}}}=[u_{\widehat{\bm{\theta}}}(\bm{z})]_{\bm{z}\in\mathcal{S}\times\mathcal{A}}$.
By invoking Theorem 4.8 from~\cite{Wang2019}, we obtain
\begin{align}
&(1-\gamma)\big(J(\bm{\pi}^{*})-J(\bm{\pi}_{\widehat{\bm{\theta}}})\big)\notag\\
\leq&2R_{0}\inf_{\bm{\theta}\in S^{\theta}_{B}}\big\{\|u_{\widehat{\bm{\theta}}}-\bm{\psi}^{\top}(\widehat{\bm{\theta}})\bm{\theta}\|_{\sigma_{\bm{\pi}_{\widehat{\bm{\theta}}}}}\big\}.\label{thefirstinequalityinTheorem3}
\end{align}
For $\inf_{\bm{\theta}\in S^{\theta}_{B}}\|u_{\widehat{\bm{\theta}}}-\bm{\psi}^{\top}(\widehat{\bm{\theta}})\bm{\theta}\|_{\sigma_{\bm{\pi}_{\widehat{\bm{\theta}}}}}$ in (\ref{thefirstinequalityinTheorem3}), we have
\begin{align}\label{thefirstinequalityoftheoremofGlobalOptimalityofStationaryPoint}
&\inf_{\bm{\theta}\in
S^{\theta}_{B}}\big\{\|u_{\widehat{\bm{\theta}}}-\bm{\psi}^{\top}(\widehat{\bm{\theta}})\bm{\theta}\|_{\sigma_{\bm{\pi}_{\widehat{\bm{\theta}}}}}\big\}\notag\\
\leq&\|u_{\widehat{\bm{\theta}}}-P_{\mathcal{F}_{B,m}}u_{\widehat{\bm{\theta}}}\|_{\sigma_{\bm{\pi}_{\widehat{\bm{\theta}}}}}\notag\\
&+\inf_{\bm{\theta}\in S^{\theta}_{B}}\big\{\|P_{\mathcal{F}_{B,m}}u_{\widehat{\bm{\theta}}}-\bm{\psi}^{\top}(\widehat{\bm{\theta}})\bm{\theta}\|_{\sigma_{\bm{\pi}_{\widehat{\bm{\theta}}}}}\big\}.
\end{align}
Let $\widetilde{\bm{\theta}}\in S^{\theta}_{B}$ and satisfy $P_{\mathcal{F}_{B,m}}u_{\widehat{\bm{\theta}}}=\bm{\psi}^{\top}\big(\bm{\theta}(0)\big)\widetilde{\bm{\theta}}\in\mathcal{F}_{B,m}$, we have
\begin{align}\label{thesecondinequalityoftheoremofGlobalOptimalityofStationaryPoint}
&\inf_{\bm{\theta}\in S^{\theta}_{B}}\big\{\|u_{\widehat{\bm{\theta}}}-\bm{\psi}^{\top}(\widehat{\bm{\theta}})\bm{\theta}\|_{\sigma_{\bm{\pi}_{\widehat{\bm{\theta}}}}}\big\}\notag\\
\leq&\|u_{\widehat{\bm{\theta}}}-P_{\mathcal{F}_{B,m}}u_{\widehat{\bm{\theta}}}\|_{\sigma_{\bm{\pi}_{\widehat{\bm{\theta}}}}}+\|\bm{\psi}^{\top}\big(\bm{\theta}(0)\big)\widetilde{\bm{\theta}}-\bm{\psi}^{\top}(\widehat{\bm{\theta}})\widetilde{\bm{\theta}}\|_{\sigma_{\bm{\pi}_{\widehat{\bm{\theta}}}}}.
\end{align}
%Since $u_{\widehat{\bm{\theta}}}\in\mathcal{F}_{B,\infty}$, similar to Eq.~(B.5) in~\cite{Wang2019}, we have
Given that $u_{\widehat{\bm{\theta}}} \in \mathcal{F}_{B,\infty}$, analogous to Eq.~(B.5) in~\cite{Wang2019}, it follows that
\begin{align}\label{theinequalityoferrorbetweenPiuandu}
\mathbb{E}_{\mathrm{init}}\big[\|P_{\mathcal{F}_{B,m}}u_{\widehat{\bm{\theta}}}-u_{\widehat{\bm{\theta}}}\|^{2}_{\sigma_{\bm{\pi}_{\widehat{\bm{\theta}}}}}\big]
%=&\int_{0}^{\infty}\mathbb{P}\big(\|\Pi_{\mathcal{F}_{B,m}}u_{\widehat{\bm{\theta}}}-u_{\widehat{\bm{\theta}}}\|^{2}_{\sigma_{\bm{\pi}_{\widehat{\bm{\theta}}}}}\geq t\big)\mathrm{d}t\notag\\
%\leq&\int_{0}^{\infty}\exp{\Big(-\frac{1}{2}\big(\frac{t(mN)^{p-\frac{1}{2}}}{B}-1\big)^{2}\Big)}\mathrm{d}t\notag\\
\leq&\mathcal{O}\big(B(mN)^{-p}\big).
\end{align}
According to (iii) of Lemma \ref{Lemma1to3}, we also have
\begin{align}
&\mathbb{E}_{\mathrm{init}}\big[\|\bm{\psi}^{\top}\big(\bm{\theta}(0)\big)\widetilde{\bm{\theta}}-\bm{\psi}^{\top}(\widehat{\bm{\theta}})\widetilde{\bm{\theta}}\|_{\sigma_{\bm{\pi}_{\widehat{\bm{\theta}}}}}\big]\notag\\
\leq&\Big(\mathbb{E}_{\mathrm{init}}\big[\|\bm{\psi}^{\top}\big(\bm{\theta}(0)\big)\widetilde{\bm{\theta}}-\bm{\psi}^{\top}(\widehat{\bm{\theta}})\widetilde{\bm{\theta}}\|^{2}_{\sigma_{\bm{\pi}_{\widehat{\bm{\theta}}}}}\big]\Big)^{1/2}\notag\\
\leq&\mathcal{O}\big(B^{\frac{3}{2}}(mN)^{\frac{1}{4}-p}\big).\label{theinequalityofthelocalerror}
\end{align}
By substituting (\ref{thesecondinequalityoftheoremofGlobalOptimalityofStationaryPoint}) (\ref{theinequalityoferrorbetweenPiuandu}), and (\ref{theinequalityofthelocalerror}) into (\ref{thefirstinequalityinTheorem3}) and subsequently taking the expectation, we have
\begin{align}
&(1-\gamma)\mathbb{E}_{\mathrm{init}}[J(\bm{\pi}^{*})-J(\bm{\pi}_{\widehat{\bm{\theta}}})]\notag\\
\leq&2R_{0}\mathbb{E}_{\mathrm{init}}\Big[\inf_{\bm{\theta}\in S^{\theta}_{B}}\big\{\|u_{\widehat{\bm{\theta}}}-\bm{\psi}^{\top}(\widehat{\bm{\theta}})\bm{\theta}\|_{\sigma_{\bm{\pi}_{\widehat{\bm{\theta}}}}}\big\}\Big]\notag\\
\leq&\mathcal{O}\big(B^{\frac{3}{2}}(mN)^{\frac{1}{4}-p}\big),\notag
\end{align}
which completes the proof of Lemma~\ref{thetheoremofGlobalOptimalityofStationaryPoint}.
\end{proof}
\par
%By invoking Lemma~\ref{thetheoremofGlobalOptimalityofStationaryPoint} and imposing the more stringent assumption that $u_{\bm{\theta}(k)}\in\mathcal{F}_{B,\infty}$, we can derive the following theorem.
By invoking Lemma~\ref{thetheoremofGlobalOptimalityofStationaryPoint} and imposing the additional, more stringent assumption that $u_{\bm{\theta}(k)} \in \mathcal{F}_{B,\infty}$, we can establish the subsequent theorem.
\begin{theorem}\label{thecorollaryofneuralpolicygradientmethod}
Suppose Assumptions~\ref{theassumptionofcommunication}-\ref{RegularityConditiononsigmaandvarsigma} hold.
In the Algorithm~\ref{distributedneuralpolicygradientAlgorithm}, let the learning rate $\eta_{a}\in(0,\frac{1}{2L})$, $K\geq4L^{2}$, and $T_{c}=\mathcal{O}\big((mN)^{2p-1}\big)$.
If $Q^{\bm{\pi}_{\bm{\theta}(k)}},u_{\bm{\theta}(k)}\in\mathcal{F}_{B,\infty}$ for all $k=0,\cdots,K-1$ during the iteration process, then it holds
%the joint policy parameter sequence $\{\bm{\theta}(k)\}^{K-1}_{k=0}$ satisfies
\begin{align}
&\frac{1-\gamma}{K}\sum_{k=0}^{K-1}\mathbb{E}[J(\bm{\pi}^{*})-J(\bm{\pi}_{\bm{\theta}(k)})]\notag\\
\leq&\frac{8(B+\eta_{a}R_{0})\sqrt{R_{0}}}{\sqrt{\eta_{a}K}}+\frac{4(B+\eta_{a}R_{0})\sigma_{\xi}}{\sqrt{|\mathcal{B}|}}+\tilde{\mathcal{O}}\big(B^{\frac{9}{4}}N^{\frac{1}{4}}T^{-\frac{3}{8}}_{c}\big)\notag\\
&+\mathcal{O}\big(B^{\frac{9}{4}}m^{\frac{1-2p}{8}}N^{\frac{3-2p}{8}}\big).\label{thelastinequalityintheorem}
\end{align}
%{\color{blue}where $\hat{K}$ is uniformly sampled among $\{0,\cdots,K-1\}$.}
\end{theorem}
\begin{proof}
%By recalling the definition of $\bm{\rho}(k)$ in (\ref{thedefinnitionofrho}), we use the property of the projection operator and have
By revisiting the definition of $\bm{\rho}(k)$ in (\ref{thedefinnitionofrho}), we leverage the property of the projection operator and have
\begin{align}\label{theprojectionoperateinequality}
\big(\eta_{a}\bm{\rho}(k)\!-\eta_{a}\nabla_{\bm{\theta}}J(\bm{\pi}_{\bm{\theta}(k)})\big)^{\top}\big(\eta_{a}\bm{\rho}(k)\!+\bm{\theta}(k)-\bm{\theta}\big)\leq0,
\end{align}
for all $\bm{\theta}\in S^{\theta}_{B}$.
According to (\ref{theprojectionoperateinequality}), we further have
\begin{align}
&\nabla_{\bm{\theta}}J(\bm{\pi}_{\bm{\theta}(k)})^{\top}\big(\bm{\theta}-\bm{\theta}(k)\big)\notag\\
\leq&\bm{\rho}^{\top}(k)\big(\bm{\theta}-\bm{\theta}(k)\big)-\eta_{a}\|\bm{\rho}(k)\|^{2}_{2}+\eta_{a}\bm{\rho}^{\top}(k)\nabla_{\bm{\theta}}J(\bm{\pi}_{\bm{\theta}(k)})\notag\\
\leq&\|\bm{\rho}(k)\|_{2}\big(\|\bm{\theta}-\bm{\theta}(k)\|_{2}+\eta_{a}\|\nabla_{\bm{\theta}}J(\bm{\pi}_{\bm{\theta}(k)})\|_{2}\big)\notag\\
\leq&\|\bm{\rho}(k)\|_{2}\Big(\|\bm{\theta}-\bm{\theta}(k)\|_{2}\notag\\
&+\eta_{a}\mathbb{E}_{\sigma_{k}}\big[\big|Q^{\bm{\pi}_{\bm{\theta}(k)}}(\bm{z})\big|\big\|\overline{\bm{\psi}}\big(\bm{z};\bm{\theta}(k)\big)\big\|_{2}\big]\Big)\notag\\
\leq&(2B+2\eta_{a}R_{0})\|\bm{\rho}(k)\|_{2},\label{theinequalityofthecorollaryofneuralpolicygradientmethod}
\end{align}
where the third inequality comes from (\ref{theresultofpolicygradientforagenti}) and the last inequality uses $\overline{\bm{\psi}}(\bm{z};\bm{\theta})$ in (\ref{thejointfeaturemapping}) and the fact that $\bm{\theta},\bm{\theta}(t)\in S^{\theta}_{B}$.
Note that the right-hand side of (\ref{theinequalityofthecorollaryofneuralpolicygradientmethod}) quantifies the deviation of $\bm{\theta}(k)$ from a stationary point $\bm{\widehat{\theta}}$.
By leveraging Lemma~\ref{thetheoremofGlobalOptimalityofStationaryPoint} and (\ref{theinequalityofthecorollaryofneuralpolicygradientmethod}), we have
%\par
%{\color{blue}Similar analysis to Theorem 4.9} in~\cite{Wang2019}, we have
\begin{align}
&\frac{1-\gamma}{K}\sum_{k=0}^{K-1}\mathbb{E}[J(\bm{\pi}^{*})-J(\bm{\pi}_{\bm{\theta}(k)})]\notag\\
\leq&\frac{2B+2\eta_{a}R_{0}}{K}\sum_{k=0}^{K-1}\mathbb{E}[\|\bm{\rho}(k)\|_{2}]+\mathcal{O}\big(B^{\frac{3}{2}}(mN)^{\frac{1}{4}-p}\big)\label{beforelastinequality}
%\leq&{\color{blue}\frac{4B\sqrt{R_{0}}}{\sqrt{\eta_{a}T}}+\frac{2\sqrt{2}B\sigma_{\xi}}{\sqrt{|\mathcal{B}|}}+\tilde{\mathcal{O}}\big(B^{\frac{3}{2}}(mN)^{\frac{1}{8}-\frac{p}{4}}\big)},\notag\\
%\leq&\frac{8(B+\eta_{a}R_{0})\sqrt{R_{0}}}{\sqrt{\eta_{a}K}}+\frac{4(B+\eta_{a}R_{0})\sigma_{\xi}}{\sqrt{|\mathcal{B}|}}+\tilde{\mathcal{O}}\big(B^{\frac{9}{4}}N^{\frac{1}{4}}T^{-\frac{1}{4}}_{c}\big)\notag\\
%&+\mathcal{O}\big(B^{\frac{9}{4}}m^{\frac{1-2p}{8}}N^{\frac{3-2p}{8}}\big),\notag
\end{align}
%By combining Theorem~\ref{ConvergencetoStationaryPoint} and (\ref{beforelastinequality}), this theorem can be proved.
By integrating Theorem~\ref{ConvergencetoStationaryPoint} with inequality~(\ref{beforelastinequality}), the proof of this theorem can be rigorously established.
\end{proof}
\par
Theorem~\ref{thecorollaryofneuralpolicygradientmethod} establishes the global convergence of Algorithm~\ref{distributedneuralpolicygradientAlgorithm}.
Note that in the right-hand side of (\ref{thelastinequalityintheorem}), the term $\tilde{\mathcal{O}}\big(B^{\frac{9}{4}}N^{\frac{1}{4}}T^{-\frac{3}{8}}_{c}\big)+\mathcal{O}\big(B^{\frac{9}{4}}m^{\frac{1-2p}{8}}N^{\frac{3-2p}{8}}\big)$ corresponds to the error introduced by function approximation in the distributed critic step, while the term $\frac{4(B+\eta_{a}R_{0})\sigma_{\xi}}{\sqrt{|\mathcal{B}|}}$ represents the policy gradient approximation error incurred in the decentralized actor step.
%Note that in the right hand of (\ref{thelastinequalityintheorem}),  $\tilde{\mathcal{O}}\big(B^{\frac{9}{4}}N^{\frac{1}{4}}T^{-\frac{1}{4}}_{c}\big)+\mathcal{O}\big(B^{\frac{9}{4}}m^{\frac{1-2p}{8}}N^{\frac{3-2p}{8}}\big)$ is the error caused by function approximation in the distributed critic step and $\frac{4(B+\eta_{a}R_{0})\sigma_{\xi}}{\sqrt{|\mathcal{B}|}}$ is the policy gradient approximation error in the decentralized actor step.
\section{Simulation}\label{SectionSimulations}
%In this section, we test our Algorithm~\ref{distributedneuralpolicygradientAlgorithm} in the simulation experiments on robot swarms in different path networks.
In this section, we evaluate the performance of Algorithm~\ref{distributedneuralpolicygradientAlgorithm} through simulation experiments conducted on robots path planning in various path networks.
\subsection{Robots path planning environment}
Similar to the path planning presented in~\cite{Zhou2023}, we
focuses on the path planning for $N$ robots (i.e., agents), where
acyclic path networks with varying structures are illustrated in Fig.~\ref{networkswithdiversestructures}.
In this figure, the ``blue'' nodes $\{b_{s}\}_{s=1,2,3}$ denote the initial positions of the agents, and the ``red'' node indicates the destination.
%Similar to the path planning simulation in~\cite{Zhou2023}, we consider the set of robots (i.e., agents)  as $\mathcal{N}=\{1,\cdots,N\}$.
%The acyclic path networks with different structures {\color{blue}are} depicted in Fig.~\ref{networkswithdiversestructures}, where the nodes $\{b_{s}\}_{s=1,2,3}$ represent the initial locations of agents and the ``red'' node represents the destination location.
%\begin{figure}[!htb]
%\centering
%\includegraphics[width=0.5\hsize]{path_network.pdf}
%\caption{Acyclic path networks {\color{blue}for robots path planning problem}.}\label{networkswithdiversestructures}
%\end{figure}
%
%\begin{figure}[!htb]
%\centering
%\includegraphics[width=0.6\hsize]{communicationnetwork.pdf}
%\caption{{\color{blue}The time-varying communication network of agents}.}\label{communicationnetwork}
%\end{figure}
\begin{figure}[htbp]
    \centering
    \subfigure[Acyclic path networks]{\includegraphics[width=0.245\textwidth]{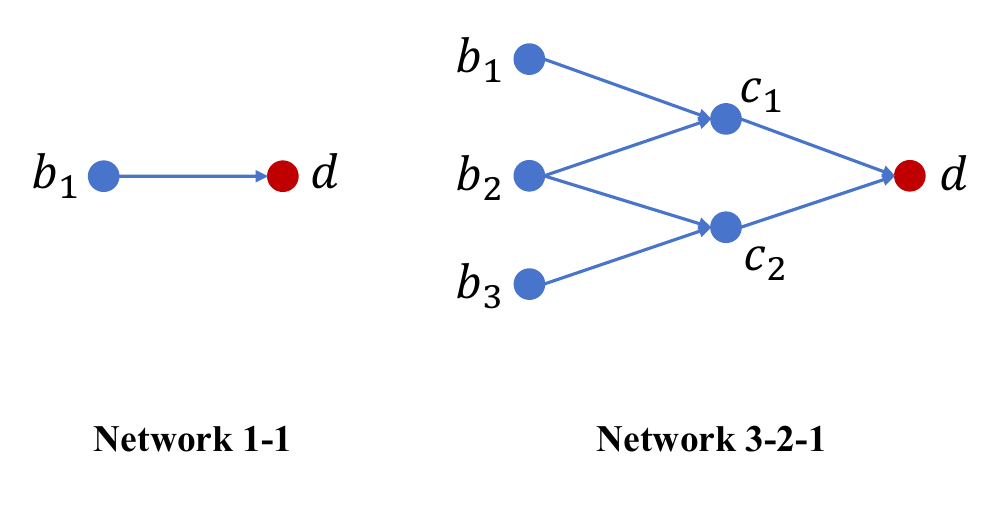}\label{networkswithdiversestructures}}
    \subfigure[Communication networks]{\includegraphics[width=0.2\textwidth]{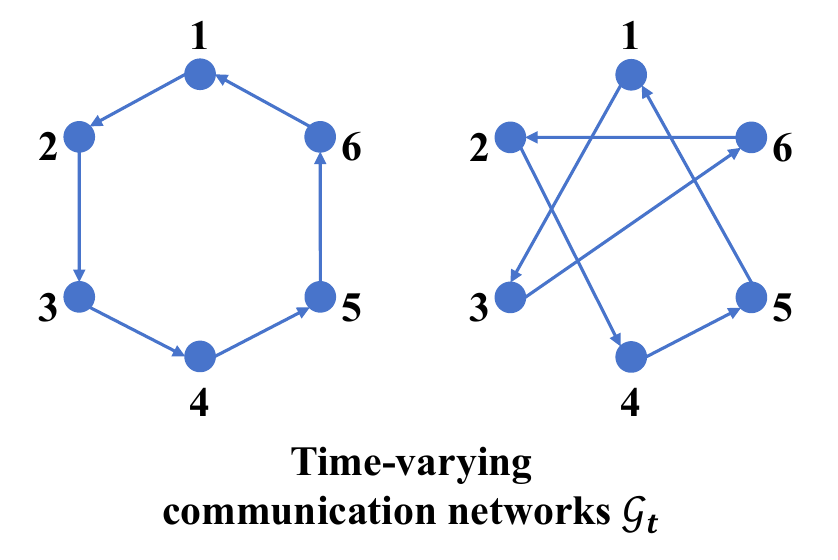}\label{communicationnetwork}}
    \caption{Acyclic path networks and communication networks in the robots path planning problem.}
    \label{networkss}
\end{figure}
\par
In this problem, the movement rules of agents are formally defined as follows.
(i) At each time $t$, every agent $i$ has the option to either traverse along an outward-directed path to reach the next location or remain stationary at its current position;
(ii) Prior to reaching the destination, each agent $i$ incurs a penalty at each time step;
(iii) If agent $i$ shares the same movement path with other agents, it will incur an additional penalty due to the potential risk of collision. The objective of the agents is to find an optimal joint policy that enables them to reach the destination as efficiently as possible while minimizing the likelihood of collisions.

%(ii) prior to reaching the destination, each agent $i$ receives a penalty at each time;
%(iii) if agent $i$ shares the same moving path with other agents, it will receive a penalty of potential collision.
%The objective of agents is to find an optimal joint policy to reach the destination location as quickly as possible while avoiding collisions.
\par
%Based on the above description, we develope a NMARL model, where the state, action, and reward are respectively described as follows.
Based on the aforementioned description, we develop a NMARL model, in which the state, action, and reward are respectively defined as follows.
\par
\textbf{State and action}: For each agent $i$, its local state $s_i$ and local action $a_i$ are formally defined as its location and movement, respectively.
In order to clearly understand the agents' movement dynamics, we illustrate this concept using the path network 3-2-1 depicted in Fig.~\ref{networkswithdiversestructures}.
If agent $i$ is located at node $b_2$, it has the option to remain stationary at the current node for one time step or to transition to the next location via edge $(b_2, c_1)$ or edge $(b_2, c_2)$.
%For each agent $i\in\mathcal{N}$, the local state $s_{i}$ and the local action $a_{i}$ are defined as its location and movement, respectively.
%To better understand the movement change of agents, we take path network~3-2-1 in Fig.~\ref{networkswithdiversestructures} as an example.
%If agent $i$ is at node $b_{2}$, it can choose to remain stationary at the current node for one time step, or choose the edge $(b_{2},c_{1})$ or edge $(b_{2},c_{2})$ to move to the next location.
\par
\textbf{Reward function}:
%Denote $\bm{s}_{t}=(s_{1,t},\cdots,s_{N,t})$ and $\bm{a}_{t}=(a_{1,t},\cdots,a_{N,t})$ as the global state and global action of agents at time $t$, respectively.
%The reward function of agent $i$ before reaching the destination location is defined as
Let $\bm{s}_{t} = (s_{1,t}, \cdots, s_{N,t})$ and $\bm{a}_{t} = (a_{1,t}, \cdots, a_{N,t})$ represent the global state and global action of agents at time $t$, respectively.
The reward function for each agent $i$ prior to reaching the destination location is defined as
\begin{align}\notag
r_{i}(\bm{s}_{t},\bm{a}_{t})=\left\{
\begin{array}{ll}
-r_{\mathrm{cost}},\mathrm{if}~s_{i,t+1}=s_{i,t}~\mathrm{or}~p_{s_{i,t}\rightarrow s_{i,t+1}}=1\\
-r_{\mathrm{cost}}-\frac{p_{s_{i,t}\rightarrow s_{i,t+1}}}{N}r_{\mathrm{collision}},\mathrm{otherwise},
\end{array}
\right.
\end{align}
where $r_{\mathrm{cost}}=0.5$ is the time run cost, $r_{\mathrm{collision}}=0.5$ is the penalty of collision, and $p_{s_{i,t}\rightarrow s_{i,t+1}}$ is the number of agents moving through path $s_{i,t}\rightarrow s_{i,t+1}$.
Specifically, upon reaching the destination location, the reward received by agent $i$ will remain at zero.
%In this NMARL problem, the objective of the agents is to find the optimal joint policy to maximize (\ref{thedefinelongtermreturninMARL}).
%Specially, after reaching the destination location, the reward that agent $i$ receives will remain at 0.
%In the NMARL problem, the objective of agents is to find {\color{blue}the optimal joint policy} to maximize (\ref{thedefinelongtermreturninMARL}).
\subsection{Path network~1-1}
In this subsection, we consider a simple path network 1-1, where the number of agents is $N=6$, the discount factor is $\gamma=0.9$, and all agents are initially positioned at $b_{1}$.
%The time-varying communication networks are shown in Fig.~\ref{communicationnetwork}, where the communication of agents switches back and forth between these two networks.
The time-varying communication networks are illustrated in Fig.~\ref{communicationnetwork}, where the agents' communication alternates between these two networks.
To demonstrate the effectiveness of the proposed Algorithm~\ref{distributedneuralpolicygradientAlgorithm} in terms of convergence performance and optimality, we compare it with the centralized algorithm presented in~\cite{MeiICML2020}, which has been rigorously proven to converge to the global optimal value.
%To demonstrate the effectiveness of the proposed
%Algorithms~\ref{distributedneuralpolicygradientAlgorithm} in terms of the performance of convergence and optimality, we compare it to the centralized algorithm in~\cite{MeiICML2020}, which has been proven to converge to the global optimal value.
Both Algorithm~\ref{distributedneuralpolicygradientAlgorithm} and the centralized algorithm are tested on a hardware device (Intel Xeon Gold 6326 CPU @ 3.50GHz 2.90 GHz).
\par
%The discounted average cumulative rewards (i.e., objective function) $J(\bm{\pi}_{\bm{\theta}(k)})$ and the norm of policy gradient $\|\nabla_{\bm{\theta}}J(\bm{\pi}_{\bm{\theta}(k)})\|_{2}$ generated by centralized algorithm and Algorithm~\ref{distributedneuralpolicygradientAlgorithm} versus the running time are shown in Figs.~\ref{NMARL-ITP1-1}-\ref{NMARL-ITP1-1policygradient}, respectively.
The discounted average cumulative rewards, i.e., the objective function $J(\bm{\pi}_{\bm{\theta}(k)})$ and the norm of the policy gradient $\|\nabla_{\bm{\theta}}J(\bm{\pi}_{\bm{\theta}(k)})\|_{2}$, generated by the centralized algorithm and our Algorithm~\ref{distributedneuralpolicygradientAlgorithm}, are depicted in Figs.~\ref{NMARL-ITP1-1}-\ref{NMARL-ITP1-1policygradient} versus  running time, respectively.
As illustrated in Fig.~\ref{NMARL-ITP1-1}, the convergence of the proposed Algorithm~\ref{distributedneuralpolicygradientAlgorithm} closely approximates the optimal value generated by the centralized algorithm, thereby validating both the convergence and optimality of Algorithm~\ref{distributedneuralpolicygradientAlgorithm}.
Additionally, as another critical indicator of convergence, the norm of the policy gradient for Algorithm~\ref{distributedneuralpolicygradientAlgorithm} is depicted in Fig.~\ref{NMARL-ITP1-1policygradient}, which asymptotically approaches zero.
%As shown in Fig.~\ref{NMARL-ITP1-1}, the convergence of the proposed Algorithm~\ref{distributedneuralpolicygradientAlgorithm} closely approximates the optimal value generated by the centralized algorithm, which confirms the convergence and optimality of Algorithm~\ref{distributedneuralpolicygradientAlgorithm}.
%As another key indicator of convergence, the norm of policy gradient of  Algorithm~\ref{distributedneuralpolicygradientAlgorithm} in Fig.~\ref{NMARL-ITP1-1policygradient} exhibits a near-convergence 0.
It is worth noting that, while the centralized algorithm can achieve the optimal value, it necessitates calculating the exact policy gradient during the learning process.
This undoubtedly results in a significantly longer time per iteration, as demonstrated in Figs.~\ref{NMARL-ITP1-1}--\ref{NMARL-ITP1-1policygradient}.
In contrast, agents in Algorithm~\ref{distributedneuralpolicygradientAlgorithm} utilize approximate policy gradients to update their policy parameters, thereby effectively reducing the computational runtime of the algorithm.
%It is worth noting that although the centralized algorithm can achieve the optimal value, it needs to calculate the exact policy gradient in the learning process, which undoubtedly requires much long time per iteration as evidenced in Figs.~\ref{NMARL-ITP1-1}-\ref{NMARL-ITP1-1policygradient}.
%In contrast, agents in Algorithm~\ref{distributedneuralpolicygradientAlgorithm} use the approximate policy gradients to update their policy parameters, effectively reducing the running time of the algorithm.
\begin{figure}[htbp]
    \centering
    \subfigure[ $J(\bm{\pi}_{\bm{\theta}(k)})$]{\includegraphics[width=0.24\textwidth]{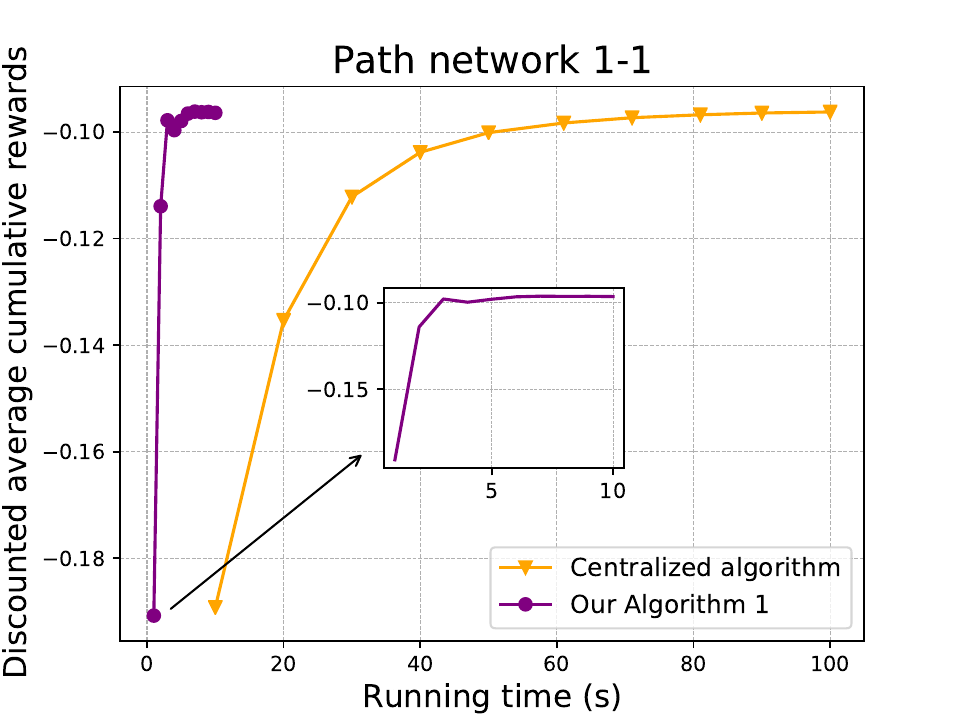}\label{NMARL-ITP1-1}}
    \subfigure[ $\|\nabla_{\bm{\theta}}J(\bm{\pi}_{\bm{\theta}(k)})\|_{2}$]{\includegraphics[width=0.24\textwidth]{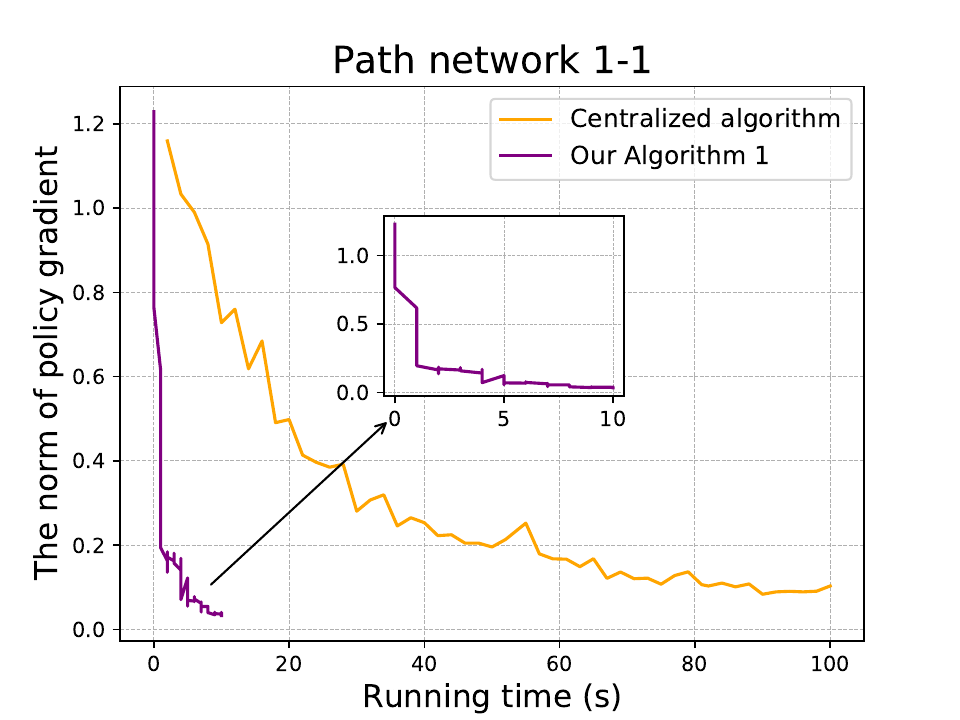}\label{NMARL-ITP1-1policygradient}}
    \caption{Performances of the centralized algorithm and our Algorithm~\ref{distributedneuralpolicygradientAlgorithm} in path network~1-1.}
    \label{network11tot}
\end{figure}
\par
The policies of agents generated by Algorithm~\ref{distributedneuralpolicygradientAlgorithm} at the 0-th, 1000-th, 2000-th, 3000-th, and 4000-th iterations are selected to exhibit the decision-making process regarding agents' movements during training.
As illustrated in Fig.~\ref{NMARL-ITP1-1policy}, with the increase in the number of iterations, the probabilities of agents moving forward at location $b_{1}$ progressively converge toward those learned by the centralized algorithm, which further validates the effectiveness of our Algorithm~\ref{distributedneuralpolicygradientAlgorithm}.
%The policies of agents generated by Algorithm~\ref{distributedneuralpolicygradientAlgorithm} at the 0-th, 1000-th, 2000-th, 3000-th, and 4000-th iterations are selected to illustrate the decisions of agents' movement during training.
%As shown in Fig.~\ref{NMARL-ITP1-1policy}, as the number of iterations increases, the probabilities of agents moving forward at location $b_{1}$ gradually approaches that learned by the centralized algorithm, which further validates the effectiveness of Algorithm~\ref{distributedneuralpolicygradientAlgorithm}.
%\begin{figure}[!htb]
%\centering
%\includegraphics[width=0.6\hsize]{11objective_dis_VS_cen.pdf}
%\caption{{\color{blue}The evolution of the objective performance $J(\bm{\pi}_{\bm{\theta}(k)})$ generated by centralized algorithm and Algorithm~\ref{distributedneuralpolicygradientAlgorithm} in path network~1-1.}}\label{NMARL-ITP1-1}
%\end{figure}
%\begin{figure}[!htb]
%\centering
%\includegraphics[width=0.6\hsize]{11gradient_dis_VS_cen.pdf}
%\caption{{\color{blue}The evolution of the norm of policy gradient generated by centralized algorithm and Algorithm~\ref{distributedneuralpolicygradientAlgorithm} in path network~1-1.}}\label{NMARL-ITP1-1policygradient}
%\end{figure}
\begin{figure}[!htb]
\centering
\includegraphics[width=0.6\hsize]{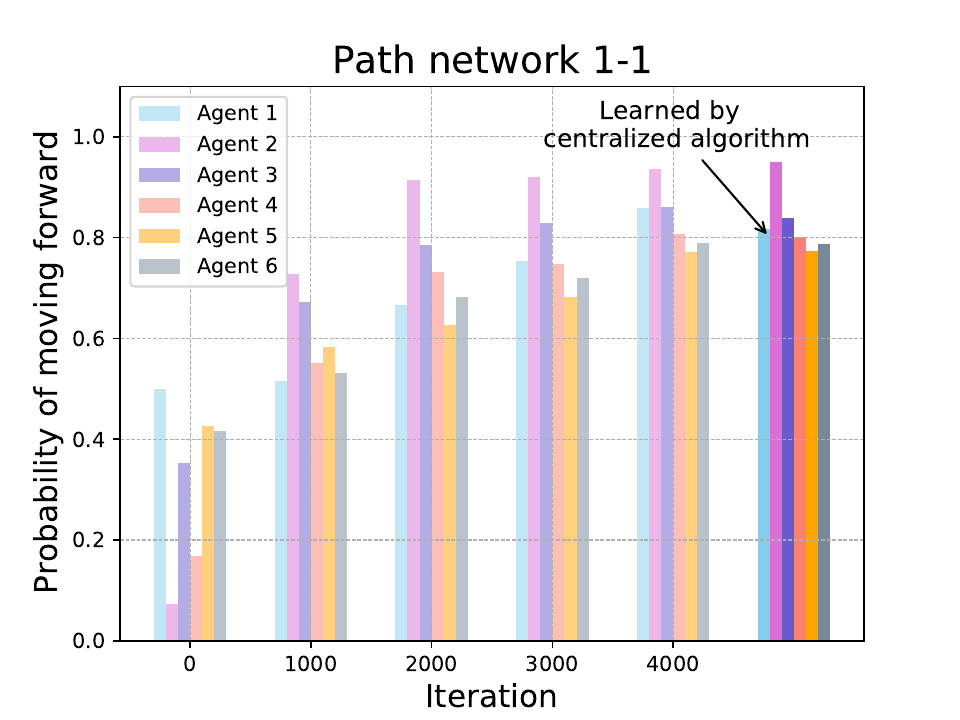}
\caption{Policies of agents at location $b_{1}$ generated by Algorithm~\ref{distributedneuralpolicygradientAlgorithm} in path network~1-1.}\label{NMARL-ITP1-1policy}
\end{figure}
\par
To evaluate the convergence performance of Algorithm~\ref{distributedneuralpolicygradientAlgorithm} under varying conditions, an ablation study is conducted by adjusting the size of the sample batch $\mathcal{B}$.
%By setting different $|\mathcal{B}|=4, 6, 8$,
The discounted average cumulative rewards $J(\bm{\pi}_{\bm{\theta}(k)})$ and the norm of the policy gradient $\|\nabla_{\bm{\theta}}J(\bm{\pi}_{\bm{\theta}(k)})\|_{2}$ generated by Algorithm~\ref{distributedneuralpolicygradientAlgorithm} with $|\mathcal{B}|=4, 6, 8$ are presented in Figs.~\ref{NMARL-ITP1-1changeB}-\ref{NMARL-ITP1-1policygradientchangeB}.
These results clearly indicate that increasing the size of the sample batch $\mathcal{B}$ significantly enhances both the objective function and the policy gradient of Algorithm~\ref{distributedneuralpolicygradientAlgorithm}.
%In order to test the convergence performance of Algorithm~\ref{distributedneuralpolicygradientAlgorithm} under different variables, an ablation experiment with varying size of sample batch $\mathcal{B}$ is designed.
%The discounted average cumulative rewards $J(\bm{\pi}_{\bm{\theta}(k)})$ and the norm of the policy gradient $\|\nabla_{\bm{\theta}}J(\bm{\pi}_{\bm{\theta}(k)})\|_{2}$ generated by Algorithm~\ref{distributedneuralpolicygradientAlgorithm} with $|\mathcal{B}|=4,6,8$ versus the iteration are illustrated in Figs.~\ref{NMARL-ITP1-1changeB}-\ref{NMARL-ITP1-1policygradientchangeB}, which clearly demonstrate that increasing the size of sample batch $\mathcal{B}$ leads to a noticeable improvement in both the objective function and policy gradient of Algorithm~\ref{distributedneuralpolicygradientAlgorithm}.
\begin{figure}[htbp]
    \centering
    \subfigure[ $J(\bm{\pi}_{\bm{\theta}(k)})$]{\includegraphics[width=0.24\textwidth]{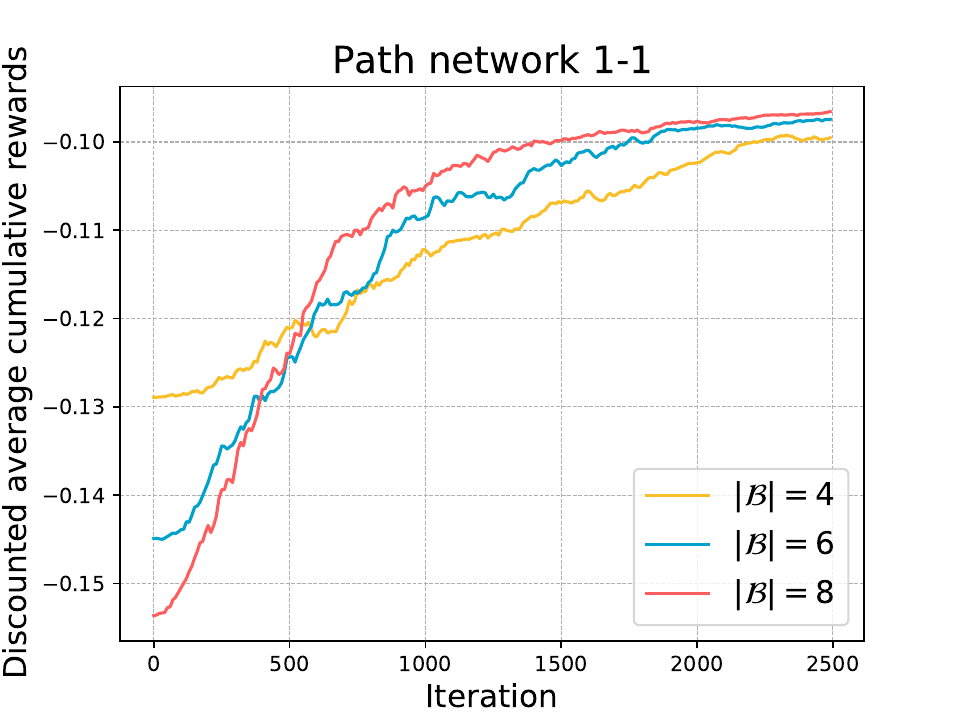}\label{NMARL-ITP1-1changeB}}
    \subfigure[ $\|\nabla_{\bm{\theta}}J(\bm{\pi}_{\bm{\theta}(k)})\|_{2}$]{\includegraphics[width=0.24\textwidth]{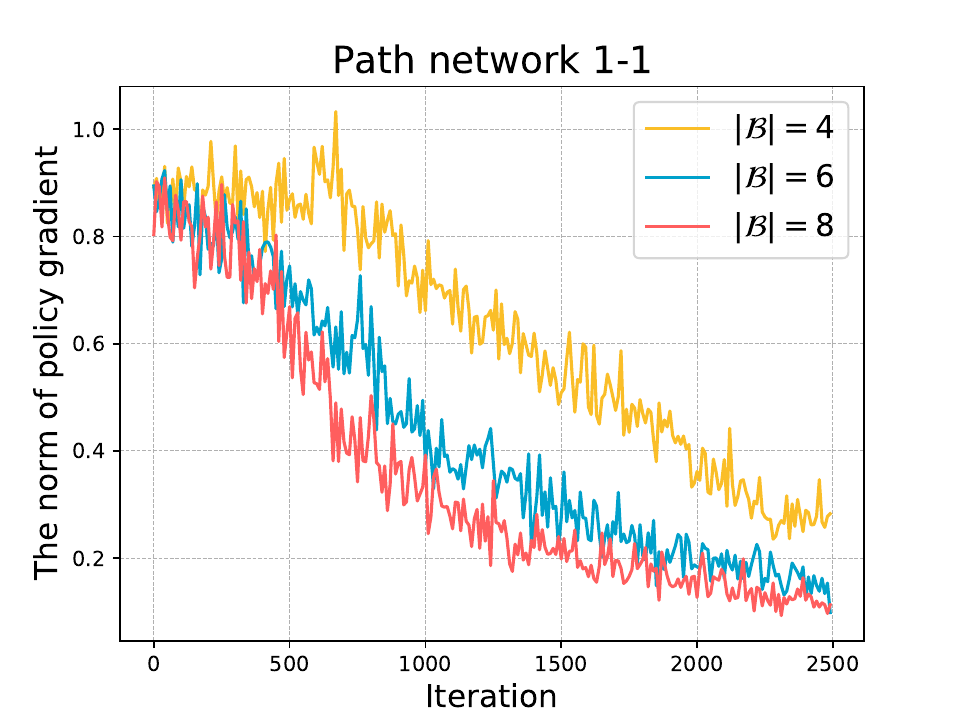}\label{NMARL-ITP1-1policygradientchangeB}}
    \caption{Performances of Algorithm~\ref{distributedneuralpolicygradientAlgorithm} with $|\mathcal{B}|=4,6,8$ in path network~1-1.}
    \label{network1changeBtot}
\end{figure}
%\begin{figure}[!htb]
%\centering
%\includegraphics[width=0.6\hsize]{11_objective_dis_H_change.pdf}
%\caption{{\color{blue}The evolution of the objective performance $J(\bm{\pi}_{\bm{\theta}(k)})$ generated by Algorithm~\ref{distributedneuralpolicygradientAlgorithm} with $|\mathcal{B}|=4,6,8$ in path network~1-1.}}\label{NMARL-ITP1-1changeB}
%\end{figure}
%\begin{figure}[!htb]
%\centering
%\includegraphics[width=0.6\hsize]{11_gradient_dis_H_change.pdf}
%\caption{{\color{blue}The evolution of the norm of policy gradient generated by Algorithm~\ref{distributedneuralpolicygradientAlgorithm} with $|\mathcal{B}|=4,6,8$ in path network~1-1.}}\label{NMARL-ITP1-1policygradientchangeB}
%\end{figure}
\subsection{Path network~3-2-1}
%Consider the path planning problem
%on the path network~3-2-1, where the number of agents $N=6$, the discount factor $\gamma=0.9$, and initial positions are $b_{1}, b_{2}, b_{3},b_{1}, b_{2}, b_{3}$.
%The time-varying communication networks of agents are illustrated in Fig.~\ref{communicationnetwork}.
This subsection considers the path planning problem on path network~3-2-1, where the number of agents is $N=6$, the discount factor is $\gamma=0.9$, and the initial positions are $b_{1}, b_{2}, b_{3}, b_{1}, b_{2}, b_{3}$.
The time-varying communication networks of the agents are depicted in Fig.~\ref{communicationnetwork}.
\par
%The performances of our Algorithm~\ref{distributedneuralpolicygradientAlgorithm} and the centralized algorithm in terms of the objective function $J(\bm{\pi}_{\bm{\theta}(k)})$ in (\ref{thedefinelongtermreturninMARL}) and the norm of policy gradient $\nabla_{\bm{\theta}}J(\bm{\pi}_{\bm{\theta}(k)})$ are shown in Figs.~\ref{NMARL-ITP3-2-1}-\ref{NMARL-ITP3-2-1policygradient}, respectively.
The performance comparisons between our Algorithm~\ref{distributedneuralpolicygradientAlgorithm} and the centralized algorithm are illustrated in Figs.~\ref{NMARL-ITP3-2-1}-\ref{NMARL-ITP3-2-1policygradient}, which respectively depict the values of the objective function $J(\bm{\pi}_{\bm{\theta}(k)})$ and the norm of the policy gradient $\nabla_{\bm{\theta}}J(\bm{\pi}_{\bm{\theta}(k)})$.
Similar to the results presented in Fig.~\ref{network11tot}, Fig.~\ref{network321tot} indicates that Algorithm~\ref{distributedneuralpolicygradientAlgorithm} achieves a close approximation to the performance of the centralized algorithm in terms of both the objective function and the policy gradient norm.
Furthermore, our Algorithm~\ref{distributedneuralpolicygradientAlgorithm} substantially reduces the running time compared to the centralized algorithm.
%Similar to the results illustrated in Fig.~\ref{network11tot}, Fig.~\ref{network321tot} demonstrates that our Algorithm~\ref{distributedneuralpolicygradientAlgorithm} also closely approximates the performance of the centralized algorithm with respect to both objective function and policy gradient norm, while significantly reducing running time compared to the centralized algorithm.
\begin{figure}[htbp]
    \centering
    \subfigure[ $J(\bm{\pi}_{\bm{\theta}(k)})$]{\includegraphics[width=0.24\textwidth]{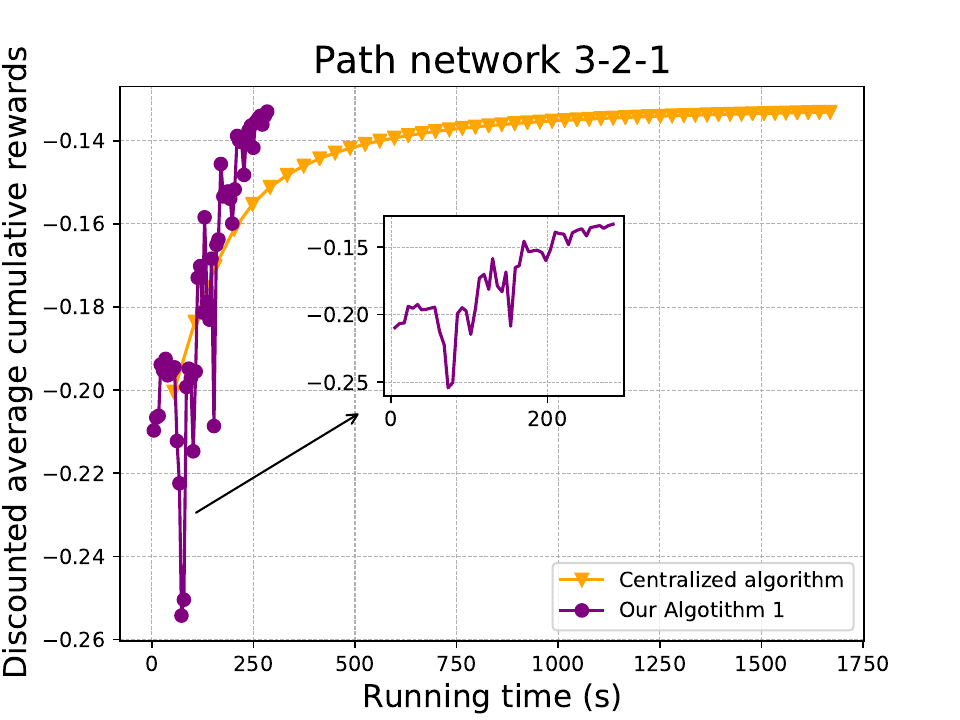}\label{NMARL-ITP3-2-1}}
    \subfigure[ $\|\nabla_{\bm{\theta}}J(\bm{\pi}_{\bm{\theta}(k)})\|_{2}$]{\includegraphics[width=0.24\textwidth]{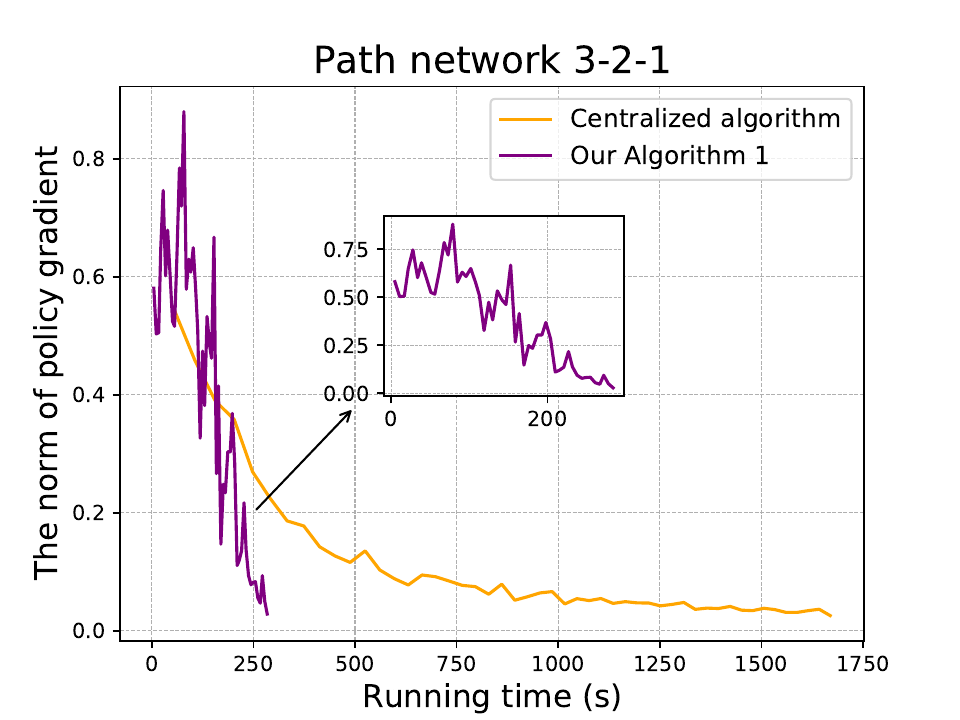}\label{NMARL-ITP3-2-1policygradient}}
    \caption{Performances of the centralized algorithm and our Algorithm~\ref{distributedneuralpolicygradientAlgorithm} in path network~3-2-1.}
    \label{network321tot}
\end{figure}
\par
The policies of selected agents~1, 3, 4, and 6 at their initial locations, generated by Algorithm~\ref{distributedneuralpolicygradientAlgorithm} and the centralized algorithm, are illustrated in Fig.~\ref{NMARL-ITP3-2-1policy}.
It can be observed that the probabilities of forward movement for these agents gradually increase and approach 1, attributed to the presence of only two agents at each initial location and relatively low collision costs.
Additionally, as the iterations progress, the probabilities of forward movement for these agents generated by our Algorithm~\ref{distributedneuralpolicygradientAlgorithm} increasingly converge with those produced by the centralized algorithm.
%The policies of agents~1, 3, 4, and 6 at their initial locations generated by Algorithm~\ref{distributedneuralpolicygradientAlgorithm} and centralized algorithm are depicted in Fig.~\ref{NMARL-ITP3-2-1policy}.
%It can be observed that the probabilities of forward movement for these agents gradually increases and approaches 1 due to the presence of only two agents at each initial locations and relatively low collision costs.
%Furthermore, as iterations progress, the probabilities of forward movement for these agents closely align with those generated by the centralized algorithm.
\begin{figure}[!htb]
\centering
\includegraphics[width=0.6\hsize]{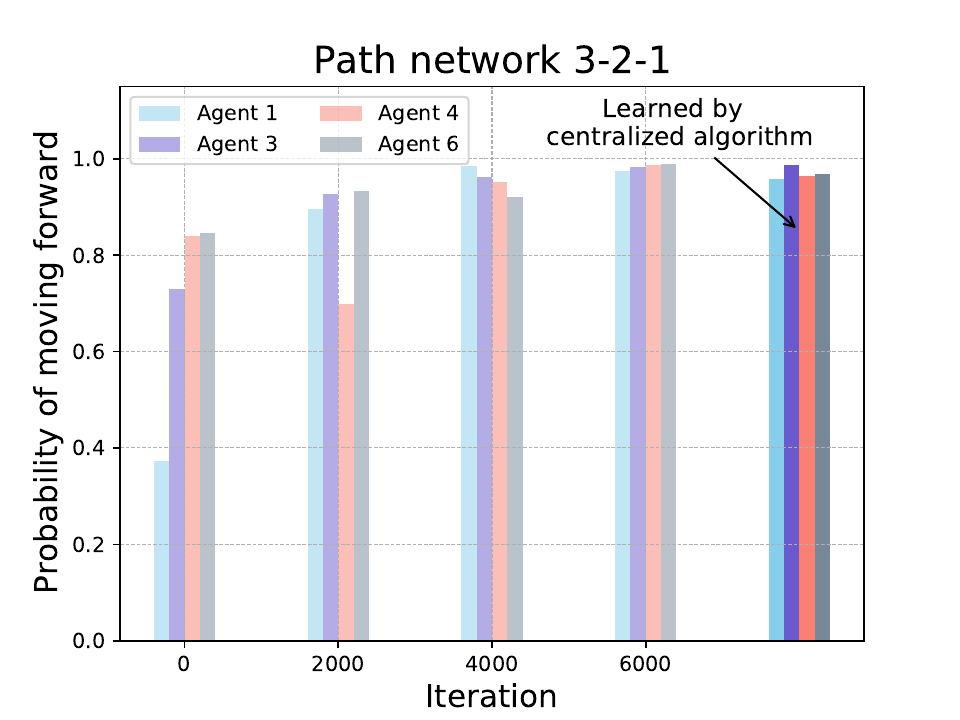}
\caption{Policies of agent~1, 3, 4, and 6 at initial location generated by Algorithm~\ref{distributedneuralpolicygradientAlgorithm} in path network~3-2-1.}\label{NMARL-ITP3-2-1policy}
\end{figure}
\par
%In the path network 3-2-1, the influence of varying sample batch sizes on Algorithm~\ref{distributedneuralpolicygradientAlgorithm} is depicted in Fig.~\ref{network321changeBtot}.
%The results indicate that an increase in the sample batch size results in a gradual enhancement of both the objective function and the policy gradient norm.
In the path network 3-2-1, the impact of varying sample batch sizes on Algorithm~\ref{distributedneuralpolicygradientAlgorithm} is illustrated in Fig.~\ref{network321changeBtot}.
These results suggest that as the sample batch size increases, there is a gradual improvement in both the objective function and the policy gradient norm.
\begin{figure}[htbp]
    \centering
    \subfigure[ $J(\bm{\pi}_{\bm{\theta}(k)})$]{\includegraphics[width=0.24\textwidth]{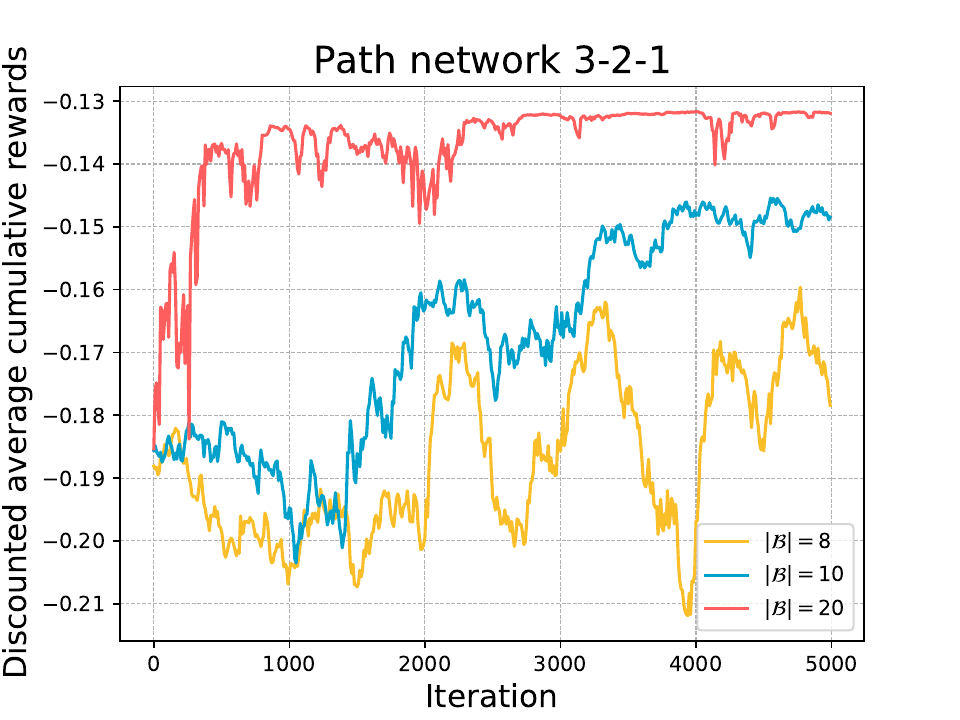}\label{NMARL-ITP3-2-1changeB}}
    \subfigure[ $\|\nabla_{\bm{\theta}}J(\bm{\pi}_{\bm{\theta}(k)})\|_{2}$]{\includegraphics[width=0.24\textwidth]{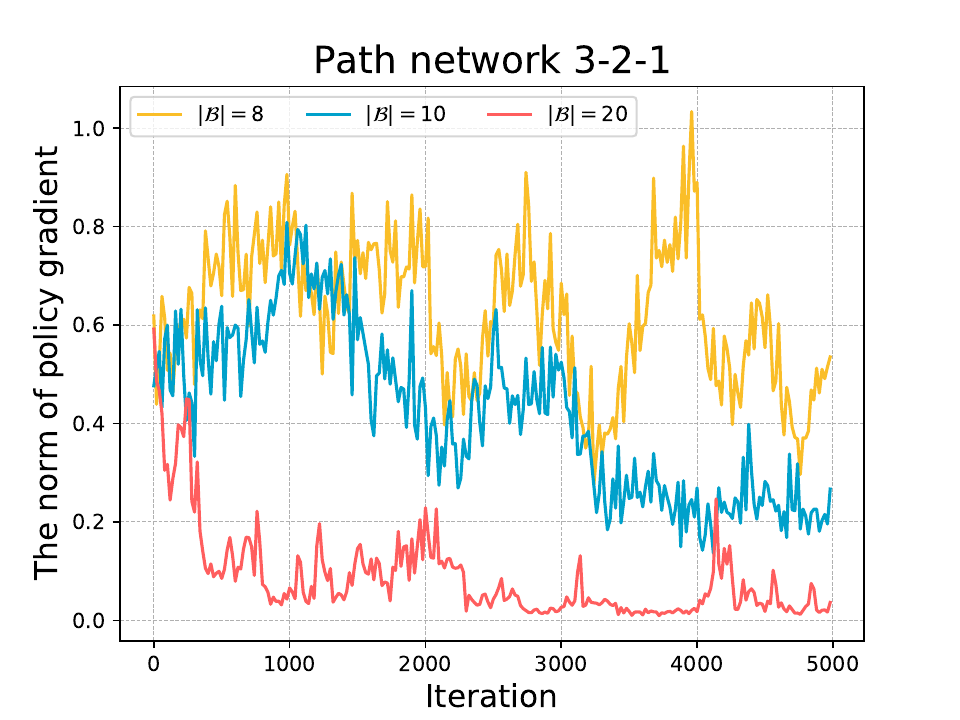}\label{NMARL-ITP3-2-1policygradientchangeB}}
    \caption{Performances of Algorithm~\ref{distributedneuralpolicygradientAlgorithm} with $|\mathcal{B}|=8,10,20$ in path network~3-2-1.}
    \label{network321changeBtot}
\end{figure}
\par
To investigate the impact of isolated agents in communication networks on Algorithm~\ref{distributedneuralpolicygradientAlgorithm}, we consider two scenarios: (i) agents 1 and 4 are isolated; (ii) agents 2 and 5 are isolated.
%As illustrated in Fig.~\ref{network321changeisolate}, the existence of isolated agents significantly affects the algorithm's performance, especially since agents 2 and 5 occupy critical positions that enhance their influence on the algorithm.
As depicted in Fig.~\ref{network321changeisolate}, the presence of isolated agents has an impact on the algorithm's performance.
This is particularly evident in second case as agents 2 and 5 occupy critical positions that amplify their influence on the algorithm's behavior.
%In order to explore the impact of isolated agents in communication networks on our Algorithm~\ref{distributedneuralpolicygradientAlgorithm}, we consider the isolated agents are 1 and 4, as well as isolated agents 2 and 5, based on their initial locations.
%As depicted in Fig.~\ref{network321changeisolate}, the presence of isolated agents directly impacts algorithmic efficacy, particularly due to the critical positions occupied by agents~2 and~5, which amplify their influence on the algorithm.
\begin{figure}[htbp]
    \centering
    \subfigure[ $J(\bm{\pi}_{\bm{\theta}(k)})$]{\includegraphics[width=0.24\textwidth]{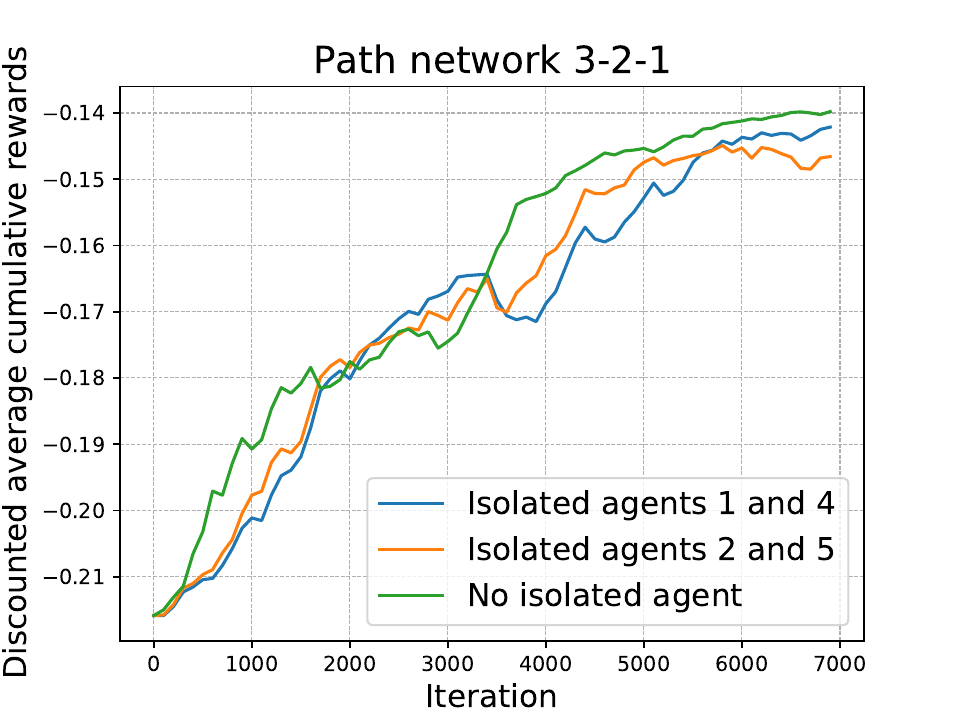}\label{NMARL-ITP3-2-1change6}}
    \subfigure[ $\|\nabla_{\bm{\theta}}J(\bm{\pi}_{\bm{\theta}(k)})\|_{2}$]{\includegraphics[width=0.24\textwidth]{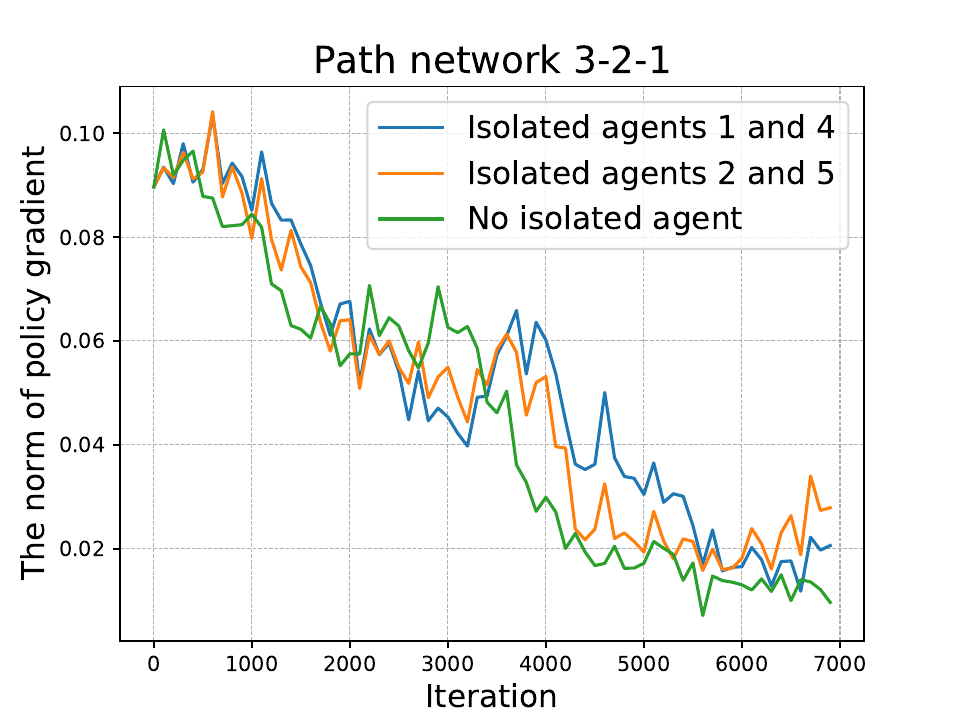}\label{NMARL-ITP3-2-1policygradientchange6}}
    \caption{Performances of Algorithm~\ref{distributedneuralpolicygradientAlgorithm} with isolated agents in path network~3-2-1.}
    \label{network321changeisolate}
\end{figure}

\section{Conclusions}\label{SectionVConclusions}
%The paper proposes two innovative neural networks and a distributed neural policy gradient algorithm, which consists of two components: the distributed critic step and the decentralized actor step.
%In this algorithm, agents solely rely on their neighbors' approximate parameters, thereby eliminating the requirement for a centralized controller.
This paper introduces two novel neural network architectures to propose a distributed neural policy gradient algorithm, which comprises two key components: the distributed critic step and the decentralized actor step.
In this algorithm, agents exclusively utilize approximated parameters from their neighboring agents, thereby eliminating the requirement for a centralized controller.
Additionally, the global convergence is established to provide the theoretical support for the proposed algorithm.
%In future, the application of the distributed neural policy gradient in realistic intelligent scenarios, such as smart grid and intelligent transportation system can be studied.
In the future, the application of the distributed neural policy gradient in realistic intelligent scenarios, such as smart grids and intelligent transportation systems, can be further explored.
%Future research can explore the application of the distributed neural algorithm in practical intelligent scenarios, such as smart grid and intelligent transportation systems.
%\section{Appendix}\label{AppendixA}
\appendix

\subsection{Proof of Lemma~\ref{lemmaofaverageconsensus}}\label{theproofoflemmaofaverageconsensus}
\begin{proof}
%Consider $\bm{W}_{\mathrm{tot}}(t+1)-\mathbf{1}\otimes\overline{\bm{W}}(t+1)$,
Recall that $g^{k}_{i}(t)$ in (\ref{theTDlearningindistributedneuralTDalgorithm1}) and $e^{k}_{i}(t)$ in (\ref{thedefinitionofstochastic-5}), define $\mathbf{G}^{k}(t)=\big(g^{k}_{1}(t)^{\top},\cdots,g^{k}_{N}(t)^{\top}\big)^{\top}$ and $\mathbf{E}^{k}(t)=\big(e^{k}_{1}(t)^{\top},\cdots,e^{k}_{N}(t)^{\top}\big)^{\top}$.
According to definition of $C=I_{N}-(1/N)\mathbf{1}_{N}\mathbf{1}_{N}^{\top}$, we have
\begin{align}
&(C\otimes I_{dmN})W_{tot}(t+1)\notag\\
=&W_{tot}(t+1)-\mathbf{1}_{N}\otimes\overline{W}(t+1)\notag\\
=&\big(A(t)\otimes I_{dmN}\big)W_{tot}(t)-\eta_{c,t}\mathbf{G}^{k}(t)-\mathbf{E}^{k}(t)\notag\\
&-\big(\mathbf{1}_{N}\otimes\overline{W}(t)-\eta_{c,t}\mathbf{1}_{N}\otimes g^{k}_{\mathrm{ave}}(t)-\mathbf{1}_{N}\otimes e^{k}_{\mathrm{ave}}(t)\big)\label{theequalityofaverageconsensusinequality}\\
=&\big(A(t)\otimes I_{dmN}\big)(C\otimes I_{dmN})W_{tot}(t)\notag\\
&-\eta_{c,t}(C\otimes I_{dmN})\mathbf{G}^{k}(t)-(C\otimes I_{dmN})\mathbf{E}^{k}(t)\notag\\
=&\prod_{t'=0}^{t}\big(A(t-t')\otimes I_{dmN}\big)(C\otimes I_{dmN})W_{tot}(0)\notag\\
&+\sum_{t'=0}^{t}\eta_{c,t'}\prod_{l=t'+1}^{t}\big(A(l)\otimes I_{dmN}\big)(C\otimes I_{dmN})\mathbf{G}^{k}(t')\notag\\
&-\sum_{t'=0}^{t}\prod_{l=t'+1}^{t}\big(A(l)\otimes I_{dmN}\big)(C\otimes I_{dmN})\mathbf{E}^{k}(t'),\label{theequalityofaverageconsensus}
\end{align}
where (\ref{theequalityofaverageconsensusinequality}) comes from (\ref{thekeystepindistributedneuralTDalgorithm}).
By taking the norm on both sides of (\ref{theequalityofaverageconsensus}), we have
\begin{align}
&\|(C\otimes I_{dmN})W_{tot}(t+1)\|_{2}\notag\\
\leq&\Big\|\Big(\prod_{t'=0}^{t}A(t-t')\Big)C\otimes I_{dmN}\Big\|_{2}\big\|W_{tot}(0)\big\|_{2}\notag\\
&+\sum_{t'=0}^{t}\eta_{c,t'}\Big\|\Big(\prod_{l=t'+1}^{t}A(l)\Big)C\otimes I_{dmN}\Big\|_{2}\big\|\mathbf{G}^{k}(t')\big\|_{2}\notag\\
&+\sum_{t'=0}^{t}\Big\|\Big(\prod_{l=t'+1}^{t}A(l)\Big)C\otimes I_{dmN}\Big\|_{2}\big\|\mathbf{E}^{k}(t')\big\|_{2}\notag\\
\leq&\frac{\lambda^{t+1}_{D}}{\lambda}\|W_{tot}(0)\|_{2}+\sum_{t'=0}^{t}\frac{\lambda^{t-t'}_{D}}{\lambda}\eta_{c,t'}\|\mathbf{G}^{k}(t')\|_{2}\notag\\
&+\sum_{t'=0}^{t}\frac{\lambda^{t-t'}_{D}}{\lambda}\|\mathbf{E}^{k}(t')\|_{2},\label{inequalityofaverageconsensus}
\end{align}
where the first inequality uses the triangle inequality and the fact that the norm of $C$ is less than 1, and the last inequality follows from Proposition~4 in~\cite{nedic2018}.
%\begin{align}
%\big\|\big(A(t)C\otimes I_{dmN}\big)W_{tot}\big\|_{2}\leq\lambda\|(C\otimes I_{dmN})W_{tot}\|_{2}.
%\end{align}
%Based on (\ref{inequalityofaverageconsensus}), we have
%\begin{align}
%&\|(C\otimes I_{dmN})W_{tot}(t+1)\|^{2}_{2}\notag\\
%\leq&3\lambda^{2t+2}\|W_{tot}(0)\|^{2}_{2}+3\sum_{t'=0}^{t}\lambda^{t-t'}\eta_{t'}\|\mathbf{G}^{k}(t')\|_{2}\notag\\
%&+\sum_{t'=0}^{t}\lambda^{t-t'}\|\mathbf{E}^{k}(t')\|_{2},\label{inequalityofaverageconsensusafter}
%\end{align}
%Similar to proof of Lemma 1 in \cite{Doan2019ICML},
By taking expectation on both sides of (\ref{inequalityofaverageconsensus}), we have
%Taking expectation on both sides of (\ref{inequalityofaverageconsensus}), we have
\begin{align}
&\mathbb{E}_{\mathrm{init}}[\|(C\otimes I_{dmN})W_{tot}(t+1)\|_{2}]\notag\\
\leq&\frac{\lambda^{t+1}_{D}}{\lambda}\mathbb{E}_{\mathrm{init}}[\|W_{tot}(0)\|_{2}]+\sum_{t'=0}^{t}\frac{\lambda^{t-t'}_{D}}{\lambda}\eta_{c,t'}\mathbb{E}_{\mathrm{init}}[\|\mathbf{G}^{k}(t')\|_{2}]\notag\\
&+\sum_{t'=0}^{t}\frac{\lambda^{t-t'}_{D}}{\lambda}\mathbb{E}_{\mathrm{init}}[\|\mathbf{E}^{k}(t')\|_{2}].\label{resultofinequalityofaverageconsensusbefore}
\end{align}
\par
For $\mathbf{G}^{k}(t')$ the right-hand side of (\ref{resultofinequalityofaverageconsensusbefore}), we can use (\ref{theTDlearningindistributedneuralTDalgorithm1}) and obtain
\begin{align}
&\mathbb{E}_{\mathrm{init}}[\|g^{k}_{i}(t')\|^{2}_{2}]\notag\\
=&\mathbb{E}_{\mathrm{init}}\big[\big\|\delta^{k}_{i}(t')\nabla_{W_{i}}\widehat{Q}_{i,t'}(\bm{z})\big\|^{2}_{2}\big]\notag\\
\leq& (mN)^{1-2p}\mathbb{E}_{\mathrm{init}}\big[\big|\delta^{k}_{i,t'}\big|^{2}\big]\notag\\
\leq& (mN)^{1-2p}\mathbb{E}_{\mathrm{init}}\big[\big(\widehat{Q}_{i,t'}(\bm{z})-(1-\gamma)r_{i}(\bm{z})-\gamma\widehat{Q}_{i,t'}(\bm{z}')\big)^{2}\big]\notag\\
\leq&3(mN)^{1-2p}\mathbb{E}_{\mathrm{init}}\big[\widehat{Q}_{i,t'}(\bm{z})^{2}+\widehat{Q}_{i,t'}(\bm{z}')^{2}+R^{2}_{0}\big]\notag\\
\leq& (mN)^{1-2p}\mathbb{E}_{\mathrm{init}}\big[12\widehat{Q}_{0}(\bm{z})^{2}+12B^{2}(mN)^{1-2p}+3R^{2}_{0}\big]\notag\\
\leq& 3R^{2}_{0}(mN)^{1-2p}+(12d_{1}+12B^{2})(mN)^{2-4p},\label{thesecondinequalityofVarianceBound}
\end{align}
where the first inequality follows from (iv) of Fact \ref{thefactinpaper}, the third inequality uses Assumption \ref{theassumptionofreward} and the fact that $(a + b + c)^2 \leq 3a^2 + 3b^2 + 3c^2$ for all $a, b, c \in \mathbb{R}$, the forth inequality is obtained from (v) of Fact \ref{thefactinpaper}, and the last inequality follows from (ii) of Fact \ref{thefactinpaper}.
According to (\ref{thesecondinequalityofVarianceBound}), we further have
\begin{align}\label{theboundofg}
\mathbb{E}_{\mathrm{init}}[\|g^{k}_{i}(t')\|_{2}]\leq\mathcal{O}\big(B(mN)^{\frac{1}{2}-p}\big).
\end{align}
\par
For $\mathbf{E}^{k}(t')$ in the right-hand side of  (\ref{resultofinequalityofaverageconsensusbefore}), we use (\ref{thedefinitionofstochastic-5}) and the projection property to derive
\begin{align}\label{theprojectionpropertyofe}
\|e^{k}_{i}(t')\|_{2}\leq\eta_{c,t'}\|g^{k}_{i}(t)\|_{2}.
\end{align}
Substituting (\ref{theboundofg}) and (\ref{theprojectionpropertyofe}) in (\ref{resultofinequalityofaverageconsensusbefore}), we have
\begin{align}
&\mathbb{E}_{\mathrm{init}}[\|(C\otimes I_{dmN})W_{tot}(t+1)\|_{2}]\notag\\
\leq&\frac{\sqrt{md_{1}}N}{\lambda}\lambda^{t+1}_{D}+\sum_{t'=0}^{t}\frac{\lambda^{t-t'}_{D}}{\lambda}\eta_{c,t'}\mathcal{O}(Bm^{\frac{1}{2}-p}N^{1-p}).\label{resultofinequalityofaverageconsensus}
\end{align}
%where (\ref{resultofinequalityofaverageconsensus}) can be get from (i) and (ii) in Lemma \ref{thelemmainpopulationsemigradient}.
Substituting $\eta_{c,t}=\frac{1-\gamma}{24\sqrt{t+1}}$ into (\ref{resultofinequalityofaverageconsensus}), we have
\begin{align}
&\mathbb{E}_{\mathrm{init}}[\|(C\otimes I_{dmN})W_{tot}(t)\|_{2}]\notag\\
\leq&\frac{\sqrt{md_{1}}N}{\lambda}\lambda^{t}_{D}+\sum_{t'=0}^{t-1}\frac{(1-\gamma)\lambda^{t-1-t'}_{D}}{24\lambda\sqrt{t'+1}}\mathcal{O}(Bm^{\frac{1}{2}-p}N^{1-p}),\notag
\end{align}
which completes the proof.
\end{proof}

\subsection{Proof of Lemma \ref{Lemma4-6}}\label{theproofofLemma4-6}
\underline{For (i) of Lemma \ref{Lemma4-6}:}
By recalling the definitions of $g^{k}_{i}(t)$ in (\ref{theTDlearningindistributedneuralTDalgorithm1}) and $\bar{g}^{k}_{i}(t)$ in (\ref{thedefinitionofstochastic-2}), we have
\begin{align}\label{thefirstinequalityofVarianceBound}
&\mathbb{E}_{\mathrm{init}}[\|g^{k}_{i}(t)-\bar{g}^{k}_{i}(t)\|^{2}_{2}]\notag\\
%=&\mathbb{E}_{\mathrm{init}}[\|g_{i}(k)\|^{2}_{2}-2g^{\top}_{i}(k)\bar{g}_{i}(k)+\|\bar{g}_{i}(k)\|^{2}_{2}]\notag\\
\leq&2\mathbb{E}_{\mathrm{init}}[\|g^{k}_{i}(t)\|^{2}_{2}]+2\mathbb{E}_{\mathrm{init}}[\|\bar{g}^{k}_{i}(t)\|^{2}_{2}].
\end{align}
By using $\bar{g}^{k}_{i}(t)=\mathbb{E}_{\varsigma^{+}_{k}}[g^{k}_{i}(t)]$ in (\ref{thedefinitionofstochastic-2}) and (\ref{thesecondinequalityofVarianceBound}), we have
\begin{align}
&\mathbb{E}_{\mathrm{init},\varsigma^{+}_{k}}[\|g^{k}_{i}(t)-\bar{g}^{k}_{i}(t)\|^{2}_{2}]\notag\\
\leq& 12R^{2}_{0}(mN)^{1-2p}+(48d_{1}+48B^{2})(mN)^{2-4p}.\label{fromthefacti}
\end{align}
\par
\underline{For (ii) of Lemma \ref{Lemma4-6}:}
By the definitions of $\bar{g}^{k}_{i}(t)$ in (\ref{thedefinitionofstochastic-2}) and $\bar{g}^{k}_{0,\mathrm{ave}}(t)$ in (\ref{thedefinitionofstochastic2-3}), we have
\begin{align}\label{gradientdifference}
&\|\bar{g}^{k}_{i}(t)-\bar{g}^{k}_{0,\mathrm{ave}}(t)\|_{2}\notag\\
=&\big\|\mathbb{E}_{\varsigma_{k}^{+}}\big[\delta^{k}_{i}(t)\nabla_{W_{i}}\widehat{Q}_{i,t}(\bm{z})\big]\notag\\
&-\mathbb{E}_{\varsigma_{k}^{+}}\big[\delta^{k}_{0,\mathrm{ave}}(t)\nabla_{W'}\widehat{Q}_{0}\big(\bm{z};\overline{W}(t)\big)\big]\big\|_2\notag\\
=&\Big\|\mathbb{E}_{\varsigma_{k}^{+}}\Big[\Big(\delta^{k}_{i}(t)-\delta^{k}_{0,\mathrm{ave}}(t)\Big)\nabla_{W_{i}}\widehat{Q}_{i,t}(\bm{z})\notag\\
&+\delta^{k}_{0,\mathrm{ave}}(t)\Big(\nabla_{W_{i}}\widehat{Q}_{i,t}(\bm{z})-\nabla_{W'}\widehat{Q}_{0}\big(\bm{z};\overline{W}(t)\big)\Big)\Big]\Big\|_{2}\notag\\
\leq&\mathbb{E}_{\varsigma_{k}^{+}}\Big[(mN)^{\frac{1}{2}-p}\big|\delta^{k}_{i}(t)-\delta^{k}_{0,\mathrm{ave}}(t)\big|\notag\\
&+\big|\delta^{k}_{0,\mathrm{ave}}(t)\big|\big\|\nabla_{W_{i}}\widehat{Q}_{i,t}(\bm{z})-\nabla_{W'}\widehat{Q}_{0}\big(\bm{z};\overline{W}(t)\big)\big\|_{2}\Big],
\end{align}
where the inequality is obtained from (iv) of Fact \ref{thefactinpaper}.
%\begin{align}
%\|\nabla_{\bm{W}}\widehat{Q}_{i,k}(s,\bm{a})\|_{2}\leq (Nm)^{1/2-p},\forall k\in[K_{c}].\notag
%\end{align}
By taking expectation over the random initialization on both sides of (\ref{gradientdifference}) after squaring, we have
\begin{align}
&\mathbb{E}_{\mathrm{init}}\big[\|\bar{g}^{k}_{i}(t)-\bar{g}^{k}_{0,\mathrm{ave}}(t)\|^{2}_{2}\big]\notag\\
\leq&2(mN)^{1-2p}\mathbb{E}_{\mathrm{init}}\Big[\mathbb{E}_{\varsigma^{+}_{k}}\big[\big|\delta^{k}_{i}(t)-\delta^{k}_{0,\mathrm{ave}}(t)\big|\big]^{2}\Big]\notag\\
&+2\mathbb{E}_{\mathrm{init}}\Big[\mathbb{E}_{\varsigma^{+}_{k}}\big[\big|\delta^{k}_{0,\mathrm{ave}}(t)\big|\cdot\notag\\
&\big\|\nabla_{W_{i}}\widehat{Q}_{i,t}(\bm{z})-\nabla_{W'}\widehat{Q}_{0}\big(\bm{z};\overline{W}(t)\big)\big\|_{2}\big]^{2}\Big]\notag\\
\leq& \underbrace{2(mN)^{1-2p}\mathbb{E}_{\mathrm{init},\varsigma^{+}_{k}}\big[\big|\delta^{k}_{i}(t)-\delta^{k}_{0,\mathrm{ave}}(t)\big|^{2}\big]}_{\mathrm{(i)}}\notag\\
%&\underbrace{-\delta_{0}\big(\bm{z},\bar{r},\bm{z}';\overline{\bm{W}}(k)\big)\big|^{2}\big]}_{\mathrm{(i)}}\notag\\
&+\underbrace{2\mathbb{E}_{\mathrm{init}}\Big[\mathbb{E}_{\varsigma_{k}^{+}}\big[\big|\delta^{k}_{0,\mathrm{ave}}(t)\big|^{2}\big]}_{\mathrm{(ii)}}\cdot\notag\\
&\underbrace{\mathbb{E}_{\varsigma_{k}^{+}}\big[\big\|\nabla_{W_{i}}\widehat{Q}_{i,t}(\bm{z})-\nabla_{W'}\widehat{Q}_{0}\big(\bm{z};\overline{W}(t)\big)\big\|^{2}_{2}\big]\Big]}_{\mathrm{(ii)}},\label{expectationofgradientdifference}
\end{align}
where the last inequality follows from the H{\"o}lder's inequality.
\par
For the $\mathrm{(i)}$-th term in (\ref{expectationofgradientdifference}), we have
\begin{align}
&\mathbb{E}_{\mathrm{init},\varsigma_{k}^{+}}\big[\big|\delta^{k}_{i}(t)-\delta^{k}_{0,\mathrm{ave}}(t)\big|^{2}\big]\notag\\
=&\mathbb{E}_{\mathrm{init},\varsigma_{k}^{+}}\Big[\Big|\big(\widehat{Q}_{i,t}(\bm{z})-(1-\gamma)r_{i}(\bm{z})-\gamma\widehat{Q}_{i,t}(\bm{z}')\big)\notag\\
&-\Big(\widehat{Q}_{0}\big(\bm{z};\overline{W}(t)\big)-(1-\gamma)\bar{r}(\bm{z})-\gamma\widehat{Q}_{0}\big(\bm{z}';\overline{W}(t)\big)\Big)\Big|^{2}\Big]\notag\\
%=&\Big|\Big(\widehat{Q}_{i,k}(s,\bm{a})-\widehat{Q}_{0}\big(s,\bm{a};\overline{\bm{W}}(k)\big)\Big)-\gamma\Big(\widehat{Q}_{i,k}(s',\bm{a}')\notag\\
%&-\widehat{Q}_{0}\big(s',\bm{a}';\overline{\bm{W}}(k)\big)\Big)+(1-\gamma)(r_{\mathrm{ave}}-r_{i})\Big|^{2}\notag\\
\leq& \mathbb{E}_{\mathrm{init},\varsigma_{k}^{+}}\Big[3\Big(\widehat{Q}_{i,t}(\bm{z})-\widehat{Q}_{0}\big(\bm{z};\overline{W}(t)\big)\Big)^{2}\notag\\
&+3\Big(\widehat{Q}_{i,t}(\bm{z}')-\widehat{Q}_{0}\big(\bm{z}';\overline{W}(t)\big)\Big)^{2}\notag\\
&+3(1-\gamma)^{2}\big(\bar{r}(\bm{z})-r_{i}(\bm{z})\big)^{2}\Big]\notag\\
\leq& 12\mathbb{E}_{\mathrm{init},\varsigma_{k}}\Big[\Big(\widehat{Q}_{i,t}(\bm{z})-\widehat{Q}_{0}\big(\bm{z};W_{i}(t)\big)\Big)^{2}\Big]\notag\\
&+12\mathbb{E}_{\mathrm{init},\varsigma_{k}}\Big[\Big(\widehat{Q}_{0}\big(\bm{z};W_{i}(t)\big)-\widehat{Q}_{0}\big(\bm{z};\overline{W}(t)\big)\Big)^{2}\Big]+12R^{2}_{0}\notag\\
\leq&48c_{1}B^{3}(mN)^{\frac{1}{2}-2p}+48B^{2}(mN)^{1-2p}+12R^{2}_{0},\label{firstterminexpectationofgradientdifference}
\end{align}
where the last inequality follows from (iii) of Lemma~\ref{Lemma1to3}, (iii) of Fact~\ref{thefactinpaper}, and Assumption~\ref{theassumptionofreward}.
\par
For the $\mathrm{(ii)}$-th term in (\ref{expectationofgradientdifference}), we have
\begin{align}\label{secondterminexpectationofgradientdifference}
&\mathbb{E}_{\varsigma_{k}^{+}}\big[\big|\delta^{k}_{0,\mathrm{ave}}(t)\big|^{2}\big]\notag\\
\leq& 3\mathbb{E}_{\varsigma_{k}^{+}}\Big[\Big(\widehat{Q}_{0}\big(\bm{z};\overline{W}(t)\big)^{2}+\gamma^{2}\widehat{Q}_{0}\big(\bm{z}';\overline{W}(t)\big)^{2}\notag\\
&+(1-\gamma)^{2}R^{2}_{0}\Big)\Big]\notag\\
\leq&12\mathbb{E}_{\varsigma_{k}}\big[\widehat{Q}_{0}(\bm{z})^{2}\big]+12B^{2}(mN)^{1-2p}+3R^{2}_{0},
\end{align}
where the last inequality comes from (vi) of Fact~\ref{thefactinpaper}.
% we have
%\begin{align}\label{theinequalityofsecondterminexpectationofgradientdifference}
%\mathbb{E}_{\varsigma_{t}}\big[\widehat{Q}_{0}\big(s,\bm{a};\overline{\bm{W}}(k)\big)^{2}\big]\leq& 2\mathbb{E}_{\varsigma_{t}}[\widehat{Q}_{0}(s,\bm{a})^{2}]+2B^{2}(Nm)^{-2p}.
%\end{align}
%Combining (\ref{secondterminexpectationofgradientdifference}) and \ref{theinequalityofsecondterminexpectationofgradientdifference}, we can get that
%\begin{align}\label{resultofsecondterminexpectationofgradientdifference}
%&\mathbb{E}_{\varsigma_{t}^{+}}\big[\big|\delta_{0}\big(s,\bm{a},r_{\mathrm{ave}},s',\bm{a}';\overline{\bm{W}}(k)\big)\big|^{2}\big]\notag\\
%\leq&12\mathbb{E}_{\varsigma_{t}}[\widehat{Q}_{0}(s,\bm{a})^{2}]+12B^{2}(Nm)^{-2p}+3(1-\gamma)^{2}R^{2}_{0}.
%\end{align}
Beside, we have
%$\mathbb{E}_{\varsigma_{t}^{+}}\big[\big\|\nabla_{\bm{W}}\widehat{Q}_{i,k}(s,\bm{a})-\nabla_{\bm{W}}\widehat{Q}_{0}\big(s,\bm{a};\overline{\bm{W}}(k)\big)\big\|^{2}_{2}\big]$
%For $\mathrm{(iii)}$-th term in (\ref{expectationofgradientdifference}), we can obtain
\begin{align}
&\mathbb{E}_{\varsigma_{k}^{+}}\Big[\big\|\nabla_{W_{i}}\widehat{Q}_{i,t}(\bm{z})-\nabla_{W'}\widehat{Q}_{0}\big(\bm{z};\overline{W}(t)\big)\big\|^{2}_{2}\Big]\notag\\
\leq&\frac{1}{(mN)^{2p}}\mathbb{E}_{\varsigma_{k}}\Big[\sum_{r=1}^{Nm}\big(\mathds{1}\{W^{\top}_{i,r}(t)x_{\lfloor(r-1)/m\rfloor+1}>0\}\notag\\
&-\mathds{1}\{W'^{\top}_{r}(0)x_{\lfloor(r-1)/m\rfloor+1}>0\}\big)^{2}\big\|x_{\lfloor(r-1)/m\rfloor+1}\big\|^{2}_{2}\Big]\notag\\
%\leq&\frac{1}{(mN)^{2p}}\mathbb{E}_{\varsigma_{k}}\Big[\sum_{r=1}^{Nm}\mathds{1}\{|W'^{\top}_{r}(0)x_{\lfloor(r-1)/m\rfloor+1}|\notag\\
%&\leq \|W_{i,r}(t)-W'_{r}(0)\|_{2}\}\Big],
=&\frac{1}{(mN)^{2p}}\mathbb{E}_{\varsigma_{k}}\Big[\sum_{r=1}^{mN}\Big|\mathds{1}\{W^{\top}_{i,r}(t)x_{\lfloor(r-1)/m\rfloor+1}>0\}\notag\\
&-\mathds{1}\{W'^{\top}_{r}(0)x_{\lfloor(r-1)/m\rfloor+1}>0\}\Big|\Big]\label{eq:distribution}\\
\leq&  \frac{1}{(mN)^{2p}}\mathbb{E}_{\varsigma_{k}}\Big[\sum_{r=1}^{mN}\mathds{1}\{|W'^{\top}_{r}(0)x_{\lfloor(r-1)/m\rfloor+1}|\notag\\
&\leq\|W_{i,r}(t)-W'_{r}(0)\|_{2}\}\Big],\label{thirdterminexpectationofgradientdifference}
%\leq&\frac{1}{(mN)^{p}}\sum_{r=1}^{mN}\mathds{1}\{|W'^{\top}_{r}(0)x_{\lfloor(r-1)/m\rfloor+1}|\leq\|W'_{r}-W'_{r}(0)\|_{2}\}\notag\\
%&\big(\|W'_{r}-W'_{r}(0)\|_{2}+\|W''_{r}-W'_{r}(0)\|_{2}\big),
\end{align}
where the first inequality follows from the fact that $b_{r}\sim\mathrm{Unif}(\{-1,1\})$, (\ref{eq:distribution}) uses  $\|x_{i}\|_{2}=1$, and the last inequality comes from the fact that
\begin{align}
&\mathds{1}\{W^{\top}_{i,r}(t)x_{\lfloor(r-1)/m\rfloor+1}>0\}\notag\\
&\neq\mathds{1}\{W'^{\top}_{r}(0)x_{\lfloor(r-1)/m\rfloor+1}>0\}\notag\\
\Rightarrow&|W'^{\top}_{r}(0)x_{\lfloor(r-1)/m\rfloor+1}|\notag\\
&\leq\big|\big(W^{\top}_{i,r}(t)-W'^{\top}_{r}(0)\big)x_{\lfloor(r-1)/m\rfloor+1}\big|\notag\\
&\leq\|W'_{i,r}(t)-W'_{r}(0)\|_{2}.\notag
\end{align}
%and {\color{blue}the last inequality} can be obtained by
%$\mathds{1}\{|u_{1}|\leq u_{2}\}|u_{1}|\leq\mathds{1}\{|u_{1}|\leq u_{2}\}u_{2},\forall u_{1},u_{2}>0$.
%where the last inequality can be similarly obtained by (\ref{theinequalityinLocalLinearization}).
%{\color{blue}and the last inequality follows from (i) of Lemma~\ref{Lemma1to3}.}
%Taking the expectation on both sides of (\ref{thirdterminexpectationofgradientdifference}), we have
%\begin{align}\label{theresultthirdterminexpectationofgradientdifference}
%&\mathbb{E}_{\varsigma_{t}^{+}}\big[\big\|\nabla_{\bm{W}}\widehat{Q}_{i,k}(s,\bm{a})-\nabla_{\bm{W}}\widehat{Q}_{0}\big(s,\bm{a};\overline{\bm{W}}(k)\big)\big\|^{2}_{2}\big]\notag\\
%\leq&\mathbb{E}_{\varsigma_{t}}\Big[\frac{1}{(Nm)^{2p}}\sum_{r=1}^{Nm}\mathds{1}\{|W_{r}(0)^{\top}x_{\lfloor(r-1)/m\rfloor+1}|\notag\\
%&\leq \|W_{i,r}(k)-W_{r}(0)\|_{2}\}\Big].
%\end{align}
%By multiplying $\mathrm{(ii)}$-term and $\mathrm{(iii)}$-term in (\ref{expectationofgradientdifference}) and taking the expectation over the random initialization,
Combining (\ref{secondterminexpectationofgradientdifference}) and (\ref{thirdterminexpectationofgradientdifference}), we have
\begin{align}
&2\mathbb{E}_{\mathrm{init}}\Big[\mathbb{E}_{\varsigma_{t}^{+}}\big[\big|\delta^{k}_{0,\mathrm{ave}}(t)\big|^{2}\big]\cdot\notag\\
&\;\;\;\;\;\;\;\;\;\;\;\mathbb{E}_{\varsigma_{k}^{+}}\big[\big\|\nabla_{W_{i}}\widehat{Q}_{i,k}(\bm{z})-\nabla_{W'}\widehat{Q}_{0}\big(\bm{z};\overline{W}(t)\big)\big\|^{2}_{2}\big]\Big]\notag\\
\leq&(24c_{2}B+24c_{1}B^{3})(mN)^{\frac{3}{2}-4p}+6c_{1}R^{2}_{0}B(mN)^{\frac{1}{2}-2p},\label{midinequality}
\end{align}
where the inequality uses (i) and (ii) of Lemma~\ref{Lemma1to3}.
\par
Substituting (\ref{firstterminexpectationofgradientdifference}) and (\ref{midinequality}) into (\ref{expectationofgradientdifference}), we have
\begin{align}
&\mathbb{E}_{\mathrm{init}}\big[\|\bar{g}^{k}_{i}(t)-\bar{g}^{k}_{0,\mathrm{ave}}(t)\|^{2}_{2}\big]\notag\\
\leq& 24R^{2}_{0}(mN)^{1-2p}+6c_{1}R_{0}^{2}B(mN)^{\frac{1}{2}-2p}\notag\\
&+(120c_{1}B^{3}+24c_{2}B)(mN)^{\frac{3}{2}-4p}+96B^{2}(mN)^{2-4p}\notag\\
=&\mathcal{O}\big(B^{3}(mN)^{1-2p}\big),\label{eq:range}
\end{align}
which complets the proof of (ii) of Lemma~\ref{Lemma4-6}.
\par
\underline{For (iii) of Lemma \ref{Lemma4-6}:}
According to the definitions of $\bar{g}^{k}_{\mathrm{ave}}(t)$ in (\ref{thedefinitionofstochastic-4}), $\bar{g}^{k*}_{0,\mathrm{ave}}$ in (\ref{thedefinitionofstochastic2-1}), and $\bar{g}^{k}_{0,\mathrm{ave}}(t)$ in (\ref{thedefinitionofstochastic2-3}), we have
\begin{align}
&\mathbb{E}_{\mathrm{init}}\big[\|\bar{g}^{k}_{\mathrm{ave}}(t)\|^{2}_{2}\big]\notag\\
=&\mathbb{E}_{\mathrm{init}}\big[\|\bar{g}^{k}_{\mathrm{ave}}(t)-\bar{g}^{k}_{0,\mathrm{ave}}(t)+\bar{g}^{k}_{0,\mathrm{ave}}(t)-\bar{g}^{k*}_{0,\mathrm{ave}}+\bar{g}^{k*}_{0,\mathrm{ave}}\|^{2}_{2}\big]\notag\\
\leq& 3\mathbb{E}_{\mathrm{init}}\big[\|\bar{g}^{k}_{\mathrm{ave}}(t)-\bar{g}^{k}_{0,\mathrm{ave}}(t)\|^{2}_{2}\big]\notag\\
&+3\mathbb{E}_{\mathrm{init}}\big[\|\bar{g}^{k}_{0,\mathrm{ave}}(t)-\bar{g}^{k*}_{0,\mathrm{ave}}\|^{2}_{2}\big]+3\mathbb{E}_{\mathrm{init}}\big[\|\bar{g}^{k*}_{0,\mathrm{ave}}\|^{2}_{2}\big]\notag\\
\leq&\underbrace{\frac{3}{N}\mathbb{E}_{\mathrm{init}}\Big[\sum_{i=1}^{N}\|\bar{g}^{k}_{i}(t)-\bar{g}^{k}_{0,\mathrm{ave}}(t)\|^{2}\Big]}_{\mathrm{(i)}}\notag\\
&+\underbrace{3\mathbb{E}_{\mathrm{init}}\big[\|\bar{g}^{k}_{0,\mathrm{ave}}(t)\!-\!\bar{g}^{k*}_{0,\mathrm{ave}}\|^{2}_{2}\big]}_{\mathrm{(ii)}}+\underbrace{3\mathbb{E}_{\mathrm{init}}\big[\|\bar{g}^{k*}_{0,\mathrm{ave}}\|^{2}_{2}\big]}_{\mathrm{(iii)}},\label{firstinequalitythelemmaofbargave}
\end{align}
where the last inequality uses (\ref{thedefinitionofstochastic2-3}).
\par
For the $\mathrm{(i)}$-th term in (\ref{firstinequalitythelemmaofbargave}), we can use (ii) of
Lemma~\ref{Lemma4-6} and have
\begin{align}\label{firstinequalitythelemmaofbargave1-1}
\frac{1}{N}\mathbb{E}_{\mathrm{init}}\Big[\sum_{i=1}^{N}\|\bar{g}^{k}_{i}(t)-\bar{g}^{k}_{0,\mathrm{ave}}(t)\|^{2}\Big]=\mathcal{O}\big(B^{3}(mN)^{1-2p}\big).
\end{align}
%{\color{blue}For} $\|\bar{g}_{\mathrm{ave}}(k)-\bar{g}^{*}_{0}\|^{2}_{2}$ on the right-hand side of (\ref{firstinequalitythelemmaofbargave}), by the Cauchy-Schwarz inequality,
%we can obtain
\par
For the $\mathrm{(ii)}$-th term in (\ref{firstinequalitythelemmaofbargave}), we have
\begin{align}\label{secondinequalitythelemmaofbargave}
&\mathbb{E}_{\mathrm{init}}\big[\|\bar{g}^{k}_{0,\mathrm{ave}}(t)-\bar{g}^{k*}_{0,\mathrm{ave}}\|^{2}_{2}\big]\notag\\
=&\mathbb{E}_{\mathrm{init}}\Big[\Big\|\mathbb{E}_{\varsigma^{+}_{k}}\Big[\big(\delta^{k}_{0,\mathrm{ave}}(t)-\delta^{k*}_{0,\mathrm{ave}}\big)\nabla_{W'}\widehat{Q}_{0}\big(\bm{z};W'(0)\big)\Big]\Big\|^{2}_{2}\Big]\notag\\
\leq&(mN)^{1-2p} \mathbb{E}_{\mathrm{init},\varsigma_{k}^{+}}\Big[\Big(\widehat{Q}_{0}\big(\bm{z};\overline{W}(t)\big)-\widehat{Q}_{0}(\bm{z};W^{*}_{k})\notag\\
&-\gamma\widehat{Q}_{0}\big(\bm{z}';\overline{W}(t)\big)+\gamma\widehat{Q}_{0}(\bm{z}';W^{*}_{k})\Big)^{2}\Big]\notag\\
\leq&2(mN)^{1-2p}\mathbb{E}_{\mathrm{init},\varsigma_{k}^{+}}\Big[\Big(\widehat{Q}_{0}\big(\bm{z};\overline{W}(t)\big)-\widehat{Q}_{0}(\bm{z};W^{*}_{k})\Big)^{2}\notag\\
&+\gamma^{2}\Big(\widehat{Q}_{0}\big(\bm{z}';\overline{W}(t)\big)-\widehat{Q}_{0}(\bm{z}';W^{*}_{k})\Big)^{2}\Big]\notag\\
\leq& 4(mN)^{1-2p}\mathbb{E}_{\mathrm{init},\varsigma_{k}}\Big[\Big(\widehat{Q}_{0}\big(\bm{z};\overline{W}(t)\big)-\widehat{Q}_{0}(\bm{z};W^{*}_{k})\Big)^{2}\Big],
\end{align}
where the first inequality is obtained from (iii) of Fact \ref{thefactinpaper} and the last inequality comes from that $\bm{z}$ and $\bm{z}'$ follow the same stationary distribution.
\par
For the $\mathrm{(iii)}$-th term in (\ref{firstinequalitythelemmaofbargave}), according to the definition of $\bar{g}^{k*}_{0,\mathrm{ave}}$ in (\ref{thedefinitionofstochastic2-1}), we have
\begin{align}
&\mathbb{E}_{\mathrm{init}}\big[\|\bar{g}^{k*}_{0,\mathrm{ave}}\|^{2}_{2}\big]\notag\\
=&\mathbb{E}_{\mathrm{init}}\big[\|\mathbb{E}_{\varsigma_{k}^{+}}[\delta^{k*}_{0,\mathrm{ave}}\nabla_{W'}\widehat{Q}_{0}(\bm{z};W^{*}_{k})]\|^{2}_{2}\big]\notag\\
%\leq& \mathbb{E}_{\varsigma_{t}^{+}}[\delta_{0}(s,\bm{a},r_{\mathrm{ave}},s',\bm{a}';\bm{W}^{*})^{2}]\|\nabla_{\bm{W}}\widehat{Q}_{0}(s,\bm{a};\bm{W}^{*})\|^{2}_{2}\notag\\
\leq&(mN)^{1-2p}\mathbb{E}_{\mathrm{init},\varsigma_{k}^{+}}\big[\big(\widehat{Q}_{0}(\bm{z};W^{*}_{k})-(1-\gamma)\bar{r}(\bm{z})\notag\\
&-\gamma\widehat{Q}_{0}(\bm{z}';W^{*}_{k})\big)^{2}\big]\label{2thirdinequalitythelemmaofbargave}\\
\leq&3(mN)^{1-2p}\mathbb{E}_{\mathrm{init},\varsigma_{k}^{+}}\big[\widehat{Q}_{0}(\bm{z};W^{*}_{k})^{2}+\widehat{Q}_{0}(\bm{z}';W^{*}_{k})^{2}+R^{2}_{0}\big]\notag\\
\leq& 12(mN)^{1-2p}\mathbb{E}_{\mathrm{init},\varsigma_{k}}[\widehat{Q}_{0}(\bm{z})^{2}]+12B^{2}(mN)^{1-4p}\notag\\
&+3(mN)^{1-2p}R^{2}_{0}\label{4thirdinequalitythelemmaofbargave}\\
\leq& (12d_{1}+12B^{2})(mN)^{2-4p}+3(mN)^{1-2p}R^{2}_{0},\label{thirdinequalitythelemmaofbargave}
\end{align}
where (\ref{2thirdinequalitythelemmaofbargave}) follows from the (iii) of Fact~\ref{thefactinpaper}, (\ref{4thirdinequalitythelemmaofbargave}) and (\ref{thirdinequalitythelemmaofbargave}) are obtained from (vi) and (ii) of Fact~\ref{thefactinpaper}, respectively.
%\begin{align}\notag
%\mathbb{E}_{\varsigma_{t}}\Big[\sum_{r=1}^{Nm}\mathrm{ReLU}\big(W_{r}(0)^{\top}x\big)^{2}\Big]\leq \sum_{r=1}^{Nm}\|W_{r}(0)\|^{2}_{2}.
%\end{align}
\par
Substituting (\ref{firstinequalitythelemmaofbargave1-1}), (\ref{secondinequalitythelemmaofbargave}), and (\ref{thirdinequalitythelemmaofbargave}) into (\ref{firstinequalitythelemmaofbargave}), we have
\begin{align}
&\mathbb{E}_{\mathrm{init}}\big[\|\bar{g}^{k}_{\mathrm{ave}}(t)\|^{2}_{2}\big]\notag\\
%\leq& 3\mathbb{E}_{\mathrm{init}}[\|\bar{g}_{\mathrm{ave}}(k)-\bar{g}_{0,\mathrm{ave}}(k)\|^{2}_{2}]+3\mathbb{E}_{\mathrm{init}}[\|\bar{g}_{0,\mathrm{ave}}(k)-\bar{g}^{*}_{0}\|^{2}_{2}]\notag\\
%&+3\mathbb{E}_{\mathrm{init}}[\|\bar{g}^{*}_{0}\|^{2}_{2}]\notag\\
%\leq&3N^{-1}\mathbb{E}_{\mathrm{init}}\Big[\sum_{i=1}^{N}\|\bar{g}_{i}(k)-\bar{g}_{0,\mathrm{ave}}(k)\|^{2}\Big]\notag\\
%&+3\mathbb{E}_{\mathrm{init}}[\|\bar{g}_{0,\mathrm{ave}}(k)-\bar{g}^{*}_{0}\|^{2}_{2}]+3\mathbb{E}_{\mathrm{init}}[\|\bar{g}^{*}_{0}\|^{2}_{2}]\notag\\
\leq&12\mathbb{E}_{\mathrm{init},\varsigma_{k}}\Big[\Big(\widehat{Q}_{0}\big(\bm{z};\overline{W}(t)\big)-\widehat{Q}_{0}(\bm{z};W^{*}_{k})\Big)^{2}\Big]\notag\\
&+\mathcal{O}\big(B^{3}(mN)^{1-2p}\big).
\end{align}
%where the last inequality is obtained from (\ref{eq:range}), (\ref{secondinequalitythelemmaofbargave}), and (\ref{thirdinequalitythelemmaofbargave}).
Thus, the proof of Lemma \ref{Lemma4-6} is completed.
\subsection{Proof of Lemma \ref{StochasticDescentLemma}}\label{theproofofStochasticDescentLemma}
%Recall the definitions of {\color{blue}$g_{i}(k)$ in (\ref{thedefinitionofstochastic-1}), $g_{\mathrm{ave}}(k)$ in (\ref{thedefinitionofstochastic-2}), $e_{i}(k)$ in (\ref{thedefinitionofstochastic-5}), and $e_{\mathrm{ave}}(k)$ in (\ref{thedefinitionofstochastic-6}),}
By recalling the update of $W_{i}(t+1)$ in (\ref{thekeystepindistributedneuralTDalgorithm}),
we have
%\begin{align}\label{inequalityofStochasticDescentLemma}
%&\mathbb{E}_{\mathrm{init}}[\|\overline{\bm{W}}(k+1)-\bm{W}^{*}\|^{2}_{2}]\notag\\
%=&\mathbb{E}_{\mathrm{init}}[\|\overline{\bm{W}}(k)-\bm{W}^{*}\|^{2}_{2}+\eta_{c}^{2}\|g_{\mathrm{ave}}(k)\|^{2}_{2}+\|e_{\mathrm{ave}}(k)\|^{2}_{2}\notag\\
%&-2\eta_{c}\big(\overline{\bm{W}}(k)-\bm{W}^{*}\big)^{\top}g_{\mathrm{ave}}(k)\notag\\
%&-2\big(\overline{\bm{W}}(k)-\bm{W}^{*}\big)^{\top}e_{\mathrm{ave}}(k)+2\eta_{c} g^{\top}_{\mathrm{ave}}(k)e_{\mathrm{ave}}(k).
%\end{align}
%Taking expectation on both sides of (\ref{inequalityofStochasticDescentLemma}), we can obtain
\begin{align}\label{expectationinequalityofStochasticDescentLemma}
&\mathbb{E}_{\mathrm{init}}[\|\overline{W}(t+1)-W^{*}_{k}\|^{2}_{2}]\notag\\
=&\mathbb{E}_{\mathrm{init}}\big[\|\overline{W}(t)-W^{*}_{k}\|^{2}_{2}+\eta_{c,t}^{2}\|g^{k}_{\mathrm{ave}}(t)\|^{2}_{2}+\|e^{k}_{\mathrm{ave}}(t)\|^{2}_{2}\notag\\
&-2\eta_{c,t}\big(\overline{W}(t)-W^{*}_{k}\big)^{\top}g^{k}_{\mathrm{ave}}(t)\notag\\
&-2\big(\overline{W}(t)-W^{*}_{k}\big)^{\top}e^{k}_{\mathrm{ave}}(t)+2\eta_{c,t}g^{k}_{\mathrm{ave}}(t)^{\top}e^{k}_{\mathrm{ave}}(t)\big]\notag\\
=&\mathbb{E}_{\mathrm{init}}[\|\overline{W}(t)-W^{*}_{k}\|^{2}_{2}]+\underbrace{\eta_{c,t}^{2}\mathbb{E}_{\mathrm{init}}[\|g^{k}_{\mathrm{ave}}(t)\|^{2}_{2}]}_{\mathrm{(i)}}\notag\\
&+\underbrace{\mathbb{E}_{\mathrm{init}}[\|e^{k}_{\mathrm{ave}}(t)\|^{2}_{2}]}_{\mathrm{(ii)}}\notag\\
&\underbrace{-2\eta_{c,t}\mathbb{E}_{\mathrm{init}}\big[\big(\overline{W}(t)-W^{*}_{k}\big)^{\top}g^{k}_{\mathrm{ave}}(t)\big]}_{\mathrm{(iii)}}\notag\\
&\underbrace{-2\mathbb{E}_{\mathrm{init}}\big[\big(\overline{W}(t)-W^{*}_{k}\big)^{\top}e^{k}_{\mathrm{ave}}(t)\big]}_{\mathrm{(iv)}}\notag\\
&+\underbrace{2\eta_{c,t}\mathbb{E}_{\mathrm{init}}[g^{k}_{\mathrm{ave}}(t)^{\top}e^{k}_{\mathrm{ave}}(t)]}_{\mathrm{(v)}}.
\end{align}
\par
For the $\mathrm{(i)}$-th term %$\eta_{c}^{2}\mathbb{E}\big[\|g_{\mathrm{ave}}(k)\|^{2}_{2}\big]$
in (\ref{expectationinequalityofStochasticDescentLemma}), according to the definition of $g^{k}_{\mathrm{ave}}(t)$ in (\ref{thedefinitionofstochastic-2}), we have
\begin{align}
&\eta^{2}_{c,t}\mathbb{E}_{\mathrm{init}}[\|g^{k}_{\mathrm{ave}}(t)\|^{2}_{2}]\notag\\
\leq&2\eta^{2}_{c,t}\mathbb{E}_{\mathrm{init}}[\|g^{k}_{\mathrm{ave}}(t)-\bar{g}^{k}_{\mathrm{ave}}(t)\|^{2}_{2}]+2\eta^{2}_{c,t}\mathbb{E}_{\mathrm{init}}[\|\bar{g}^{k}_{\mathrm{ave}}(t)\|^{2}_{2}]\notag\\
\leq&\frac{2\eta^{2}_{c,t}}{N}\sum_{i=1}^{N}\mathbb{E}_{\mathrm{init}}[\|g^{k}_{i}(t)-\bar{g}^{k}_{i}(t)\|^{2}_{2}]+2\eta^{2}_{c,t}\mathbb{E}_{\mathrm{init}}[\|\bar{g}^{k}_{\mathrm{ave}}(t)\|^{2}_{2}]\notag\\
\leq&24\eta^{2}_{c,t}\mathbb{E}_{\mathrm{init},\varsigma_{k}}\Big[\Big(\widehat{Q}_{0}\big(\bm{z};\overline{W}(t)\big)-\widehat{Q}_{0}(\bm{z};W^{*}_{k})\Big)^{2}\Big]\notag\\
&+\eta^{2}_{c,t}\mathcal{O}\big(B^{3}(mN)^{1-2p}\big),\label{newconclusioninequality}
\end{align}
where the last inequality uses the (i) and (iii) of Lemma~\ref{Lemma4-6}.
\par
For the $\mathrm{(ii)}$-th term in (\ref{expectationinequalityofStochasticDescentLemma}), we use the properties of projection operation and have
\begin{align}\notag
\|e^{k}_{\mathrm{ave}}(t)\|^{2}_{2}\leq\frac{1}{N}\sum_{i=1}^{N}\|e^{k}_{i}(t)\|^{2}_{2}\leq\frac{1}{N}\sum_{i=1}^{N}\eta_{c,t}^{2}\|g^{k}_{i}(t)\|^{2}_{2}.
\end{align}
By using (\ref{theboundofg}), we further have
\begin{align}\label{expectationresultofiithinequalityofStochasticDescentLemma}
&\mathbb{E}_{\mathrm{init}}[\|e^{k}_{\mathrm{ave}}(t)\|^{2}_{2}]\notag\\
\leq&3\eta_{c,t}^{2}R^{2}_{0}(mN)^{1-2p}+\eta_{c,t}^{2}(12d_{1}+12B^{2})(mN)^{2-4p}\notag\\
=&\eta^{2}_{c,t}\mathcal{O}\big(B^{2}(mN)^{1-2p}\big).
\end{align}
\par
For the $\mathrm{(iii)}$-th term in (\ref{expectationinequalityofStochasticDescentLemma}),
%since $\bm{W}^{*}$ is an approximate stationary point of (\ref{theMSBEofmathbfW}),
we have
\begin{align}
&-2\eta_{c,t}\big(\overline{W}(t)-W^{*}_{k}\big)^{\top}g^{k}_{\mathrm{ave}}(t)\notag\\
\leq&-2\eta_{c,t}\big(\overline{W}(t)-W^{*}_{k}\big)^{\top}\big(g^{k}_{\mathrm{ave}}(t)-\bar{g}^{k*}_{0,\mathrm{ave}}\big)\label{eq:explain}\\
=&-2\eta_{c,t}\big(\overline{W}(t)-W^{*}_{k}\big)^{\top}\big(g^{k}_{\mathrm{ave}}(t)-\bar{g}^{k}_{\mathrm{ave}}(t)+\bar{g}^{k}_{\mathrm{ave}}(t)\notag\\
&-\bar{g}^{k}_{0,\mathrm{ave}}(t)\big)-2\eta_{c,t}\big(\overline{W}(t)-W^{*}_{k}\big)^{\top}\big(\bar{g}^{k}_{0,\mathrm{ave}}(t)-\bar{g}^{k*}_{0,\mathrm{ave}}\big), \label{iiithterminequationofPopulationDescentLemma}
\end{align}
where the first inequality follows from (\ref{theinequalityofstationarypoint}).
Since $\overline{W}(t),W^{*}_{k}\in S^{W}_{B}$, we have
\begin{align}\label{1thiiithterminequationofPopulationDescentLemma}
&\big(\overline{W}(t)-W^{*}_{k}\big)^{\top}\big(g^{k}_{\mathrm{ave}}(t)-\bar{g}^{k}_{\mathrm{ave}}(t)+\bar{g}^{k}_{\mathrm{ave}}(t)-\bar{g}^{k}_{0,\mathrm{ave}}(t)\big)\notag\\
\geq&-2B\big(\|g^{k}_{\mathrm{ave}}(t)-\bar{g}^{k}_{\mathrm{ave}}(t)\|_{2}+\|\bar{g}^{k}_{\mathrm{ave}}(t)-\bar{g}^{k}_{0,\mathrm{ave}}(t)\|_{2}\big).
\end{align}
%Note that $\big(\overline{\bm{W}}(k)-\bm{W}^{*}\big)^{\top}\big(\bar{g}_{0,\mathrm{ave}}(k)-\bar{g}^{*}_{0}\big)$ in (\ref{iiithterminequationofPopulationDescentLemma}) obeys
According to the definitions of $\bar{g}^{k}_{0,\mathrm{ave}}(t)$ in (\ref{thedefinitionofstochastic2-3}) and $\bar{g}^{k*}_{0,\mathrm{ave}}$ in (\ref{thedefinitionofstochastic2-1}), we have
\begin{align}\label{2thiiithterminequationofPopulationDescentLemma}
&\big(\overline{W}(t)-W^{*}_{k}\big)^{\top}\big(\bar{g}^{k}_{0,\mathrm{ave}}(t)-\bar{g}^{k*}_{0,\mathrm{ave}}\big)\notag\\
=&\mathbb{E}_{\varsigma^{+}_{k}}\Big[\Big(\delta^{k}_{0,\mathrm{ave}}(t)-\delta^{k*}_{0,\mathrm{ave}}\Big)\cdot\notag\\
&\Big(\nabla_{W'}\widehat{Q}_{0}\big(\bm{z};W'(0)\big)^{\top}\big(\overline{W}(t)-W^{*}_{k}\big)\Big)\Big]\notag\\
=&\mathbb{E}_{\varsigma^{+}_{k}}\Big[\Big(\widehat{Q}_{0}\big(\bm{z},\overline{W}(t)\big)-\widehat{Q}_{0}(\bm{z};W^{*}_{k})\Big)^{2}\notag\\
&-\gamma\Big(\widehat{Q}_{0}\big(\bm{z}',\overline{W}(t)\big)-\widehat{Q}_{0}(\bm{z}';W^{*}_{k})\Big)\cdot\notag\\
&\Big(\widehat{Q}_{0}\big(\bm{z},\overline{W}(t)\big)-\widehat{Q}_{0}(\bm{z};W^{*}_{k})\Big)\Big]\notag\\
\geq&(1-\gamma)\mathbb{E}_{\varsigma_{k}}\Big[\Big(\widehat{Q}_{0}\big(\bm{z},\overline{W}(t)\big)-\widehat{Q}_{0}(\bm{z};W^{*}_{k})\Big)^{2}\Big],
\end{align}
where the last inequality follows from the fact that $\bm{z}$ and $\bm{z}'$ have the same stationary station-action distribution over $\varsigma_{k}$.
Substituting (\ref{1thiiithterminequationofPopulationDescentLemma}) and (\ref{2thiiithterminequationofPopulationDescentLemma}) into
(\ref{iiithterminequationofPopulationDescentLemma}), we have
\begin{align}\label{resultofiiithterminequationofPopulationDescentLemma}
&-2\eta_{c,t}\big(\overline{W}(t)-W^{*}_{k}\big)^{\top}g^{k}_{\mathrm{ave}}(t)\notag\\
\leq&-2(1-\gamma)\eta_{c,t}\mathbb{E}_{\varsigma_{k}}\Big[\Big(\widehat{Q}_{0}(\bm{z},\overline{W}(t)\big)-\widehat{Q}_{0}(\bm{z};W^{*}_{k})\Big)^{2}\Big]\notag\\
&+4\eta_{c,t} B\big(\|g^{k}_{\mathrm{ave}}(t)-\bar{g}^{k}_{\mathrm{ave}}(t)\|_{2}+\|\bar{g}^{k}_{\mathrm{ave}}(t)-\bar{g}^{k}_{0,\mathrm{ave}}(t)\|_{2}\big).
\end{align}
Taking expectation on both sides of (\ref{resultofiiithterminequationofPopulationDescentLemma}), we have
\begin{align}\label{expectationresultofiiithterminequationofPopulationDescentLemma}
&-2\eta_{c,t}\mathbb{E}_{\mathrm{init}}\big[\big(\overline{W}(t)-W^{*}_{k}\big)^{\top}g^{k}_{\mathrm{ave}}(t)\big]\notag\\
\leq&-2(1-\gamma)\eta_{c,t}\mathbb{E}_{\mathrm{init},\varsigma_{k}}\Big[\Big(\widehat{Q}_{0}(\bm{z},\overline{W}(t)\big)-\widehat{Q}_{0}(\bm{z};W^{*}_{k})\Big)^{2}\Big]\notag\\
&+4\eta_{c,t} B\mathbb{E}_{\mathrm{init}}[\|g^{k}_{\mathrm{ave}}(t)-\bar{g}^{k}_{\mathrm{ave}}(t)\|_{2}]\notag\\
&+4\eta_{c,t} B\mathbb{E}_{\mathrm{init}}[\|\bar{g}^{k}_{\mathrm{ave}}(t)-\bar{g}^{k}_{0,\mathrm{ave}}(t)\|_{2}]\notag\\
\leq&-2\eta_{c,t}(1-\gamma)\mathbb{E}_{\mathrm{init},\varsigma_{k}}\Big[\Big(\widehat{Q}_{0}\big(\bm{z},\overline{W}(t)\big)-\widehat{Q}_{0}(\bm{z};W^{*}_{k})\Big)^{2}\Big]\notag\\
&+\eta_{c,t}\mathcal{O}\big(B^{\frac{5}{2}}(mN)^{\frac{1}{2}-p}\big),
\end{align}
where the last inequality follows from (i) and (ii) of Lemma~\ref{Lemma4-6}.
\par
For the $\mathrm{(iv)}$-th term
%$-2\mathbb{E}\big[\big(\overline{\bm{W}}(k)-\bm{W}^{*}\big)^{\top}e_{\mathrm{ave}}(k)\big]$
in (\ref{expectationinequalityofStochasticDescentLemma}), we use (\ref{expectationresultofiithinequalityofStochasticDescentLemma}) and have
\begin{align}\label{expectationresultofivthinequalityofStochasticDescentLemma}
-2\mathbb{E}_{\mathrm{init}}\big[\big(\overline{W}(t)-W^{*}_{k}\big)^{\top}e^{k}_{\mathrm{ave}}(t)\big]\leq&4B\mathbb{E}_{\mathrm{init}}[\|e^{k}_{\mathrm{ave}}(t)\|_{2}]\notag\\
=&\eta_{c,t}\mathcal{O}\big(B^{2}(mN)^{\frac{1}{2}-p}\big).
\end{align}
\par
For the $\mathrm{(v)}$-th term in (\ref{expectationinequalityofStochasticDescentLemma}), we use Cauchy-Schwartz inequality and have
\begin{align}\label{expectationresultofvthinequalityofStochasticDescentLemma}
&2\eta_{c,t}\mathbb{E}_{\mathrm{init}}[g^{k}_{\mathrm{ave}}(t)^{\top}e^{k}_{\mathrm{ave}}(t)]\notag\\
\leq&2\eta_{c,t}\mathbb{E}_{\mathrm{init}}[\|g^{k}_{\mathrm{ave}}(t)\|_{2}\|e^{k}_{\mathrm{ave}}(t)\|_{2}]\notag\\
\leq&2\eta_{c,t}\mathbb{E}_{\mathrm{init}}[\|g^{k}_{\mathrm{ave}}(t)\|^{2}_{2}]^{\frac{1}{2}}\mathbb{E}_{\mathrm{init}}[\|e^{k}_{\mathrm{ave}}(t)\|^{2}_{2}]^{\frac{1}{2}}\notag\\
=&\eta^{2}_{c,t}\mathcal{O}\big(B^{2}(mN)^{1-2p}\big),
\end{align}
where the last inequality follows from (\ref{theboundofg}) and (\ref{theprojectionpropertyofe}).
Substituting (\ref{newconclusioninequality}), (\ref{expectationresultofiithinequalityofStochasticDescentLemma}), (\ref{expectationresultofiiithterminequationofPopulationDescentLemma}), (\ref{expectationresultofivthinequalityofStochasticDescentLemma}), and (\ref{expectationresultofvthinequalityofStochasticDescentLemma}) into (\ref{expectationinequalityofStochasticDescentLemma}), we can complete the proof of Lemma~\ref{StochasticDescentLemma}.
\subsection{Proof of Lemma~\ref{thefirstlemmaoftheequalityofLipschitzbefore}}\label{theproofofthefirstlemmaoftheequalityofLipschitzbefore}
Recall that $\xi_{i,k}=\widehat{\nabla}_{\theta_{i}}J(\bm{\pi}_{\bm{\theta}(k)})-\mathbb{E}_{\sigma_{k}}[\widehat{\nabla}_{\theta_{i}}J(\bm{\pi}_{\bm{\theta}(k)})]$,
we have
\begin{align}\label{thei-thinequalityofthelemmaoferrorbetweennablaandwidenabla}
&\mathbb{E}[\|\nabla_{\theta_{i}}J(\bm{\pi}_{\bm{\theta}(k)})-\widehat{\nabla}_{\theta_{i}}J(\bm{\pi}_{\bm{\theta}(k)})\|^{2}_{2}]\notag\\
\leq&2\mathbb{E}[\|\xi_{i,k}\|^{2}_{2}]+2\mathbb{E}[\|\nabla_{\theta_{i}}J(\bm{\pi}_{\bm{\theta}(k)})-\mathbb{E}_{\sigma_{k}}[\widehat{\nabla}_{\theta_{i}}J(\bm{\pi}_{\bm{\theta}(k)})]\|^{2}_{2}].
\end{align}
Based on (\ref{theresultofpolicygradientforagenti}) in Proposition~\ref{ThepolicygradienttheoremforMARL}, we have
\begin{align}\label{theii-thinequalityofthelemmaoferrorbetweennablaandwidenabla}
&\|\nabla_{\theta_{i}}J(\bm{\pi}_{\bm{\theta}(k)})-\mathbb{E}_{\sigma_{k}}[\widehat{\nabla}_{\theta_{i}}J(\bm{\pi}_{\bm{\theta}(k)})]\|_{2}\notag\\
=&\big\|\mathbb{E}_{\sigma_{k}}\big[\overline{\psi}_{i}\big(\bm{s},a_{i};\theta_{i}(k)\big)\big(Q^{\bm{\pi}_{\bm{\theta}(k)}}(\bm{z})-\widehat{Q}^{\bm{\pi}_{\bm{\theta}(k)}}_{i,\mathrm{out}}(\bm{z})\big)\big]\big\|_{2}\notag\\
\leq&\mathbb{E}_{\sigma_{k}}\big[\big\|\overline{\psi}_{i}\big(\bm{s},a_{i};\theta_{i}(k)\big)\big\|_{2}\big|Q^{\bm{\pi}_{\bm{\theta}(k)}}(\bm{z})-\widehat{Q}^{\bm{\pi}_{\bm{\theta}(k)}}_{i,\mathrm{out}}(\bm{z})\big|\big]\notag\\
\leq&2(mN)^{\frac{1}{2}-p}\mathbb{E}_{\sigma_{k}}[|Q^{\bm{\pi}_{\bm{\theta}(k)}}(\bm{z})-\widehat{Q}^{\bm{\pi}_{\bm{\theta}(k)}}_{i,\mathrm{out}}(\bm{z})|]\notag\\
\leq&2\kappa_{a}(mN)^{\frac{1}{2}-p}\|Q^{\bm{\pi}_{\bm{\theta}(k)}}(\bm{z})-\widehat{Q}^{\bm{\pi}_{\bm{\theta}(t)}}_{i,\mathrm{out}}(\bm{z})\|_{\varsigma_{k}},
\end{align}
where the first inequality comes from the Jensen's inequality, the second inequality holds from the fact that $\|\psi_{i,r}(\bm{s},a_{i};\theta_{i})\|_{2}\leq(mN)^{-p}$ in (\ref{thenormofpsi}), and the last inequality achieves by Assumption~\ref{RegularityConditiononsigmaandvarsigma}.
Substituting (\ref{theii-thinequalityofthelemmaoferrorbetweennablaandwidenabla}) into (\ref{thei-thinequalityofthelemmaoferrorbetweennablaandwidenabla}), we have
\begin{align}
&\mathbb{E}[\|\nabla_{\theta_{i}}J(\bm{\pi}_{\bm{\theta}(k)})-\widehat{\nabla}_{\theta_{i}}J(\bm{\pi}_{\bm{\theta}(k)})\|^{2}_{2}]\notag\\
\leq&2\mathbb{E}[\|\xi_{i,k}\|^{2}_{2}]+8\kappa^{2}_{a}(mN)^{1-2p}\mathbb{E}[\|Q^{\bm{\pi}_{\bm{\theta}(k)}}(\bm{z})-\widehat{Q}^{\bm{\pi}_{\bm{\theta}(k)}}_{i,\mathrm{out}}(\bm{z})\|_{\varsigma_{k}}^{2}],\notag
\end{align}
which completes the proof of Lemma~\ref{thefirstlemmaoftheequalityofLipschitzbefore}.


\begin{thebibliography}{1}

% Introduction of RL+MARL
\bibitem{Sutton1998}
R. S. Sutton and A. G. Barto, \emph{Reinforcement Learning: An Introduction}. Cambridge, MA, USA: MIT Press, 1998.
% smart grid
\bibitem{Dai2020TII}
P. Dai, W. Yu, G. Wen, and S. Baldi, ``Distributed reinforcement learning algorithm for dynamic economic dispatch with unknown generation cost functions,'' \emph{IEEE Trans. Ind. Informat.}, vol. 16, no. 4, pp. 2258-2267, Apr. 2020.
\bibitem{Dai2021TCYB}
P. Dai, W. Yu, and D. Chen, ``Distributed Q-learning algorithm
for dynamic resource allocation with unknown objective functions
and application to microgrid,'' \emph{IEEE Trans. Cybern.}, \emph{IEEE Trans. Cybern.}, vol. 52, no. 11, pp. 12340-12350, Nov. 2022.
% Intelligent Transportation
\bibitem{Chu2020TITS}
T. Chu, J. Wang, L. Codec\`a, and Z. Li, ``Multi-agent deep reinforcement learning for large-scale traffic signal control,'' \emph{IEEE Trans. Intell. Transp. Syst.}, vol. 21, no. 3, pp. 1086-1095, Mar. 2020.

\bibitem{Wang2021TCYB}
X. Wang, L. Ke, Z. Qiao, and X. Chai, ``Large-Scale traffic signal control using a novel multiagent reinforcement learning,'' \emph{IEEE Trans. Cybern.}, vol. 51, no. 1, pp. 174-187, Jan. 2021.
% CPS
\bibitem{Ding2017Automatica}
K. Ding, Y. Li, D. E. Quevedo, S. Dey, and L. Shi, ``A multi-channel transmission schedule for remote state estimation under DoS attacks,'' \emph{Automatica}, vol. 78, pp. 194-201, Apr. 2017.
\bibitem{Dai2020TNSE}
P. Dai, W. Yu, H. Wang, G. Wen, and Y. Lv, ``Distributed reinforcement learning for cyber-physical system with multiple remote state estimation
under DoS attacker,'' \emph{IEEE Trans. Netw. Sci. Eng.}, vol. 7, no. 4, pp. 3212-3222, Oct. 2020.
% wireless communication
\bibitem{Nasir2019AreasCommun}
Y. S. Nasir and D. Guo, ``Multi-agent deep reinforcement learning for dynamic power allocation in wireless networks,'' \emph{IEEE J. Sel. Areas Commun.}, vol. 37, no. 10, pp. 2239-2250, Oct. 2019.

\bibitem{Zhao2019TWC}
N. Zhao, Y.-C. Liang, D. Niyato, Y. Pei, M. Wu, and Y. Jiang, ``Deep reinforcement learning for user association and resource allocation in heterogeneous cellular networks,'' \emph{IEEE Trans. Wireless Commun.}, vol. 18, no. 11, pp. 5141-5152, Nov. 2019.

% centralized method
\bibitem{Wei2018ICKDD}
H. Wei, G. Zheng, H. Yao, and Z. Li, ``Intellilight: A reinforcement learning approach for intelligent traffic light control,'' in \emph{Proc. ACM SIGKDD Int. Conf. Knowl. Disc. Data Min.}, pp. 2496-2505, 2018.


% IQL
\bibitem{Tan1993}
M. Tan, ``Multi-agent reinforcement learning: Independent vs. cooperative agents,''
in \emph{Proc. Int. Conf. Mach. Learn.}, pp. 330-337, 1993.

% StarCraft game--CTDE
\bibitem{QMIX}
T. Rashid, M. Samvelyan, C. S. De Witt, G. Farquhar, J. Foerster, and S. Whiteson, ``Qmix: Monotonic value function factorisation for deep multi-agent reinforcement learning,''
in \emph{Proc. Int. Conf. Mach. Learn.}, pp. 6846-6859, 2018.

\bibitem{QPLEX}
J. Wang, Z. Ren, T. Liu, Y. Yu, and C. Zhang, ``QPLEX: Duplex dueling multi-agent Q-Learning,''
in \emph{Proc. Int. Conf. Learn. Represent.}, 2021.
% CCDA
\bibitem{COMA}
J. N. Foerster, G. Farquhar, T. Afoura, N. Nardelli, and S. Whiteson, ``Counterfactual multi-agent policy gradients,'' in \emph{Proc. AAAI Conf. Artif. Intell.}, pp. 2974-2982, 2018.

\bibitem{DOP}
Y. Wang, B. Han, T. N. Wang, H. Dong, and C. Zhang, ``Off-policy
multi-agent decomposed policy gradients,'' \emph{arXiv:2007.12322}, 2020.



% distributed AC
\bibitem{Zhang2018ICML}
K. Zhang, Z. Yang, H. Liu, T. Zhang, and T. Ba\c sar, ``Fully decentralized multi-agent reinforcement learning with networked agents,'' in \emph{Proc. Int. Conf. Mach. Learn.}, pp. 5872-5881, 2018.

\bibitem{Dai2022TNNLS}
P. Dai, W. Yu, H. Wang, and S. Baldi, ``Distributed actor-critic algorithms for multiagent reinforcement learning over directed graphs,'' \emph{IEEE Trans. Neural Netw. Learn. Syst.}, vol. 34, no. 10, pp: 7210-7221, Oct. 2023.


\bibitem{Wang2019}
L. Wang, Q. Cai, Z. Yang, and Z. Wang, ``Neural policy gradient
methods: Global optimality and rates of convergence,"
in \emph{Proc. Int. Conf. Learn. Represent.}, 2021.

\bibitem{Doan2019ICML}
T. Doan, S. Maguluri, and J. Romberg.
``Finite-time analysis of distributed TD(0) with linear function approximation on multi-agent reinforcement learning,''
in \emph{Proc. Int. Conf. Mach. Learn.},
pp. 1626-1635, 2019.

\bibitem{Lin2019CDC}
Y. Lin, K. Zhang, Z. Yang, Z. Wang, T. Ba\c{s}ar, R. Sandhu, and J. Liu,
``A communication-efficient multi-agent actor-critic algorithm for distributed reinforcement learning,''
in~\emph{Proc. IEEE Conf. Decis. Control}, pp. 5562-5567, 2019.

\bibitem{Suttle2020IFAC}
W. Suttle, Z. Yang, K. Zhang, Z. Wang, T. Ba\c{s}ar, and J. Liu, ``A multiagent off-policy actor-critic algorithm for distributed reinforcement learning,'' \emph{IFAC-PapersOnLine},
vol. 53, no. 2, pp. 1549-1554, 2020.


\bibitem{Sutton2000}
R. S. Sutton, D. A. McAllester, S. P. Singh, and Y. Mansour, ``Policy gradient methods for reinforcement learning with function approximation,''
in~\emph{Proc. Adv. Neural Inf. Process. Syst.}, vol. 12, 2000.

\bibitem{Cai2019}
Q. Cai, Z. Yang, J. D. Lee, and Z. Wang, ``Neural temporal-difference learning converges to global optima,''
in~\emph{Proc. Adv. Neural Inf. Process. Syst.}, vol 32, 2019.

% natural policy gradient
\bibitem{Peters2008Neurocomputing}
J. Peters and S. Schaal, ``Natural actor-critic,''
\emph{Neurocomputing}, vol. 71, pp. 1180-1190, 2008.

\bibitem{Wagner2013NIPS}
P. Wagner, ``Optimistic policy iteration and natural actor-critic: A unifying view and a non-optimality result,''
in~\emph{Proc. Adv. Neural Inf. Process. Syst.}, vol 26, 2013.

\bibitem{Xu2020NIPS}
T. Xu, Z. Wang, and Y. Liang,
``Improving sample complexity bounds for (natural) actor-critic algorithms,''
in~\emph{Proc. Adv. Neural Inf. Process. Syst.}, vol 33, 2020.



%%% communication
\bibitem{nedic2010}
A. Nedi\'{c}, A. Ozdaglar and A. P. Parrilo,
``Constrained consensus and optimization in multi-agent networks,'' \emph{IEEE Trans. Autom. Control},
vol. 55, no. 4, pp. 922-938, Apr. 2010.


%
%\bibitem{Doan2019ICML}
%T. Doan, S. Maguluri, and J. Romberg,
%``Finite-time analysis of distributed TD(0) with linear function approximation on multi-agent reinforcement learning,''
%in \emph{Proc. Int. Conf. Mach. Learn.}, 2019, pp. 1626--1635.

%\bibitem{AllenNIPS2019}
%Z. Allen-Zhu, Y. Li, and Y. Liang, ``Learning and generalization in overparameterized neural networks, going beyond two layers,'' in \emph{Proc. Adv. Neural Inf. Process.
%Syst.}, 2019, pp. 6158--6169.
%
%\bibitem{AllenICML2019}
%Z. Allen-Zhu, Y. Li, and Z. Song, ``A convergence theory for deep learning via over-parameterization,''
%in \emph{Proc. Int. Conf. Mach. Learn.}, 2019, pp. 242--252.


%\bibitem{Rahimi2008}
%A. Rahimi and B. Recht, ``Weighted sums of random kitchen sinks:
%Replacing minimization with randomization in learning," \emph{Proc.
%Adv. Neural Inf. Process. Syst.}, pp. 1313-1320, 2008.

\bibitem{Nesterov2018}
Y. Nesterov,
\emph{Introductory Lectures on Convex Optimization},
Berlin, Germany: Springer, 2018.

\bibitem{Zhou2023}
Z. Zhou, Z. Chen, Y. Lin, and A. Wierman,
``Convergence rates for localized actor-critic in networked markov potential games,''
in \emph{Proc. Uncertainty Artif. Intell. Conf.}, pp. 2563-2573, 2023.



\bibitem{MeiICML2020}
J. Mei, C. Xiao, C. Szepesv\'{a}ri, and D. Schuurmans,
``On the global convergence rates of softmax policy gradient methods,'' in \emph{Proc. Int. Conf. Mach. Learn.}, pp. 6820-6829, 2020.

\bibitem{nedic2018}
A. Nedi\'{c}, A. Olshevsky, and M. G. Rabbat,
``Network topology and communication-computation tradeoffs in
decentralized optimization,'' \emph{Proc. IEEE}, vol. 106, no. 5, pp. 953-976, May 2018.


%\bibitem{Doan2019ICML}
%T. Doan, S. Maguluri, and J. Romberg,
%``Finite-time analysis of distributed TD(0) with linear function approximation on multi-agent reinforcement learning,''
%in \emph{Proc. Int. Conf. Mach. Learn.}, pp. 1626-1635, 2019.

\end{thebibliography}
\end{document}